\documentclass[11pt]{article}

\usepackage{times}
\usepackage{amsmath,amssymb,amsthm}
\usepackage{color,soul}
\usepackage{xspace}
\usepackage[numbers]{natbib}
\usepackage{comment}
\usepackage{enumerate}
\usepackage[colorlinks=true,citecolor=darkgreen,linkcolor=euro-blue,pdfencoding=auto,psdextra,pagebackref=true]{hyperref}
\usepackage{etex}
\usepackage{mathtools}

\usepackage{nicefrac}
\usepackage{paralist}
\usepackage{todonotes}
\usepackage{setspace}
\usepackage[matrix,arrow,ps]{xy}

\usepackage{subcaption}
\usepackage[noend,ruled,linesnumbered]{algorithm2e} 
\usepackage{xifthen}

\usepackage{booktabs}

\usepackage{tikz}
\usetikzlibrary{decorations,arrows,petri,topaths,backgrounds,shapes,positioning,fit,calc,decorations.pathreplacing,patterns,intersections,decorations.pathmorphing,matrix}
\def \xsc {1}
\def \ysc {.9}
\def \dist {1}

\makeatletter
\newcommand{\gettikzxy}[3]{%
  \tikz@scan@one@point\pgfutil@firstofone#1\relax
  \edef#2{\the\pgf@x}%
  \edef#3{\the\pgf@y}%
}
\makeatother

\newcommand{\searchprob}[3]{
  \begin{center}%
    \begin{minipage}{0.9\linewidth}%
      \textsc{#1}\\[0.2ex]
      \textbf{Input:} #2\\[0.2ex]
      \textbf{Output:} #3
    \end{minipage}%
  \end{center} }

\usepackage[sort&compress,nameinlink]{cleveref}

\newtheorem{theorem}{Theorem}[section]
\newtheorem{lemma}{Lemma}[section]
\newtheorem{corollary}[theorem]{Corollary}
\newtheorem{proposition}[theorem]{Proposition}

\newtheorem{claim}{Claim}

\theoremstyle{remark}
\newtheorem{example}{Example}[section]
\newtheorem{definition}{Definition}[section]

\crefname{table}{Table}{Tables}
\crefname{figure}{Figure}{Figures}
\crefname{theorem}{Theorem}{Theorems}
\crefname{claim}{Claim}{Claims}
\crefname{definition}{Definition}{Definitions}
\crefname{corollary}{Corollary}{Corollaries}
\crefname{observation}{Observation}{Observations}
\crefname{lemma}{Lemma}{Lemmas}
\crefname{example}{Example}{Examples}
\crefname{reduction}{Reduction}{Reductions}
\crefname{construction}{Construction}{Constructions}
\crefname{subsection}{Subsection}{Subsections}
\crefname{section}{Section}{Sections}
\crefname{proposition}{Proposition}{Propositions}
\crefname{algorithm}{Algorithm}{Algorithms}
\crefname{algocf}{Algorithm}{Algorithms}
\crefname{rrule}{Reduction rule}{Reduction rules}
\crefalias{AlgoLine}{line}

\newcommand{\Pot}{\mathcal{P}}

\newcommand{\pref}{\ensuremath{\succ}}

\usepackage{xcolor}
\usepackage{colortbl}
\definecolor{dargray}{rgb}{0.18, 0.18, 0.18}
\definecolor{darkgreen}{rgb}{0.01,0.6,0.1}
\definecolor{lightrose}{rgb}{0.996,0.75,0.793}
\definecolor{rose}{cmyk}{0.75, 0.75, 0,0}
\definecolor{winered}{rgb}{0.6,0.1,0.1}
\definecolor{lightyellow}{rgb}{1, 1, 0.6}
\definecolor{transparent}{rgb}{1,1,1}
\definecolor{lightlightgray}{rgb}{0.88, 0.88, 0.88}
\definecolor{lightgray}{rgb}{0.8, 0.8, 0.8}
\definecolor{lightblue}{rgb}{0.527,0.805,0.977}
\definecolor{lightgreen}{rgb}{.74,1,0}
\definecolor{cambridgeblue}{rgb}{0.87, 1.0, 0.87}
\definecolor{mistyrose}{rgb}{1.0, 0.89, 0.88}
\definecolor{euro-blue}{rgb}{0.0,0,0.5}
\newcommand{\cellnph}[1]{\cellcolor{mistyrose}#1}
\newcommand{\cellnphh}[1]{\cellcolor{mistyrose}#1}
\newcommand{\cellp}[1]{\cellcolor{cambridgeblue}#1}
\newcommand{\celllogsnp}[1]{\cellcolor{orange!25}#1}

\newcommand{\cellwh}[1]{\cellcolor{rose!30}#1}
\newcommand{\cellgh}[1]{\cellcolor{blue!15}#1}

\newcommand{\pairset}{\ensuremath{U \star W}}

\newcommand{\gstable}[1]{#1-layer globally stable\xspace}
\newcommand{\abstable}[1]{#1-layer pair stable\xspace}
\newcommand{\aistable}[1]{#1-layer individually stable\xspace}

\newcommand{\gstability}[1]{#1-layer global stability\xspace}
\newcommand{\abstability}[1]{#1-layer pair stability\xspace}
\newcommand{\aistability}[1]{#1-layer individual stability\xspace}

\newcommand{\GSM}{\textsc{Globally Stable Marriage}\xspace}
\newcommand{\BSM}{\textsc{Pair Stable Marriage}\xspace}
\newcommand{\ISM}{\textsc{Individually Stable Marriage}\xspace}

\tikzstyle{blueline} = [thick, blue, dotted]
\tikzstyle{redline} = [thick, red, dashed]
\tikzstyle{greenline} = [thick, darkgreen, dashed]
\tikzstyle{blackline} = [thick, black]
\tikzstyle{concept} = [minimum height=4ex, inner sep=1pt, text centered, align=center]

\newcommand{\myemph}[1]{{\color{winered}\emph{#1}}}

\newcommand{\true}{\mathsf{true}}
\newcommand{\false}{\mathsf{false}}

\newcommand{\SM}{\textsc{Stable Marriage}\xspace}
\newcommand{\SR}{\textsc{Stable Roommates}\xspace}
\newcommand{\IS}{\textsc{Independent Set}\xspace}

\makeatletter
\let\cref@old@stepcounter\stepcounter
\def\stepcounter#1{%
  \cref@old@stepcounter{#1}%
  \cref@constructprefix{#1}{\cref@result}%
  \@ifundefined{cref@#1@alias}%
    {\def\@tempa{#1}}%
    {\def\@tempa{\csname cref@#1@alias\endcsname}}%
  \protected@edef\cref@currentlabel{%
    [\@tempa][\arabic{#1}][\cref@result]%
    \csname p@#1\endcsname\csname the#1\endcsname}}
\makeatother

\begin{document}
\sloppy
\allowdisplaybreaks

%%%%%%%%%%%%%%%%%%%%%%%%%%%%%%%%%%%%%%%%%%%%%%%%%%%%%%%%%%%%%%%%%%%%%%%%%%
%%%%%%%%%%%%%%%%%%%%%%%%%%%%%%%%%%%%%%%%%%%%%%%%%%%%%%%%%%%%%%%%%%%%%%%%%%
%\title{Stable Marriages for Multi-Layer Peferences}
%\title{Stable Marriage with Multi-Modal Preference Lists} %\todo[inline]{Hua: This title sounds like stable marriage has several modes, which we do not consider?!
%There is a paper on multimode control attacks which assume that the attack could have multiple forms, e.g. voter deletion combined with addition, etc.}}
\title{Stable Marriage with Multi-Modal Preferences\thanks{Work started when 
all authors were with TU~Berlin.}}

\author{
  \makebox[0.25\linewidth]{Jiehua Chen\thanks{Supported by the People Programme (Marie Curie Actions) of the European Union's Seventh Framework Programme (FP7/2007-2013) under REA grant agreement number~631163.11, 
and by the
Israel Science Foundation (grant number 551145/14).}} \\ 
  Ben-Gurion University of the Negev\\ 
  Beer-Sheva, Israel 
\and
  \makebox[0.25\linewidth]{Rolf Niedermeier} \\
  TU Berlin \\
  Berlin, Germany
\and 
  \makebox[0.25\linewidth]{Piotr Skowron\thanks{Supported by a postdoctoral fellowship of the Alexander von Humboldt Foundation, Bonn, Germany.}} \\
  TU Berlin \\
  Berlin, Germany
}

\date{}
\maketitle

\begin{abstract}
We introduce a generalized version of the famous \SM problem,
now based on multi-modal preference lists. The central twist herein
is to allow each agent to rank its potentially matching counterparts
based on more than one ``evaluation mode'' (e.g., more than one criterion); thus, each agent is equipped 
with multiple preference lists, each 
ranking the counterparts in a possibly
different way. 
We introduce and study three natural concepts of stability, investigate their mutual relations 
and focus on computational complexity aspects with 
respect to computing stable matchings in these new scenarios. 
Mostly encountering computational hardness (NP-hardness), we can also 
spot few islands of tractability and make a surprising 
connection to the \textsc{Graph Isomorphism} problem.

\end{abstract}

%%%%%%%%%%%%%%%%%%%%%%%%%%%%%%%%%%%%%%%%%%%%%%%%%%%%%%%%%%%%%%%%%%%%%%%%%
%%%%%%%%%%%%%%%%%%%%%%%%%%%%%%%%%%%%%%%%%%%%%%%%%%%%%%%%%%%%%%%%%%%%%%%%
\section{Introduction}\label{sec:intro}
%%%%%%%%%%%%%%%%%%%%%%%%%%%%%%%%%%%%%%%%%%%%%%%%%%%%%%%%%%%%%%%%%%%%%%%%%

Information about the same ``phenomenon'' can come from different, 
possibly ``contradicting'', sources. 
For instance, when evaluating candidates for an open position,
data concerning experience and so far achieved successes 
of the candidates may give different
candidate rankings than data 
concerning their formal qualifications and degrees.
In other words, one has to deal with a multi-modal data scenario.
Clearly, in maximally objective and 
rationality-driven decision making, it makes sense 
to take into account several information resources in order 
to achieve best possible results. In this work we systematically apply this point of view to the 
\SM problem~\cite{GaleShapley1962}; a key observation here is that several natural and well-motivated
``multi-modal variants'' of \SM need to be studied.
We investigate the complexity of computing matchings that are stable according to the considered definitions. 

In the classic (conservative) \SM problem~\cite{GaleShapley1962},
we are given two disjoint sets~$U$ and~$W$ of $n$~agents each,
where each of the agents has a \myemph{strict preference list} that ranks 
\myemph{every} member of the other set.
The goal is to find a bijection (which we call a \myemph{matching})
between $U$ and~$W$ without any \myemph{blocking pair} which can endanger the stability of the matching.
A pair of agents is \myemph{blocking} a matching if they are not matched to each other 
but rank each other higher than their respective partners in the matching.

\citet{GaleShapley1962} introduced the \SM problem in the fields of Economics and Computer Science in the 1960s.
One of their central results was that 
every \SM instance with $2n$ agents admits a stable matching,
which can be found by their algorithm in $O(n^2)$~time.
Since then \SM has been intensively studied in Economics, Computer Science, and Social and Political Science~\cite{AbBiMa2005,GusfieldIrving1989,Irving2016a,Irving2016b,Knuth1976,Manlove2016,Manlove2013}.
Practical applications of \SM (and its variants) 
include partnership issues in various 
real-world scenarios, matching graduating medical students (so-called residents) with hospitals, students with schools,
and organ donors with patients~\cite{GusfieldIrving1989,Knuth1976,Manlove2013}, and the design of content delivery systems~\cite{MaggsSit15} and other distributed markets~\cite{RothSotomayor1992}.

%% %~\cite{GaleShapley1962}, 
% RN: I feel that the subesequent examples about because they only seem 
%``asymmetric'' wehereas Stable matching has a symmertic flavor. 
%ANywat, Stable Marrigae is famous and well-known enough and can go without 
%detailed motivation.
%allocating items (e.g.\ time-slots, locations, or resources) to agents (e.g.\ event hosts or individuals) where the agents have preferences over the available items~\cite{HyZe1979,AbrCheKummir2006,ChenSoenmez2002},
%centralized automated mechanisms that assign children to schools ~\cite{AbPaRo2005a,AbPaRo2005b}, school graduates to universities~\cite{BaBa2004,BiKi2015}, or medical students to hospitals~\cite{NResidentMatchingP,SResidentMatchingP}, 
%scheduling user jobs on machines so that users do not want to switch to some other machines, and finding receiver-donor pairs for organ transplants (for more details, see \citet{MaOM2014} and \citet{RoSoUn2007,RoSoUn2005}).

The original model of \SM assumes, roughly speaking, that there is a (subjective) criterion and that each agent has a single preference list depending on this criterion. 
In typically complex real-world scenarios, however, there are usually multiple aspects one takes into account when making a decision.
For instance, if we consider the classical partnership scenario,
then there could be different criteria such as working hours, family background, physical appearance, health, hobbies, etc. 
In other words, we face a much more complex multi-modal scenario.
Accordingly, the agents may have multiple preference lists, 
each defined by a different criterion;
we call each of these criteria a \myemph{layer}.
For an illustration, let us consider the following stable marriage example 
with two sets of two agents each, denoted as $u_1$, $u_2$, $w_1$, and~$w_2$,
and three layers of preferences, denoted as $P_1$, $P_2$, and~$P_3$.

\noindent 
{\centering
\begin{tikzpicture} 
  \node at (0, -0.5) (Layer1) {\Large $P_1$:};
  \foreach \i in {1, 2}
  {
    \foreach \j/\p/\o in {u/0/left,w/1/right} {
      \node[draw, circle, minimum size=3ex, inner sep=1pt] at (\i*\dist, -\p*\xsc) (n\j\i) {$\j_\i$};
    %  \node[\o = 0pt of \j\i]  (n\j\i) {$\j_\i$}; 
    }
  }

  \foreach \n / \i / \o / \a/\b   in {u1/w/left/1/2, u2/w/left/1/2}{
  % \node[\o = -7pt of n\n] {:};
    \gettikzxy{(n\n)}{\xx}{\yy};
    \node at (\xx,\yy+\ysc*25) {$\i_\a$};
    \node at (\xx,\yy+\ysc*15) {$\i_\b$};
  }  
  \foreach \n / \i / \o / \a/\b  in {%
    w1/u/right/1/2, w2/u/right/1/2}{
  % \node[\o = -7pt of n\n] {:};
    \gettikzxy{(n\n)}{\xx}{\yy};
    \node at (\xx,\yy-\ysc*18) {$\i_\a$};
    \node at (\xx,\yy-\ysc*28) {$\i_\b$};
  }
  \foreach \s/\t in {u1/w1,u2/w2}{
    \draw[blackline] (n\s) -- (n\t);
  } 
\end{tikzpicture}
~~\qquad
\begin{tikzpicture}
  \node at (0, -0.5) (Layer1) {\Large $P_2$:};
  \foreach \i in {1, 2}
  {
    \foreach \j/\p/\o in {u/0/left,w/1/right} {
      \node[draw, circle, minimum size=3ex, inner sep=1pt] at (\i*\dist, -\p*\xsc) (n\j\i) {$\j_\i$};
    %  \node[\o = 0pt of \j\i]  (n\j\i) {$\j_\i$}; 
    }
  }

  \foreach \n / \i / \o / \a/\b   in {u1/w/left/2/1, u2/w/left/2/1}{
  % \node[\o = -7pt of n\n] {:};
    \gettikzxy{(n\n)}{\xx}{\yy};
    \node at (\xx,\yy+\ysc*25) {$\i_\a$};
    \node at (\xx,\yy+\ysc*15) {$\i_\b$};
  }  
  \foreach \n / \i / \o / \a/\b  in {%
    w1/u/right/1/2, w2/u/right/1/2}{
  % \node[\o = -7pt of n\n] {:};
    \gettikzxy{(n\n)}{\xx}{\yy};
    \node at (\xx,\yy-\ysc*18) {$\i_\a$};
    \node at (\xx,\yy-\ysc*28) {$\i_\b$};
  }
  \foreach \s/\t in {u1/w2,u2/w1}{
    \draw[blackline] (n\s) -- (n\t);
  } 
\end{tikzpicture}
~~\qquad
\begin{tikzpicture}
  \node at (0, -0.5) (Layer1) {\Large $P_3$:};
  \foreach \i in {1, 2}
  {
    \foreach \j/\p/\o in {u/0/left,w/1/right} {
      \node[draw, circle, minimum size=3ex, inner sep=1pt] at (\i*\dist, -\p*\xsc) (n\j\i) {$\j_\i$};
    %  \node[\o = 0pt of \j\i]  (n\j\i) {$\j_\i$}; 
    }
  }

  \foreach \n / \i / \o / \a/\b   in {u1/w/left/1/2, u2/w/left/2/1}{
  % \node[\o = -7pt of n\n] {:};
    \gettikzxy{(n\n)}{\xx}{\yy};
    \node at (\xx,\yy+\ysc*25) {$\i_\a$};
    \node at (\xx,\yy+\ysc*15) {$\i_\b$};
  }  
  \foreach \n / \i / \o / \a/\b  in {%
    w1/u/right/2/1, w2/u/right/1/2}{
  % \node[\o = -7pt of n\n] {:};
    \gettikzxy{(n\n)}{\xx}{\yy};
    \node at (\xx,\yy-\ysc*18) {$\i_\a$};
    \node at (\xx,\yy-\ysc*28) {$\i_\b$};
  }
  \foreach \s/\t in {u1/w1,u2/w2}{
    \draw[blackline] (n\s) -- (n\t);
  } 
  \foreach \s/\t in {u1/w2,u2/w1} {
    \draw[greenline] (n\s) -- (n\t);
  }
\end{tikzpicture}
\par}

\noindent 
In the above diagram the preferences are depicted right above (respectively, right below) the corresponding agents; preferences are represented through vertical lists where more preferred agents are put above the less preferred ones.
For example, in the first layer, 
all agents from the same set have the same preference list, \emph{i.e.}\ both~$u_1$ and $u_2$ rank $w_1$ higher than $w_2$ while both $w_1$ and $w_2$ rank $u_1$ higher than $u_2$.
Similarly, in the second layer, both~$u_1$ and $u_2$ rank $w_2$ higher than $w_1$ while both $w_1$ and $w_2$ rank $u_1$ higher than $u_2$.
In the last layer, the preference lists of two agents from the same set are reverse to each other. For instance, $u_1$ ranks $w_1$ higher than $w_2$, which is opposite to $u_2$. 
In terms of the classic stable marriage problem, we will have three independent instances, one for each layer. 
The corresponding stable matching(s) for each instance are depicted through the edges between the agents. 
For instance, the first layer admits exactly one stable matching, which matches $u_1$ with $w_1$, and $u_2$ with $w_2$.
Yet, if we want to take all these layers jointly into account, then we need to extend the traditional concept of stability.

With multiple preference lists for each agent,
there are many natural ways to extend the original stability concept.
%While one extension 
%could be to require a matching to be \myemph{stable} in %\myemph{each} of 
%(at least) a certain number of layers,
%another could be to forbid any unmatched pair to \myemph{block} more than a certain number of layers (we will explain the differences between these approaches and discuss their applicability in Sections~\ref{sec:defi}~and~\ref{sec:relation}).
We propose three naturally emerging concepts of stability. Assume each agent has $\ell$ (possibly different) preferences lists.
%We give concrete scenarios for our concepts in \cref{ex:global,ex:pair,ex:individual}.
All three concepts are defined for a certain threshold~$\alpha$ with $1\le \alpha \le \ell$, which quantifies ``the strength'' of stability.
In the following, we briefly describe our three concepts and defer the formal definitions to \cref{sec:defi}.
\begin{itemize}[-]
\item The first one, called \myemph{\gstability{$\alpha$}}, 
extends the original stability concept in a straightforward way.
It assumes that the matched pairs agree on a set~$S$ of $\alpha$ layers where no unmatched pair is blocking the matching in any layer from~$S$.

In our introductory example, the matching~$M_1=\{\{u_1,w_1\}, \{u_2,w_2\}\}$ is stable in the first and the last layer, and thus it is a \gstable{$2$} matching.

\item The second one, called \myemph{\abstability{$\alpha$}}, 
changes the ``blocking ability'' of the unmatched pairs.
It \myemph{forbids} an unmatched pair to block more than $\ell-\alpha$ layers. In other words, each pair of matched agents needs to be stable in some $\alpha$ layers, but the choice of these layers can be different for different pairs.

Considering again our running example, we can verify that the \gstable{$2$} matching~$M_1$ is also \abstable{$2$} as each unmatched pair is blocking at most one layer.
Indeed, we will see that \abstability{$\alpha$} is strictly weaker than \gstability{$\alpha$} (\cref{prop:relation1} and \cref{ex:matching_istab_bstab_not_gstab}).

\item The last one, called \myemph{\aistability{$\alpha$}},
focuses on the ``willingness'' of an agent to stay with its partner.
It requires that for each unmatched pair, at least one of the agents prefers to stay with its partner in at least $\alpha$ layers.

In our introductory example, the matching~$M_1$ is also \aistable{$2$}. 
Thus, it is tempting to assume that \aistability{$\alpha$} also generalizes \gstability{$\alpha$}.
This is, however, not true as the following matching~$M_2=\{\{u_1,w_2\}, \{u_2,w_1\}\}$ is \gstable{$2$} but not \aistable{$2$}. 
Neither does the latter implies the former.
We refer to \cref{ex:matching_bstab_not_istab} for more explanations.
\end{itemize}
%\cref{fig:relation_between_stability_concepts} in \cref{sec:relation} illustrates how the three concepts relate to each other.
%We also analyze the computational complexity of finding a multi-layer stable matching with respect to the different concepts (see \cref{table:results}).
%The complexity results are summarized in \cref{table:results}

\subsection{Related work}
While we are not aware of research on an arbitrary number~$\ell$ of layers, there is some work on $\ell =2$~layers. %there seems to be work on two layers. 
\citet{Weem1999} considered the case where % where there are two layers and
each agent has two preference lists that are 
%a preference list for each layer such that the preference lists are 
the reverse of each other.
He provided a polynomial-time algorithm to find a \myemph{bistable} matching, \emph{i.e.}\ a matching that is stable in both layers.
Thus, while his concept falls into our \gstability{$\alpha$} concept for $\alpha=\ell=2$, it is a special case since the preference lists in the two layers are the reverse of each other.
In fact, for $\alpha=\ell=2$, we show that  the complexity of determining \gstability{$\alpha$} is NP-hard. 

%Given the fact that each agent may have different preference lists for multiple layers, 
%one could think about first 
Aggregating the preference lists of multiple layers into one (by comparing each pair of agents) and then searching for a ``stable'' matching for the agents with aggregated preferences is a plausible approach to multi-modal stable marriages.
As already noted by \citet{FarGeoKoe2016}, the aggregated preferences may be intransitive or even cyclic. 
Addressing this situation, they consider a generalized variant of \SM, 
where each agent~$u$ of one side, say~$U$, has a strict preference list~$\pref_u$ (as in the original \SM) 
while each agent~$w$ of the other side, say $W$, may order each possible pair of partners separately, expressed by a subset~$B_w \subseteq U\times U$ of ordered pairs.
They defined a matching~$M$ to be \myemph{stable} if no unmatched pair~$\{u,w\}$ satisfies ``$w\pref_u M(u)$ and $(u,M(u))\in B_w$''.
It turns out that our concept of individual stability and their concept for a more generalized case where both sides of the agents may have intransitive preferences are related, and we can use one of their results as a subroutine. In a way, our analysis provides a more fine-grained view, since we consider a richer model and thus are able to discuss how certain assumptions on elements of this model (e.g., the number of layers, the threshold value $\alpha$, etc.) affect the computational complexity of the problem. 
%The reduction from their proposed variant produces an instance with unbounded~$\ell$ and $\alpha=\nicefrac{\ell}{2}$. 
%Since they showed NP-hardness for their problem, we obtain as a corollary that our \aistability{$\alpha$} with $\alpha=\nicefrac{\ell}{2}$ and $\ell$ unbounded is NP-hard.
%We strengthen their result by showing that the hardness even holds if $\alpha=\nicefrac{\ell}{2}=2$.

\citet{ABFGHMR17} considered a variant of \SM{}, where each agent has a probability for each ordered pair of potential partners.
Assigning a probability of 1 to either $(x,y)$ or $(y,x)$ for each $x$ and $y$, their variant is closely related to the one of \citet{FarGeoKoe2016}
and is shown to be NP-hard.
%Thus, using a similar argument as above, we can also show that the specific variant of \citet{ABFGHMR17} is polynomial-time reducible to our \aistability{$\alpha$} concept with $\alpha=\nicefrac{\ell}{2}$ and $\ell$ unbounded.

%The \SM problem has been extended in many ways.
%First, the preferences lists of the agents may be \myemph{incomplete} or have \myemph{ties}.
%Two prominent extension of the stability, the \myemph{strong stability} and the \myemph{super stability}, are introduced for the case when ties exist.
%Second, there are no two disjoint sets of agents but rather each agent can be matched to any other agent. 
%This case is referred to as the \SR problem in the literature.
We refer to several expositions \cite{Knuth1976,GusfieldIrving1989,IwamaMiyazaki2008,Manlove2013,KMR16,Bir17} for a broader overview on \SM and related problems.

\subsection{Our contributions} 
We introduce three main concepts of stability for \SM with multi-modal preferences.
In \cref{sec:defi} we formally define these concepts, 
\myemph{global stability}, \myemph{pair stability}, and \myemph{individual stability}, and provide motivating and illustrating examples.
In \cref{sec:relation}, we study the relations between the three concepts and
show that pair stability is the least restrictive form while global and 
individual stability are in general incomparable (also see
\cref{fig:relation_between_stability_concepts} for a much refined picture).
In \cref{sec:all-stab}, we consider the special case of 
all-layers stability ($\alpha = \ell$) for the three concepts.
On the one hand, we provide a polynomial-time algorithm for checking individual stability for arbitrary large number of preference lists.
On the other hand, through an involved construction, we show NP-hardness for the other two stability concepts, even if there are only two layers.
The hardness results demonstrate a complexity dichotomy for both global and pair stability since for single-layer preference lists, all three concepts of stability are the same and polynomial-time computable.
In \cref{sec:notall-stab} we investigate the case of finding stable matchings with respect to less than all layers and only find NP-hardness results.
In \cref{sec:constr-pref}, we identify two special scenarios with strong but natural 
restrictions on the preference lists. 
For the fist scenario we assume that one side of the agents has \myemph{single-layered preferences}, \emph{i.e.}\ on one side the preference list of each agent remains the same in all layers.
We find that under such restrictions two out of three studied concepts are equivalent, and can be computed in polynomial time; for global stability we obtain W[1]-hardness (and also NP-hardness) and XP~membership for 
the threshold parameter~$\alpha$.
In the second scenario, we assume that 
the preferences of all agents on each side are uniform in each layer, \emph{i.e.}\ 
when for each fixed layer and side all agents have the same preference list, and when considering individual stability we find surprising tight connections 
to the complexity of the \textsc{Graph Isomorphism} problem.
\cref{table:results} gives a broad overview on our complexity 
results.

{
\begin{table}[t]
\caption{The computational complexity of finding matchings stable according to the three consideerd definitions---\gstability{$\alpha$}, \abstability{$\alpha$}, and \aistability{$\alpha$}---for instances with $2n$ agents and $\ell$ layers. All results hold for each value of $\alpha$ specified in the first column. Results marked with~$^{*}$ hold even if we assume that each agent of one side has the same preference list in all layers. The NP-hardness results hold even for a fixed number of layers.}
\centering
\resizebox{\textwidth}{!}{
\begin{tabular}{@{}l l l l@{}}
  \toprule
  Parameters & global stability & pair stability & individual stability\\
  \midrule
  Arbitrary \\
  \; $1=\alpha $ & $O(n^2)$~\cite{GaleShapley1962} &  $O(n^2)$~\cite{GaleShapley1962} & $O(n^2)$~\cite{GaleShapley1962} \\
  \; $2 \le \alpha=\ell$ & \cellnph{NP-h~[T.~\ref{thm:global-hard-alpha=ell=constant}+P.~\ref{prop:global=pair-hard-alpha=ell}]}  & \cellnphh{NP-h~[C.~\ref{cor:pair-hard-alpha=ell=constant}+P.~\ref{prop:global=pair-hard-alpha=ell}] } & \cellp{$O(\ell\cdot n^2)$~[T.~\ref{thm:l_individual_polynomial}]}\\
  \; $\lfloor \nicefrac{\ell}{2} \rfloor < \alpha<\ell$ &  \cellnphh{NP-h~[P.~\ref{prop:global-alpha_ge_2}]}  &  \cellnphh{NP-h~[P.~\ref{pro:pair-alpha_ge_ell/2}+P.~\ref{prop:hardness_pair_stable-alpha=ell/2+1}]} & ?\\
  \; $2\le \alpha \le \lfloor \nicefrac{\ell}{2} \rfloor$ &  \cellnphh{NP-h~[P.~\ref{prop:global-alpha_ge_2}]} &  \cellnphh{NP-h$^{*}$~[C.~\ref{cor:hardness_pair_stable}]} & \cellnph{NP-h$^{*}$~[T.~\ref{thm:hardness_individual_stable}]} \\
  \\[-2ex]
       %      & ? & NP-h ({\small for each $\alpha \le \lfloor \nicefrac{\ell}{2} \rfloor$}) & NP-h ({\small for each $\alpha \le \lfloor \nicefrac{\ell}{2} \rfloor$})\\[1ex]
  Single-layered\\
 % \; $\alpha > \lfloor \nicefrac{\ell}{2} \rfloor $ & \cellfpt{FPT for~$\alpha$~[T.~\ref{thm:single-sided-global-w-hard}]} & \cellp{$O(\ell \cdot n^2)$~[P.~\ref{prop:single-sided-individual-poly}]} & \cellp{$O(\ell \cdot n^2)$~[P.~\ref{prop:single-sided-individual-poly}]}\\
  %\; $2\le \alpha \le  \lfloor \nicefrac{\ell}{2} \rfloor$ 
             & \cellnph{NP-h for unbounded~$\alpha$~[T.~\ref{thm:single-sided-global-w-hard}]} &  \cellnphh{NP-h$^{*}$~when $2\le \alpha \le  \lfloor \nicefrac{\ell}{2} \rfloor$ [C.~\ref{cor:hardness_pair_stable}]}  & \cellnph{NP-h$^{*}$ when $2\le \alpha \le  \lfloor \nicefrac{\ell}{2} \rfloor$~[T.~\ref{thm:hardness_individual_stable}]}\\
 & \cellwh{W[1]-h \& in XP for $\alpha$~[T.~\ref{thm:single-sided-global-w-hard}]} 
                                 & \cellp{$O(\ell \cdot n^2)$ when $\alpha > \lfloor \nicefrac{\ell}{2} \rfloor$~[P.~\ref{prop:single-sided-individual-poly}]} & \cellp{$O(\ell \cdot n^2)$ when $\alpha > \lfloor \nicefrac{\ell}{2} \rfloor $~[P.~\ref{prop:single-sided-individual-poly}]}\\
  \\[-2ex]
  Uniform\\
  \; $\alpha \geq \nicefrac{\ell}{2}+1$& \cellp{$O(\ell \cdot n)$~[P.~\ref{prop:uniform-global-p}]}  & ? & \celllogsnp{$n^{O(\log{(n)})}+O(\ell \cdot n^2)$~[C.~\ref{cor:uniform-alpha>l/2-individual-subexp}]}\\
  \; $\alpha = \nicefrac{\ell}{2}$& \cellp{$O(\ell \cdot n)$~[P.~\ref{prop:uniform-global-p}]} & ? & \cellgh{\textsc{Graph Isom.}-hard~[T.~\ref{thm:uniform_istable_matching_as_hard_as_isomorphism}]}\\
  \bottomrule
\end{tabular}
}
\label{table:results}
\end{table}
}

%%%%%%%%%%%%%%%%%%%%%%%%%%%%%%%%%%%%%%%%%%%%%%%%%%%%%%%%%%%%%%%%%%%%%%%%%
%%%%%%%%%%%%%%%%%%%%%%%%%%%%%%%%%%%%%%%%%%%%%%%%%%%%%%%%%%%%%%%%%%%%%%%%%
\section{Definitions and Notations}
\label{sec:defi}
%%%%%%%%%%%%%%%%%%%%%%%%%%%%%%%%%%%%%%%%%%%%%%%%%%%%%%%%%%%%%%%%%%%%%%%%%
For each natural number $t$ by $[t]$ we denote the set $\{1, 2, \ldots, t\}$.

Let $U = \{u_1, \ldots, u_n\}$ and $W = \{w_1, \ldots, w_n\}$ be two $n$-element disjoint sets of agents. 
There are~$\ell$~layers of preferences, where $\ell$ is a non-negative integer.
For each $i \in [\ell]$ and each $u \in U$, let $\pref_u^{(i)}$ be a linear order on~$W$ that represents the ranking of agent $u$ over all agents from $W$ in layer~$i$. 
Analogously, for each $i \in [\ell]$ and each $w \in W$, the symbol~$\pref_w^{(i)}$ represents a linear order on~$U$ that encodes preferences of~$w$ in layer~$i$.
We refer to such linear orders as preference lists.
A preference profile~$P_i$ of layer~$i\in[\ell]$ is a collection of preference lists of all the agents in layer $i$,~$\{\succ^{(i)}_{a} \mid a \in U \cup W\}$.

Let $\pairset = \{\{u,w\}\mid u \in U \wedge w \in W\}$.
A \myemph{matching}~$M\subseteq \pairset$ is a set of pairwisely disjoint pairs, \emph{i.e.}\ 
for each two pairs~$p, p'\in M$ it holds that $p \cap p' = \emptyset$.
If $\{u,w\}\in M$, then we also use $M(u)$ to refer to $w$ and $M(w)$ to refer to $u$,
and we say that $u$ and $w$ are their respective partners under $M$;
otherwise we say that $\{ u,w\}$ is an \myemph{unmatched pair}.
\cref{ex:first_example} below shows an example matching and introduces a graphical notation that we will use throughout the paper. 

\begin{example}\label{ex:first_example}
%\em
Consider two sets of agents, $U=\{u_1,u_2,u_3\}$ and $W=\{w_1,w_2,w_3\}$, and two layers of preference profiles, $P_1$ and $P_2$. Let us recall that in the following diagram the preferences are represented through vertical lists where more preferred agents are put above the less preferred ones. For instance, in the diagram the preference list of agent~$u_3$ in the first layer (profile~$P_1$) is \mbox{$w_2 \pref_{u_3}^{(1)} w_3 \pref_{u_3}^{(1)} w_1$}. 

%\bigskip
%
%\begin{tikzpicture}
%  \node at (-3, -2 * 0.55) (Layer1) {\Large $P_1$:};
%  \foreach \i in {1, 2, 3}
%  {
%    \def \xsc {2}
%    \foreach \j/\p/\o in {u/0/left,w/1/right} {
%      \node[draw, circle, minimum size=3ex, inner sep=1pt] at (\p*\xsc, -\i * 0.55) (n\j\i) {$\j_\i$};
%    %  \node[\o = 0pt of \j\i]  (n\j\i) {$\j_\i$}; 
%    }
%  }
%
%  \foreach \n / \i / \o / \a/\b/\c  in {u1/w/left/1/2/3, u2/w/left/2/3/1, u3/w/left/3/2/1,
%  w1/u/right/1/2/3, w2/u/right/2/3/1, w3/u/right/3/1/2} {
%  %\node[\o = -7pt of n\n] {:};
%  \node[\o = 0pt of n\n ] {$\i_\a \pref \i_\b \pref \i_\c$};
%  }
%  \foreach \s/\t in {u1/w2,u2/w1,u3/w3} {
%    \draw[blackline] (n\s) -- (n\t);
%  }
%\end{tikzpicture}
%
%\vspace{-2mm}
%\hrulefill
%\vspace{2mm}
%
%\begin{tikzpicture}
%  \node at (-3, -2 * 0.55) (Layer2) {\Large $P_2$:};
%  \foreach \i in {1, 2, 3}
%  {
%    \def \xsc {2}
%    \foreach \j/\p/\o in {u/0/left,w/1/right} {
%      \node[draw, circle, minimum size=3ex, inner sep=1pt] at (\p*\xsc, -\i *0.55) (n\j\i) {$\j_\i$};
%    %  \node[\o = 0pt of \j\i]  (n\j\i) {$\j_\i$}; 
%    }
%  }
%
%  \foreach \n / \i / \o / \a/\b/\c  in {u1/w/left/2/3/1, u2/w/left/3/1/2, u3/w/left/1/2/3,
%  w1/u/right/1/2/3, w2/u/right/2/3/1, w3/u/right/3/1/2} {
%  %\node[\o = -7pt of n\n] {:};
%  \node[\o = 0pt of n\n ] {$\i_\a \pref  \i_\b \pref \i_\c$};
%  }
%  % first matching
%  \foreach \s/\t in {u1/w2,u2/w1,u3/w3} {
%    \draw[blackline] (n\s) -- (n\t);
%  }
%\end{tikzpicture}
%\smallskip

%\smallskip

{\centering
\begin{tikzpicture}
  \node at (0, -0.5) (Layer1) {\Large $P_1$:};
  \foreach \i in {1, 2, 3}
  {
    \foreach \j/\p/\o in {u/0/left,w/1/right} {
      \node[draw, circle, minimum size=3ex, inner sep=1pt] at (\i*\dist, -\p*\xsc) (n\j\i) {$\j_\i$};
    %  \node[\o = 0pt of \j\i]  (n\j\i) {$\j_\i$}; 
    }
  }

  \foreach \n / \i / \o / \a/\b/\c   in {u1/w/left/3/2/1, u2/w/left/1/2/3, u3/w/left/2/3/1}{
  % \node[\o = -7pt of n\n] {:};
    \gettikzxy{(n\n)}{\xx}{\yy};
    \node at (\xx,\yy+\ysc*35) {$\i_\a$};
    \node at (\xx,\yy+\ysc*25) {$\i_\b$};
    \node at (\xx,\yy+\ysc*15) {$\i_\c$};
  }  
  \foreach \n / \i / \o / \a/\b/\c  in {%
    w1/u/right/2/3/1, w2/u/right/3/1/2, w3/u/right/3/1/2}{
  % \node[\o = -7pt of n\n] {:};
    \gettikzxy{(n\n)}{\xx}{\yy};
    \node at (\xx,\yy-\ysc*18) {$\i_\a$};
    \node at (\xx,\yy-\ysc*28) {$\i_\b$};
    \node at (\xx,\yy-\ysc*38) {$\i_\c$};
  }
  \foreach \s/\t in {u1/w3,u2/w1,u3/w2} {
    \draw[blackline] (n\s) -- (n\t);
  }
\end{tikzpicture}
~~\qquad
\begin{tikzpicture}
  \node at (0, -0.5) (Layer1) {\Large $P_2$:};
  \foreach \i in {1, 2, 3}
  {
    \foreach \j/\p/\o in {u/0/left,w/1/right} {
      \node[draw, circle, minimum size=3ex, inner sep=1pt] at (\i*\dist, -\p*\xsc) (n\j\i) {$\j_\i$};
    %  \node[\o = 0pt of \j\i]  (n\j\i) {$\j_\i$}; 
    }
  }

  \foreach \n / \i / \o / \a/\b/\c   in {u1/w/left/2/3/1, u2/w/left/3/1/2, u3/w/left/1/2/3} {
  % \node[\o = -7pt of n\n] {:};
    \gettikzxy{(n\n)}{\xx}{\yy};
    \node at (\xx,\yy+\ysc*35) {$\i_\a$};
    \node at (\xx,\yy+\ysc*25) {$\i_\b$};
    \node at (\xx,\yy+\ysc*15) {$\i_\c$};
  }  
  \foreach \n / \i / \o / \a/\b/\c  in {%
     w1/u/right/1/2/3, w2/u/right/2/3/1, w3/u/right/3/1/2}{
  % \node[\o = -7pt of n\n] {:};
    \gettikzxy{(n\n)}{\xx}{\yy};
    \node at (\xx,\yy-\ysc*18) {$\i_\a$};
    \node at (\xx,\yy-\ysc*28) {$\i_\b$};
    \node at (\xx,\yy-\ysc*38) {$\i_\c$};
  }
 
  \foreach \s/\t in {u1/w1,u2/w2,u3/w3} {
    \draw[blackline] (n\s) -- (n\t);
  }
  % second matching
  \foreach \s/\t in {u1/w2,u2/w3,u3/w1} {
    \draw[blueline] (n\s) edge (n\t);
  }
  \foreach \s/\t in {u1/w3,u2/w1,u3/w2} {
    \draw[redline] (n\s) edge (n\t);
  }

\end{tikzpicture}
\par}

\noindent
In our diagrams we will depict stable matchings in each layer through edges between matched nodes. If a layer has more than one stable matching, then we will use different types of lines (solid, dashed, dotted) and different colors to distinguish between them. For instance, in the above example profile~$P_1$ has one stable matching $M_1=\{\{u_1, w_3\}, \{u_2, w_1\}, \{u_3, w_2\}\}$, and~$P_2$ has three stable matchings:
\begin{inparaenum}[(1)]
\item $M_2=\{\{u_1, w_1\}, \{u_2, w_2\}, \{u_3, w_3\}\}$,
\item $M_3=\{\{u_1, w_2\}, \{u_2, w_3\}, \{u_3, w_1\}\}$, and
\item $M_1$.
\end{inparaenum}
\hfill $\diamond$
\end{example}

Let us now introduce two notions that we will use when defining various concepts of stability.
%there are no two unmatched agents, $u$ and $w$, 
%$u$ and $w$ prefer to be with each other rather with their respective partners, i.e., % for all
% unmatched pairs $\{u, w\}$ it must hold that 

\begin{definition}[Dominating pairs and blocking pairs]
  Let $M$ be a matching over $U \cup W$.
  Consider an unmatched pair~$\{u,w\}\in (\pairset)\setminus M$ and a layer~$i\in [\ell ]$.
  We say that \myemph{$\{u,w\}$ dominates $\{u,v\}$ in layer~$i$} if $w \pref^{(i)}_{u} v$.
  We say that \myemph{$\{u,w\}$ is blocking matching~$M$ in layer~$i$}
  if it holds that 
  \begin{compactenum}[(1)]
    \item $u$ is unmatched in $M$ or $\{u,w\}$ dominates $\{u,M(u)\}$ in layer~$i$,
    and 
    \item $w$ is unmatched in $M$ or $\{u,w\}$ dominates $\{w,M(w)\}$ in layer~$i$.
  \end{compactenum}
%
%  A pair~$\{u,w\}\in M$ is \myemph{stable in matching $M$ in layer~$i$} if \myemph{no} other pair~$\{u',w'\}\in (\pairset) \setminus M$ blocks $\{u,w\}$.
\end{definition}

For a single layer~$i$, a matching $M$ is \myemph{stable in layer~$i$} if no unmatched pair is blocking $M$ in layer~$i$.
%Stable matching can be equivalently defined as one for which there exists no blocking pair.
Let us illustrate the concept of dominating and blocking pairs through \cref{ex:first_example}.
Consider the matching $M_3 = \{\{u_1, w_2\}, \{u_2, w_3\}, \{u_3, w_1\}\}$ and profile~$P_1$ of layer~$1$. Here, pair $\{u_1, w_3\}$ dominates both $\{u_1, w_2\}$ (since $u_1$ prefers $w_3$ to $w_2$) and
$\{u_2, w_3\}$ (since $w_3$ prefers $u_1$ to $u_2$). Thus, $\{u_1, w_3\}$ is a blocking pair and so it witnesses that $M$ is not stable in profile~$P_1$.
%\textbf{RN: No, this is wrong!}

We are interested in matchings which are stable in multiple layers, \emph{i.e.}\ we aim at generalizing the classic \textsc{Stable Marriage} problem~\cite{GaleShapley1962, Knuth1976, GusfieldIrving1989, Manlove2013} which is defined for a single layer to the case of multiple layers.
\begin{comment}
Before we motivate and discuss several concepts of matching stabilities for multiple layers, 
we first introduce the following three concepts that will be useful in our further discussion.
\begin{inparaenum}[(i)]
\item of blocking from one pair to the other, 
\item of blocking pairs, and 
\item of stable pairs with respect to a matching~$M\subseteq \pairset$.
\end{inparaenum}

% We say that a pair~$\{u,w\}\in M$ \myemph{blocks $M$ in layer $i$},
% if both $u$ and $w$ prefer each other to their respect partner under $M$,
% formally,
% \begin{align*} w\pref_u^{(i)} M(u) \text{ and } u \pref_w^{(i)} M(w)\text{,}\end{align*}
% where $\pref_u^{(i)}$ and $\pref_w^{(i)}$ are $u$'s and $w$'s linear preferences in layer $i$. 

Let us now describe different forms of multi-layer stability.
\end{comment}
The idea behind each of the concepts defined below is similar: in order to call a matching stable for multiple layers we require
that it must be stable in at least a certain, given number of layers $\alpha$ ($\alpha$ is a number indicating the ``strength'' of the stability).
However, for different concepts we require a different level of agreement with respect to which layers are required for stability.
Informally speaking, on the one end of the spectrum we have a variant of stability where we require a global agreement of the agents regarding the set of $\alpha$ layers for which the matching must be stable. On the other end of the spectrum we have a variant where 
we assume that the agents act independently: an agent $a$ would deviate if it would find another agent, say $b$, such that $a$
prefers~$b$ to its matched partner in some $\alpha$ layers, and $b$ prefers $a$ to its matched partner in another, possibly different, set of $\alpha$ layers.
In the intermediate case, we require that a deviating pair must agree on the subset of layers which form the reason for deviation.
We formally define the three concepts below.

\subsection{\gstability{$\alpha$}}

Informally speaking, a matching~$M$ is \myemph{\gstable{$\alpha$}} if there exist $\alpha$ layers in each of which $M$ is stable. %Formally:

\begin{definition}[global stability]\label{def:global}
A matching $M$ is \emph{\gstable{$\alpha$}} if there exists a set~$S\subseteq [\ell ]$ of $\alpha$~layers,
such that for each layer~$i\in S$ and for each unmatched pair~$\{u,w\}\in \pairset\setminus M$
at least one of the two following conditions holds:
   \begin{compactenum}[(1)]
    \item pair~$\{u,M(u)\}$ dominates $\{u,w\}$ in layer~$i$, 
    or 
    \item 
    pair~$\{w, M(w)\}$ dominates $\{w, u\}$ in layer~$i$. 
  \end{compactenum}
\end{definition}

The following example describes a scenario where the above concept of multi-layer stability appears to be useful.

\begin{example}\label{ex:global} %\em
Assume that the preferences of the agents depend on external circumstances which are not known a priori. 
Assume that each layer represents a different possible state of the universe. If we want to find a matching that is stable in as many states of the universe as possible, then we need to find an \gstable{$\alpha$} matching for the highest possible value of $\alpha$.    
%For instance, assume that we want to match students to courses (notably, an instance of the one-to-many matching problem can be reduced to an equivalent instance of the one-to-one matching~\cite{GusfieldIrving1989}).
%The preferences of the students over the courses may depend on many aspects, such as on whether a course will get additional funding for scholarships (and so students may 
%prefer courses where it is easier to get a higher grade, increasing the chances for a scholarship), or on whether the corresponding university will decide to open related courses in the following semesters.
%However, this information is not known when students are matched. If we want to match students to courses so that the matching is stable in sufficiently many, say $\alpha$, cases, then we need to find an \gstable{$\alpha$} matching.   
%We note that the lecturers of the courses may also have preferences over the students so this scenario is indeed a two-sided stable matching scenario.
\hfill $\diamond$
\end{example}

Already for \gstability{$\alpha$} we see substantial differences compared to the original concept of stability for a single layer. 
It is guaranteed that such a matching always exists for $\alpha = 1$; indeed this would be a matching that is stable in an arbitrary layer.  
However, one can observe that as soon as $\alpha > 1$ an \gstable{$\alpha$} matching might not exist (see \cref{ex:matching_istab_bstab_not_gstab}).

\subsection{\abstability{$\alpha$}}
While $\alpha$-layer global stability requires that the agents globally agree on a certain subset of $\alpha$~layers for which the matching should be stable,
pair stability forbids each unmatched pair to block more than a certain number of layers.
The formal definition, using the domination concept, is as follows:
\begin{definition}[pair stability]\label{def:abstability}
  A matching $M$ is \emph{\abstable{$\alpha$}} if for each unmatched pair~$\{u,w\}\in (\pairset) \setminus M$,
  there is a set~$S\subseteq [\ell]$ of $\alpha$~layers such that
  for each layer~$i\in S$ at least one of the following conditions holds:
  \begin{compactenum}[(1)]
    \item pair~$\{u,M(u)\}$ dominates $\{u,w\}$ in layer~$i$, 
    or 
    \item 
    pair~$\{w, M(w)\}$ dominates $\{w, u\}$ in layer~$i$. 
  \end{compactenum}
\end{definition}

% in layer $i$. Clearly, there exist $\alpha$~layers where $\{u, w\}$ is not blocking if and only if there do not exist $(\ell - \alpha + 1)$ layers where $\{u, w\}$ is blocking.  

\cref{def:abstability} can be equivalently formulated using a generalization of the concept of blocking pairs.
Let $\beta\in [\ell]$ be an integer bound. We say that a pair $\{u,w\}\in (\pairset)\setminus M$ is \myemph{$\beta$-blocking}~$M$
if there exists a subset~$S \subseteq \{1,2,\ldots, \ell\}$ of $\beta$ layers such that for each $i \in S$, pair~$\{u,w\}$ is blocking~$M$ in layer~$i$.

\begin{proposition}\label{prop:pair-stable-alternative-def}
  A matching~$M$ is \abstable{$\alpha$} if and only if no unmatched pair~$p$ is $(\ell-\alpha+1)$-blocking~$M$.
\end{proposition}
\begin{proof}
  To prove the statement, we show that a matching~$M$ is \myemph{not} \abstable{$\alpha$} if and only if there is an unmatched pair~$p$ that is $(\ell-\alpha+1)$-blocking $M$.
  For the ``if'' direction, assume that $\{u,w\}$ is an unmatched pair and $R\subseteq [\ell]$ is a subset of $\ell-\alpha+1$ layers such that $\{u,w\}$ is blocking every layer in~$R$.
  Now consider an arbitrary subset~$S\subseteq [\ell]$ of size~$\alpha$. 
  By the cardinalities of $R$ and $S$, it is clear that $S\cap R \neq \emptyset$.
  Let $i\in S\cap R$ be such a layer. 
  Then, by assumption, we have that $\{u,w\}$ is blocking $M$ in layer~$i$.
  This means that none of the conditions stated in \cref{def:abstability} holds.
  Thus, $\{u,w\}$ is an unmatched pair witnessing that $M$ is not \abstable{$\alpha$}.

  For the ``only if'' direction, assume that $M$ is not \abstable{$\alpha$} and let $\{u,w\}$ be an unmatched pair that witnesses the non-\abstability{$\alpha$} of $M$.
  We claim that $\{u,w\}$ is $(\ell-\alpha+1)$-blocking $M$.
  Towards a contradiction, suppose that $\{u,w\}$ is not $(\ell-\alpha+1)$-blocking $M$.
  Then, there must be a subset~$S\subseteq [\ell]$ of at least $\alpha$ layers where the pair $\{u,w\}$ is not blocking $M$ in each layer in $S$.
  Equivalently, we can say that for each layer~$i\in S$, $\{u,M(u)\}$ dominates $\{u,w\}$ in layer~$i$ or $\{w,M(w)\}$ dominates $\{u,w\}$ in layer~$i$---a contradiction to $\{u,w\}$ being a witness.
\end{proof}

The following example motivates \abstability{$\alpha$}.

\begin{example}\label{ex:pair} %\em
Consider the case when the preferences of the agents depend on a context, yet a context is pair-specific. For instance, in matchmaking a woman may have different preferences over men depending on which country they will decide to live in. Thus, a pair of a man and a woman is blocking if they agree on certain conditions, and if they will find each other more attractive than their current partners according to the agreed conditions. 
\hfill $\diamond$
\end{example}

In \cref{sec:relation}, we show that \gstability{$\alpha$} implies \abstability{$\alpha$} (\cref{prop:relation1}). This, among other things, implies that for $\alpha=1$
an \abstable{$\alpha$} matching always exists. 
However, as soon as $\alpha \ge 2$ the existence is no longer guaranteed (see \cref{ex:matching_istab_bstab_not_gstab}).

\subsection{\aistability{$\alpha$}}

We move to the third and last concept of stability.

\begin{definition}[individual stability]\label{def:individual}
  A matching~$M$ is \myemph{\aistable{$\alpha$}} 
  if for each unmatched pair~$\{u,w\}\in (\pairset )\setminus M$
  there is a set~$S \subseteq [\ell ]$ of $\alpha$ layers
  such that at least one of the following conditions holds: 
  \begin{compactenum}[(1)]
    \item pair~$\{u,M(u)\}$ dominates $\{u,w\}$ in each layer of $S$, 
    or 
    \item % there is a (possibly different) set of layers: $S' \subseteq \{1, \ldots, \ell\}$ with $|S'| \ge \alpha$, such that
    % in each layer of $S'$, 
    pair~$\{w, M(w)\}$ dominates $\{w, u\}$ in each layer of $S$. 
  \end{compactenum}
\end{definition}

The following example illustrates a potential application in the domain of partnership agencies. %for the concept of \aistability{$\alpha$}.

\begin{example}\label{ex:individual} %\em
Assume that each layer describes a single criterion for preferences. 
The preferences of each agent may differ depending on the criterion.
For instance, the two sets of agents can represent, respectively, men and women, as in the traditional stable marriage problem.
Different criteria may correspond, for instance, to the intelligence, sense of humor, physical appearance etc. Assume that an agent $a$ will have no incentive to break his or her relationship with~$b$, and to have an affair with~$c$ if he or she prefers~$b$ to~$c$ according to at least $\alpha$ criteria. In order to match men with women so that they form stable relationships, one needs to find an \aistable{$\alpha$} matching.  
\hfill $\diamond$
\end{example}

\cref{def:individual} can be equivalently formulated using a generalization of the concept of dominating pairs.
Let $\beta\in [\ell]$ be an integer bound.
We say that a pair~$\{u,w\}$ is \myemph{$\beta$-dominating} $\{u,w'\}$ if there is a subset~$R\subseteq [\ell]$ of $\beta$~layers such that for each $i\in R$ the pair~$\{u,w\}$ dominates $\{u,w'\}$ in layer~$i$.

\begin{proposition}\label{prop:individual-stable-alternative-def}
  A matching~$M$ is \aistable{$\alpha$} if and only if no unmatched pair~$\{u,w\}$ exists that is both $(\ell-\alpha+1)$-dominating~$\{u,M(u)\}$ and $(\ell-\alpha+1)$-dominating~$\{w,M(w)\}$.
\end{proposition}

\begin{proof}
   To prove the statement, we show that a matching~$M$ is \myemph{not} \aistable{$\alpha$} if and only if there is an unmatched pair~$p$ that is $(\ell-\alpha+1)$-dominating $\{u,M(u)\}$ and $(\ell-\alpha+1)$-dominating $\{w,M(w)\}$.
  For the ``if'' direction, assume that $\{u,w\}$ is an unmatched pair and $R_1, R_2\subseteq [\ell]$ are two (possibly different) subsets of $\ell-\alpha+1$ layers each, such that $\{u,w\}$ is dominating $\{u,M(u)\}$ in each layer~$i\in R_1$
  and is dominating $\{w,M(w)\}$ in each layer~$j\in R_2$.
  Now consider an arbitrary subset~$S\subseteq [\ell]$ of size $\alpha$. 
  By the cardinalities of $R_1$, $R_2$, and $S$ , it is clear that $S\cap R_1 \neq \emptyset$ and $S\cap R_2\neq \emptyset$.
  Let $i\in S\cap R_1$ and $j\in S\cap  R_2$ be two layers in the intersections.
  Then, by assumption, we have that $\{u,w\}$ is dominating $\{u,M(u)\}$ in layer~$i$ and $\{u,w\}$ is dominating~$\{w,M(w)\}$ in layer~$j$.
  This means that none of the conditions stated in \cref{def:individual} holds.
  Thus, $\{u,w\}$ is an unmatched pair witnessing that~$M$ is not \aistable{$\alpha$}.
  
  For the ``only if'' direction, assume that $M$ is not \aistable{$\alpha$} and let $\{u,w\}$ be an unmatched pair that witnesses the non-\aistability{$\alpha$} of $M$.
  We claim that $\{u,w\}$ is $(\ell-\alpha+1)$-dominating $\{u,M(u)\}$ and 
  is $(\ell-\alpha+1)$-dominating $\{w,M(w)\}$.
  Towards a contradiction, first suppose that~$\{u,w\}$ is \emph{not} $(\ell-\alpha+1)$-dominating $\{u,M(u)\}$, meaning that there are at most $\ell-\alpha$ layers where $\{u,w\}$ dominates $\{u,M(u)\}$.
  This implies that there is a subset~$S\subseteq [\ell]$ of $\alpha$ layers such that 
  for each $i \in S$, the pair~$\{u,M(u)\}$ is dominating $\{u,w\}$, a contradiction to $\{u,w\}$ being a witness of the non-\aistability{$\alpha$} of $M$.
  Analogously, if  $\{u,w\}$ was \emph{not} $(\ell-\alpha+1)$-dominating $\{w,M(w)\}$, then we could obtain the same contradiction.
\end{proof}

For $\alpha=1$ an \aistable{$\alpha$} matching always exists (it will follow from \cref{prop:relation1,prop:1-pair<=>1-individual}); however, this is no longer the case when $\alpha \ge 2$ (see \cref{prop:relation2} and \cref{ex:matching_istab_bstab_not_gstab}). 

Observe that according to \aistability{$\alpha$} the preferences of the agents can be represented as \myemph{sets} of linear orders: it does not matter which preference order comes from which layer. This is not the case for the other two concepts. 

\subsection{Central computational problems}
In this paper, we study the algorithmic complexity of finding matchings that are stable according to the above definitions. To this end, we first investigate 
how the three concepts relate to each other.
Next, we formally define the search problem of finding an \gstable{$\alpha$} matching.

\searchprob{\GSM}
{Two disjoint sets of $n$~agents each, $U$ and $W$,
$\ell$ preference profiles, and an integer bound~$\alpha \in [\ell]$.}
{Return an \gstable{$\alpha$} matching if one exists, or claim there is no such.}

\noindent The other two problems, \BSM and \ISM, are defined analogously.

\section{Relations Between the Multi-Layer Concepts of Stability}
\label{sec:relation}
Below we establish relations among the three concepts. We start by showing that \abstability{$\alpha$} is a weaker notion than \gstability{$\alpha$} and \aistability{$\alpha$}.

\begin{proposition}\label{prop:relation1}
  An \gstable{$\alpha$} matching is \abstable{$\alpha$}.
\end{proposition}
\begin{proof}
Let $M$ be an \gstable{$\alpha$} matching and let $S \subseteq \{1, \ldots, \ell\}$ be such that $|S| = \alpha$ and
that for each $i \in S$, matching~$M$ is stable in layer~$i$. 
For the sake of contradiction let us assume that $M$ is not
\abstable{$\alpha$}.
By \cref{prop:pair-stable-alternative-def}, let $\{u,w\}$ be an $(\ell-\alpha+1)$-blocking pair for $M$. 
Let $S' \subseteq \{1, \ldots, \ell\}$
be such that $|S'| = \ell - \alpha + 1$ and that for each $i \in S'$, 
pair~$\{u,w\}$ blocks $M$ in layer $i$. Since $|S'| + |S| = \ell + 1$, we get that $S \cap S' \neq \emptyset$.
Let $i \in S \cap S'$. This gives a contradiction since on the one hand $M$ is stable in layer~$i$, and on the other hand $\{u,w\}$ blocks $M$ in layer~$i$.
%
%Let $M$ be an \apstable{$\alpha$} matching. 
%For the sake of contradiction let us assume that $M$ is not
%\abstable{$\alpha$}.
%Let $\{u,w\}$ be an $(\ell-\alpha+1)$-blocking pair for $M$. 
%Let $S \subseteq \{1, \ldots, \ell\}$
%be such that $|S| = \ell - \alpha + 1$ and that for each $i \in S$, 
%pair~$\{u,w\}$ blocks $M$ in layer $i$. 
%Thus, for each $S' \subseteq \{1, \ldots, \ell\}$ with $|S'| = \alpha$, we know that there exists an $i \in S'$ such that 
%pair~$\{u,w\}$ blocks $M$ in layer $i$, that is, 
%$w\pref_u^{(i)} M(u)$ and $u \pref_w^{(i)} M(w)$, where $\pref_u^{(i)}$ and $\pref_w^{(i)}$ are $u$'s and $w$'s preference orders in layer~$i$.
%In particular, these two conditions say that pair~$\{u, M(u)\} \in M$ is not stable in~$M$ in layer~$i$. 
%Summarizing, we have shown that for each subset $S'$ of $\alpha$ layers there exist a layer $i \in S'$ such that 
%pair~$\{u, M(u)\}$ is not stable in~$M$ in layer~$i$. This means that $M$ is 
%not \apstable{$\alpha$}, and gives a contradiction.
%
%It is clear that an \aistable{$\alpha$} matching is also \abstable{$\alpha$}.
%This completes the proof.
\end{proof}

\begin{proposition}\label{prop:relation2}
  An \aistable{$\alpha$} matching is \abstable{$\alpha$}.
\end{proposition}
\begin{proof}
Consider an \aistable{$\alpha$} matching $M$.
Towards a contradiction, suppose that $M$ is not \abstable{$\alpha$}. 
By \cref{prop:pair-stable-alternative-def}, this means that there exists an unmatched pair~$\{u, w\}$ and a subset~$S'\subseteq [\ell]$ of $\ell - \alpha + 1$ layers such that  $\{u, w\}$ is blocking $M$ in each layer from $S'$. 
This means that $\{u, w\}$ is both $(\ell-\alpha+1)$-dominating $\{u, M(u)\}$
and $(\ell-\alpha+1)$-dominating $\{w,M(w)\}$.
Then, by \cref{prop:individual-stable-alternative-def}, $M$ is not \aistable{$\alpha$}, a contradiction.
% both $\{u, M(u)\}$ and $\{w, M(w)\}$. Thus, for each set $S'$ of $\alpha$ layers there is at least one layer $i$ (a layer from $S \cap S'$) such that $\{u, w\}$ dominates both $\{u, M(u)\}$ and $\{w, M(w)\}$ in $i$. Consequently, $M$ cannot be \aistable{$\alpha$}, a contradiction.
\end{proof}

\cref{ex:matching_istab_bstab_not_gstab}, below, shows a matching which is \abstable{$\alpha$}, but which is not \gstable{$\alpha$}. 
This example, together with \cref{prop:relation1}, also shows that \gstability{$\alpha$} is strictly stronger than \abstability{$\alpha$}. 

\begin{example}\label{ex:matching_istab_bstab_not_gstab} %\em
Consider an instance with six agents and two layers of preference profiles.

{\centering
\begin{tikzpicture}
  \node at (0, -0.5) (Layer1) {\Large $P_1$:};
  \foreach \i in {1, 2, 3}
  {
    \foreach \j/\p/\o in {u/0/left,w/1/right} {
      \node[draw, circle, minimum size=3ex, inner sep=1pt] at (\i*\dist, -\p*\xsc) (n\j\i) {$\j_\i$};
    %  \node[\o = 0pt of \j\i]  (n\j\i) {$\j_\i$}; 
    }
  }

  \foreach \n / \i / \o / \a/\b/\c   in {u1/w/left/1/2/3, u2/w/left/2/1/3, u3/w/left/3/1/2}{
  % \node[\o = -7pt of n\n] {:};
    \gettikzxy{(n\n)}{\xx}{\yy};
    \node at (\xx,\yy+\ysc*35) {$\i_\a$};
    \node at (\xx,\yy+\ysc*25) {$\i_\b$};
    \node at (\xx,\yy+\ysc*15) {$\i_\c$};
  }  
  \foreach \n / \i / \o / \a/\b/\c  in {%
    w1/u/right/2/1/3, w2/u/right/2/1/3, w3/u/right/3/1/2}{
  % \node[\o = -7pt of n\n] {:};
    \gettikzxy{(n\n)}{\xx}{\yy};
    \node at (\xx,\yy-\ysc*18) {$\i_\a$};
    \node at (\xx,\yy-\ysc*28) {$\i_\b$};
    \node at (\xx,\yy-\ysc*38) {$\i_\c$};
  }
  \foreach \s/\t in {u1/w1,u2/w2,u3/w3} {
    \draw[blackline] (n\s) -- (n\t);
  }
\end{tikzpicture}
~~\qquad
\begin{tikzpicture} 
  \node at (0, -0.5) (Layer1) {\Large $P_2$:};
  \foreach \i in {1, 2, 3}
  {
    \foreach \j/\p/\o in {u/0/left,w/1/right} {
      \node[draw, circle, minimum size=3ex, inner sep=1pt] at (\i*\dist, -\p*\xsc) (n\j\i) {$\j_\i$};
    %  \node[\o = 0pt of \j\i]  (n\j\i) {$\j_\i$}; 
    }
  }

  \foreach \n / \i / \o / \a/\b/\c   in {u1/w/left/2/1/3, u2/w/left/3/1/2, u3/w/left/1/3/2} {
  % \node[\o = -7pt of n\n] {:};
    \gettikzxy{(n\n)}{\xx}{\yy};
    \node at (\xx,\yy+\ysc*35) {$\i_\a$};
    \node at (\xx,\yy+\ysc*25) {$\i_\b$};
    \node at (\xx,\yy+\ysc*15) {$\i_\c$};
  }  
  \foreach \n / \i / \o / \a/\b/\c  in {%
     w1/u/right/3/1/2, w2/u/right/1/2/3, w3/u/right/2/1/3}{
  % \node[\o = -7pt of n\n] {:};
    \gettikzxy{(n\n)}{\xx}{\yy};
    \node at (\xx,\yy-\ysc*18) {$\i_\a$};
    \node at (\xx,\yy-\ysc*28) {$\i_\b$};
    \node at (\xx,\yy-\ysc*38) {$\i_\c$};
  }
 
  \foreach \s/\t in {u1/w2,u2/w3,u3/w1} {
    \draw[blackline] (n\s) -- (n\t);
  }
\end{tikzpicture}
\par}

\noindent
Observe that matching $M = \{ \{u_1, w_1\},  \{u_2, w_3\},  \{u_3, w_2\}\}$ is \aistable{$1$} and, thus \abstable{$1$}. 
However,~$M$ is blocked by pair $\{u_2, w_1\}$ in the first layer and by $\{u_1, w_2\}$ in the second. Thus,~$M$ is not \gstable{1}.
Indeed the only \gstable{$1$} matchings are indicated by the solid lines, 
which are also \aistable{$1$} (and thus \abstable{$1$}).

As soon as $\alpha\ge 2$, \abstability{$\alpha$} is not guaranteed to exist, 
even if  $\ell>\alpha$.
To see this we augment the instance with one more layer whose preference lists are identical to the first layer given in \cref{ex:first_example}.
One can verify that for each of all six possible matchings, there is always an unmatched pair that is blocking at least two layers.
\hfill $\diamond$
\end{example}

For $\alpha=1$, we observe that \abstability{$1$} is equivalent to \aistability{$1$}.

\begin{proposition}\label{prop:1-pair<=>1-individual}
  A matching is \abstable{$1$} if and only if it is \aistable{$1$}.
\end{proposition}

\begin{proof}
  By \cref{prop:relation2}, we know that \aistability{$1$} implies \abstability{$1$}.
  It remains to show the other direction.
  Let $M$ be a \abstable{$1$} matching.
  Suppose, for the sake of contradiction, that $M$ is not \aistable{$1$}.
  By \cref{prop:individual-stable-alternative-def}, this means that there is an unmatched pair~$\{u,w\}$ that is both $\ell-1+1=\ell$-dominating $\{u,M(u)\}$
  and $\ell-1+1=\ell$-dominating $\{w,M(w)\}$. 
  This implies that the pair~$\{u,w\}$ is indeed $\ell$-blocking $M$, which by \cref{prop:pair-stable-alternative-def}, is a contradiction to $M$ being \abstable{$1$}.
%such that none of both conditions stated in \cref{def:individual} holds.
  % Thus, there is \myemph{no} layer~$i$ such that $\{u, M(u)\}$ dominates $\{u,w\}$ in layer~$i$,
  % and there is \myemph{no} layer~$j$ such that $\{w,M(w)\}$ dominates $\{u,w\}$ in layer~$j$.
  % In other words, for \myemph{each} layer~$i$ it holds that both $\{u,M(u)\}$ and $\{w,M(w)\}$ \myemph{do not} dominate $\{u,w\}$ in layer $i$,
  % implying that $\{u,w\}$ is $\ell$ blocking $M$--a contradiction to $M$ being \abstable{$1$}.
\end{proof}

\cref{ex:matching_bstab_not_istab} shows that for $\alpha>1$, 
individual stability and pair stability are not equivalent, 
neither is individual stability equivalent to global stability.

\begin{example}\label{ex:matching_bstab_not_istab} %\em
%Consider the following instance with four agents and three layers of preference profiles.
Consider the example given in \cref{sec:intro}.
Recall that the first layer admits exactly one stable matching, namely $M_1=\{\{u_1,w_1\}, \{u_2, w_2\}\}$ (depicted by solid lines).
The second layer also admits exactly one (different) stable matching, namely $M_2=\{\{u_1,w_2\}, \{u_2,w_1\}\}$ (also depicted by solid lines).
The third layer has two stable matchings, $M_1$ and $M_2$ (depicted by solid lines and dashed lines, resp.).

Thus, both $M_1$ and $M_2$ are \gstable{$2$} (and \abstable{$2$}).
However, $M_1$ is \aistable{$2$} while $M_2$ is not.
To see why $M_2$ is not \aistable{$2$}, 
we can verify that the unmatched pair~$p=\{u_1, w_1\}$ dominates $\{u_1, w_2\}$ in the first and the third layer 
and it dominates $\{u_2,w_1\}$ in the first two layers.
By \cref{prop:individual-stable-alternative-def}, $M_2$ is not \aistable{$2$} since $\ell-\alpha+1=2$. 

If we restrict the example to the last two layers only, 
then matching~$M_2$ is also evidence that an \gstable{$\ell$} (which, by~\cref{prop:g_stability_equiv_pair_stability}, implies \abstability{$\ell$}) is not \aistable{$\ell$}.
\hfill $\diamond$
\end{example}

For $\alpha=\ell$ global stability and pair stability are equivalent.

\begin{proposition}\label{prop:g_stability_equiv_pair_stability}
For $\alpha=\ell$, a matching is \gstable{$\alpha$} if and only if it is \abstable{$\alpha$}.
\end{proposition}
\begin{proof}
By \cref{prop:relation1}, \gstability{$\ell$} implies \abstability{$\ell$}. Now, assume that a matching $M$ is \abstable{$\ell$}.
For the sake of contradiction, suppose that $M$ is not \gstable{$\ell$}. 
This means that there exists a pair, say $\{u, w\}$, and a layer, say $i$, such that $\{u, w\}$ is blocking in layer~$i$. Thus, $\{u, w\}$ is $1$-blocking~$M$, and so $M$ cannot satisfy \abstability{$\ell$}. This gives a contradiction.
\end{proof}

It is somehow counter-intuitive that even an \gstable{$\ell$} matching (\emph{i.e.}\ a matching that is stable in each layer) may not be \aistable{$\ell$} (see \cref{ex:matching_bstab_not_istab}). 
By \cref{prop:i_stability_stronger_g_stability} we can thus infer that \gstability{$\ell$} is strictly weaker than \aistability{$\ell$}.

\begin{proposition}\label{prop:i_stability_stronger_g_stability}
  For $\alpha=\ell$, an \aistable{$\alpha$} matching is \gstable{$\alpha$}.
\end{proposition}

\begin{proof}
  \cref{prop:relation2} and \cref{prop:g_stability_equiv_pair_stability} imply the statement since $\alpha=\ell$.
\end{proof}

By \cref{ex:matching_bstab_not_istab}, \gstability{$\ell$} does not imply \aistability{$\ell$}.
However, it implies \aistability{$\lceil\nicefrac{\ell}{2}\rceil$}.

\begin{proposition}\label{prop:lglobal->lhalfindividual}
  Every \gstable{$\ell$} matching is \aistable{$\lceil\nicefrac{\ell}{2}\rceil$}.
  There are instances where \gstable{$\ell$} matchings are not \aistable{$(\lceil\nicefrac{\ell}{2}\rceil+1)$}.
\end{proposition}

%\todo[inline]{RN: Jiehua wanted to check this proof again...\newline Hua: Checked and adjusted.}
\begin{proof}
  For the first statement, let $M$ be an \gstable{$\ell$} matching.
  Suppose, for the sake of contradiction, that $M$ is not \aistable{$\lceil\nicefrac{\ell}{2}\rceil$}.
  Let  $\beta=\ell-\lceil \nicefrac{\ell}{2} \rceil + 1$, which is $\lfloor\nicefrac{\ell}{2}\rfloor+1$.
  By \cref{prop:individual-stable-alternative-def}, let $\{u,w\}$ be an unmatched pair that is both $\beta$-dominating $\{u, M(u)\}$ and $\beta$-dominating $\{w, M(w)\}$.
  Since $2\cdot \beta > \ell$, there is at least one layer~$i$ where $\{u,w\}$ is dominating both $\{u, M(u)\}$ and $\{w, M(w)\}$, meaning that $\{u,w\}$ is blocking layer~$i$, a contradiction to $M$ being \gstable{$\ell$}.
  % Consider an arbitrary unmatched pair~$\{u,w\}$, 
  % and define $A=\{i\colon M(u) \pref^{(i)}_u w\}$ as the set of all layers in which $\{u,M(u)\}$ dominates $\{u, w\}$.
  % If $|A| \ge \lceil\nicefrac{\ell}{2}\rceil$, then $A$ fulfills one of the conditions stated in the definition of individual stability for $\alpha= \lceil\nicefrac{\ell}{2}\rceil$ (see \cref{def:individual}).
  % Otherwise, let $B=[\ell] \setminus A$. 
  % In each layer of $B$, we have that $\{u,w\}$ dominates $\{u,M(u)\}$.
  % By assumption, $M$ is stable in each layer, 
  % implying that in each layer of $B$ we have that $\{w,M(w)\}$ dominates $\{u,w\}$.
  % Since $|B| = \ell - |A| \ge \lceil\nicefrac{\ell}{2}\rceil$, the set~$B$ fulfills one of the conditions stated in the definition of individual stability for $\alpha= \lceil\nicefrac{\ell}{2}\rceil$.

  To see why \gstability{$\ell$} may not imply \aistability{$(\lceil\nicefrac{\ell}{2} \rceil +1)$}, consider the following instance with four agents and $\ell=4$~layers.

  \noindent
  {\centering
  \begin{tikzpicture} 
  \node at (0.1, -0.5) (Layer1) {\Large $P_1$:};
  \foreach \i in {1, 2}
  {
    \foreach \j/\p/\o in {u/0/left,w/1/right} {
      \node[draw, circle, minimum size=3ex, inner sep=1pt] at (\i*\dist, -\p*\xsc) (n\j\i) {$\j_\i$};
    %  \node[\o = 0pt of \j\i]  (n\j\i) {$\j_\i$}; 
    }
  }

  \foreach \n / \i / \o / \a/\b   in {u1/w/left/1/2, u2/w/left/1/2}{
  % \node[\o = -7pt of n\n] {:};
    \gettikzxy{(n\n)}{\xx}{\yy};
    \node at (\xx,\yy+\ysc*25) {$\i_\a$};
    \node at (\xx,\yy+\ysc*15) {$\i_\b$};
  }  
  \foreach \n / \i / \o / \a/\b  in {%
    w1/u/right/1/2, w2/u/right/1/2}{
  % \node[\o = -7pt of n\n] {:};
    \gettikzxy{(n\n)}{\xx}{\yy};
    \node at (\xx,\yy-\ysc*18) {$\i_\a$};
    \node at (\xx,\yy-\ysc*28) {$\i_\b$};
  }
  \foreach \s/\t in {u1/w1,u2/w2}{
    \draw[blackline] (n\s) -- (n\t);
  } 
\end{tikzpicture}
~~\qquad
\begin{tikzpicture}
  \node at (0.1, -0.5) (Layer1) {\Large $P_2$:};
  \foreach \i in {1, 2}
  {
    \foreach \j/\p/\o in {u/0/left,w/1/right} {
      \node[draw, circle, minimum size=3ex, inner sep=1pt] at (\i*\dist, -\p*\xsc) (n\j\i) {$\j_\i$};
    %  \node[\o = 0pt of \j\i]  (n\j\i) {$\j_\i$}; 
    }
  }

  \foreach \n / \i / \o / \a/\b   in {u1/w/left/1/2, u2/w/left/2/1}{
  % \node[\o = -7pt of n\n] {:};
    \gettikzxy{(n\n)}{\xx}{\yy};
    \node at (\xx,\yy+\ysc*25) {$\i_\a$};
    \node at (\xx,\yy+\ysc*15) {$\i_\b$};
  }  
  \foreach \n / \i / \o / \a/\b  in {%
    w1/u/right/1/2, w2/u/right/1/2}{
  % \node[\o = -7pt of n\n] {:};
    \gettikzxy{(n\n)}{\xx}{\yy};
    \node at (\xx,\yy-\ysc*18) {$\i_\a$};
    \node at (\xx,\yy-\ysc*28) {$\i_\b$};
  }
  \foreach \s/\t in {u1/w1,u2/w2}{
    \draw[blackline] (n\s) -- (n\t);
  } 
\end{tikzpicture}
~~\qquad
\begin{tikzpicture}
  \node at (0.1, -0.5) (Layer1) {\Large $P_3$:};
  \foreach \i in {1, 2}
  {
    \foreach \j/\p/\o in {u/0/left,w/1/right} {
      \node[draw, circle, minimum size=3ex, inner sep=1pt] at (\i*\dist, -\p*\xsc) (n\j\i) {$\j_\i$};
    %  \node[\o = 0pt of \j\i]  (n\j\i) {$\j_\i$}; 
    }
  }

  \foreach \n / \i / \o / \a/\b   in {u1/w/left/2/1, u2/w/left/2/1}{
  % \node[\o = -7pt of n\n] {:};
    \gettikzxy{(n\n)}{\xx}{\yy};
    \node at (\xx,\yy+\ysc*25) {$\i_\a$};
    \node at (\xx,\yy+\ysc*15) {$\i_\b$};
  }  
  \foreach \n / \i / \o / \a/\b  in {%
    w1/u/right/2/1, w2/u/right/2/1}{
  % \node[\o = -7pt of n\n] {:};
    \gettikzxy{(n\n)}{\xx}{\yy};
    \node at (\xx,\yy-\ysc*18) {$\i_\a$};
    \node at (\xx,\yy-\ysc*28) {$\i_\b$};
  }
  \foreach \s/\t in {u1/w1,u2/w2}{
    \draw[blackline] (n\s) -- (n\t);
  } 
\end{tikzpicture}
~~\qquad
\begin{tikzpicture}
  \node at (0.1, -0.5) (Layer1) {\Large $P_4$:};
  \foreach \i in {1, 2}
  {
    \foreach \j/\p/\o in {u/0/left,w/1/right} {
      \node[draw, circle, minimum size=3ex, inner sep=1pt] at (\i*\dist, -\p*\xsc) (n\j\i) {$\j_\i$};
    %  \node[\o = 0pt of \j\i]  (n\j\i) {$\j_\i$}; 
    }
  }

  \foreach \n / \i / \o / \a/\b   in {u1/w/left/2/1, u2/w/left/2/1}{
  % \node[\o = -7pt of n\n] {:};
    \gettikzxy{(n\n)}{\xx}{\yy};
    \node at (\xx,\yy+\ysc*25) {$\i_\a$};
    \node at (\xx,\yy+\ysc*15) {$\i_\b$};
  }  
  \foreach \n / \i / \o / \a/\b  in {%
    w1/u/right/1/2, w2/u/right/2/1}{
  % \node[\o = -7pt of n\n] {:};
    \gettikzxy{(n\n)}{\xx}{\yy};
    \node at (\xx,\yy-\ysc*18) {$\i_\a$};
    \node at (\xx,\yy-\ysc*28) {$\i_\b$};
  }
  \foreach \s/\t in {u1/w1,u2/w2}{
    \draw[blackline] (n\s) -- (n\t);
  } 
\end{tikzpicture}
\par}

\noindent \mbox{$M=\{\{u_1,w_1\}, \{u_2, w_2\}\}$} is the only \gstable{$4$} matching.
However, it is not \aistable{$3$} as the unmatched pair~$\{u_1, w_2\}$ dominates $\{u_1,w_1\}$ in layers $3$ and $4$
and dominates $\{u_2,w_2\}$ in layers~$1$ and $2$.
By \cref{prop:individual-stable-alternative-def}, $M$ is not \aistable{$3$} since $\ell-\alpha+1=2$.
\end{proof}

The relations among the different concepts of multi-layer stability are depicted in \cref{fig:relation_between_stability_concepts}.

\crefname{example}{Ex}{Examples}
\crefname{proposition}{Prop}{Props}

\begin{figure}
\centering
\begin{tikzpicture}[>=stealth']
  \def \scc {2.4}
  \def \ycc {1.2}
  \def \tcc {3.6}
  \node[concept] at (0*\scc,0*\scc) (1global) {$1$-global};
  \node[concept] at (1.5*\scc, -2*\ycc) (1pair) {$1$-pair};
  \node[concept] at (3*\scc, 0*\scc) (1individual) {$1$-individual};

  \node[concept,above = 2ex of 1global, xshift=-8ex] (aglobal) {$\alpha$-global};
  \node[concept,below = 3ex of 1pair] (apair) {$\alpha$-pair};
  \node[concept,above = 2ex of 1individual, xshift=8ex] (aindividual) {$\alpha$-individual};

%  \node[concept,above = 0ex of aindividual, xshift=4ex] (lhindividual) {$\lceil \nicefrac{\ell}{2} \rceil$-individual};

  \node[concept,above = 4ex of aglobal, xshift=-8ex] (lglobal) {$\ell$-global};
  \node[concept,below = 4ex of apair] (lpair) {$\ell$-pair};
  \node[concept,above = 4ex of aindividual, xshift = 8ex] (lindividual) {$\ell$-individual};

  \draw (1global) edge[->] node[midway,fill=white,inner sep=1pt,sloped, text width=5.5ex*\scc] {\footnotesize \cref{prop:relation1}, \cref{ex:matching_istab_bstab_not_gstab}} (1pair);
  \draw (1global) edge[->] node[midway,fill=white,inner sep=1pt,sloped, text width=6ex*\tcc] {\footnotesize \cref{prop:relation1,prop:1-pair<=>1-individual}, \cref{ex:matching_istab_bstab_not_gstab}} (1individual);
  \draw (1pair) edge[-,double,double distance=2pt] node[midway,fill=white,inner sep=1pt,sloped] {\footnotesize \cref{prop:1-pair<=>1-individual}} (1individual);

  \draw (aglobal) edge[->,bend right] node[midway,fill=white,inner sep=1pt,sloped] {\footnotesize \cref{prop:relation1}, \cref{ex:matching_istab_bstab_not_gstab}} (apair);
  \draw (aindividual) edge[->,bend left] node[midway,fill=white,inner sep=1pt,sloped] {\footnotesize \cref{prop:relation2}, \cref{ex:matching_bstab_not_istab}} (apair);

  \draw (lglobal) edge[-,double,double distance=2pt,bend right] node[midway,fill=white,inner sep=1pt,sloped] {\footnotesize \cref{prop:g_stability_equiv_pair_stability}} (lpair);
  \draw (lglobal) edge[<-] node[midway,fill=white,inner sep=1pt,sloped] {\footnotesize \cref{prop:i_stability_stronger_g_stability}, \cref{ex:matching_bstab_not_istab}} (lindividual);
  \draw (lindividual) edge[->,bend left] node[midway,fill=white,inner sep=1pt,sloped] {\footnotesize \cref{prop:i_stability_stronger_g_stability,prop:g_stability_equiv_pair_stability}, \cref{ex:matching_bstab_not_istab}} (lpair);

  \draw (lglobal)edge[->] node[midway,fill=white,inner sep=1pt,sloped] {\footnotesize \cref{prop:lglobal->lhalfindividual} for $\alpha\le \lceil\nicefrac{\ell}{2} \rceil$} (aindividual);

  \foreach \s / \t in {l/a, a/1} {
    \foreach \n in {global, individual, pair} {
      \draw (\s\n) edge[shorten <= -2pt, shorten >= -2pt, ->] (\t\n);
    }
  }
\end{tikzpicture}
\caption{Relations among the different multi-layer concepts of stability for different values of $\alpha$. Herein, an arc is to be read like an implication: one property implies the other.}\label{fig:relation_between_stability_concepts}
\end{figure}
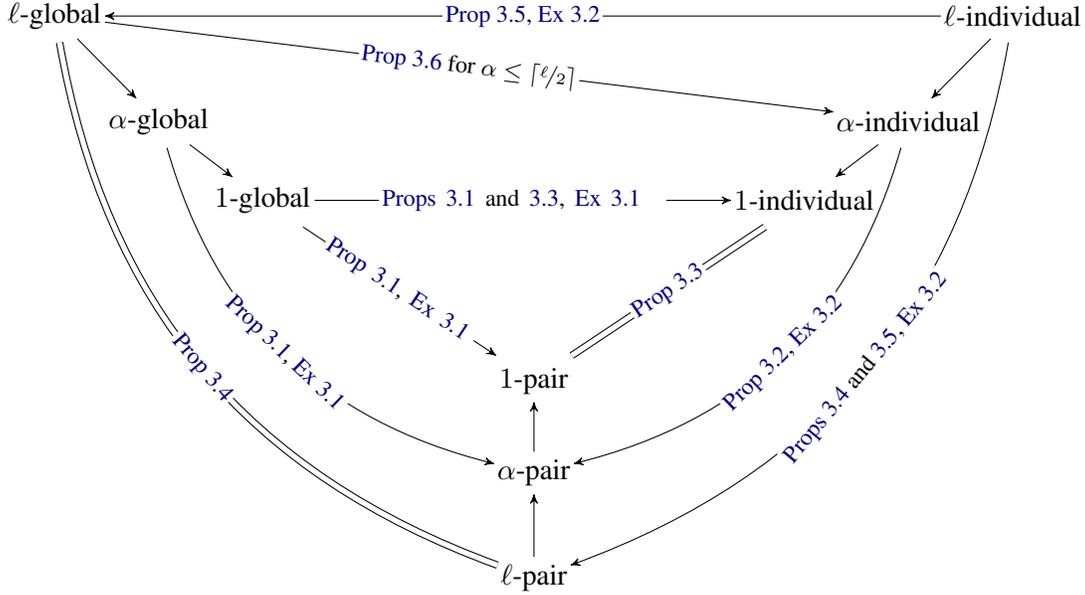

\crefname{example}{Example}{Examples}
\crefname{proposition}{Proposition}{Propositions}

A \gstable{$1$} matching always exists. 
Together with \cref{prop:relation1,prop:relation2}, we obtain the following.

\begin{proposition}\label{prop:1-stable-exists}
  A preference profile with $\ell$ layers always admits a 
  matching, which is \gstable{$1$}, \abstable{$1$}, and \aistable{$1$}.
\end{proposition}

\section{All-Layers Stability ($\alpha=\ell$)}
\label{sec:all-stab}
In this section, we discuss the special case when $\alpha=\ell$.
It turns out that deciding whether a given instance admits an \aistable{$\ell$} matching can be solved in polynomial time.
For the other two concepts of stability, however, the problem becomes NP-hard even when $\ell=2$.

\subsection{Algorithm for \aistability{$\ell$}}

The algorithm for deciding \aistability{$\ell$} is based on the following simple lemma.

\begin{lemma}\label{lem:ell-individual}
  Let $u\in U$ and $w\in W$ be two agents such that $w$ is the first ranked agent of $u$ in some layer~$i\in [\ell]$,
  and let $u'\in U\setminus \{u\}$ be another agent such that $w$ prefers $u$ over $u'$ in some layer~$j\in [\ell]$.
  Then, no \aistable{$\ell$} matching contains~$\{u',w\}$.
\end{lemma}

\begin{proof}
  Let $u,w,u'$ be the three agents and let $i,j$ be the two (possibly equal) layers as described by the assumption.
  Suppose towards a contradiction that there is an \aistable{$\ell$} matching with $\{u',w\}\in M$.
  This implies that $\{u,w\}$ is an unmatched pair under $M$.
  However, $w$ prefers $u$ over $u'=M(w)$ in layer~$j$ and $u$ prefers~$w$ over~$M(u)$ in layer~$i$---a contradiction to $M$ being \aistable{$\ell$}.
\end{proof}

\cref{lem:ell-individual} leads to \cref{alg:ell-individual} which looks 
quite similar to the so-called extended Gale-Shapley algorithm by~\citet{Irving1994}.
%However, we cannot use Irving's algorithm by aggregating the preference list into one as the transitivity may not be preserved. 
The crucial difference is that we loop into different layers and we cannot delete a pair~$p$ of agents that does not belong to any stable matching, as it may still serve to block certain matchings.
Instead of deleting such pair, we will \emph{mark} it. % pairs~$p$.
Herein, \myemph{marking a pair~$\{u,w\}$} means marking the agent~$u$ (resp.~$w$) in the preference list of $w$ (resp.~$u$) in every layer.

%\SetKwProg{Fn}{Function}{}{}
%\SetAlFnt{\sffamily}
\renewcommand\ArgSty{\normalfont}
%\renewcommand\KwSty[1]{\textnormal{\textbf{\sffamily#1}}\unskip}
%\SetAlCapFnt{\normalfont\sffamily\large}
%\renewcommand\AlCapNameFnt{\sffamily}

\begin{algorithm}[t]
  \DontPrintSemicolon
  \caption{Algorithm for finding an \aistable{$\ell$} matching.}
  \label{alg:ell-individual}
  \small
  \SetKwInOut{Input}{Input}
  \Input{\small A set of agents $U \cup W$ and $\ell$ layers of preferences.}
  \Repeat{(some agent's preference list consists of only marked agents)
    \textbf{or}
    (no new pair was marked in the last iteration)
  }
  {
    \ForEach{agent~$u \in U$\label{line:outer_foreach_loop}}
    {
      \ForEach{layer~$i=1,2,\ldots,\ell$\label{line:inner_foreach_loop}}
      {

       $w\leftarrow$ the first ranked agent in $u$'s preference list in layer~$i$\label{line:first_assignment_to_w}

        $r\leftarrow 1$
    
        \Repeat{$\{u,w\}$ is \myemph{not} marked\label{line:repeat_condition}}
        {
          \ForEach{$u'$ with $w\colon u\pref^{(j)}_{w} u'$ for some layer~$j$\label{alg:marking-condition}}{
            \textbf{mark} $\{u',w\}$\label{alg:marking}
          }
          
          $r\leftarrow r+1$

          $w\leftarrow$ the $r^{\text{th}}$ ranked agent in $u$'s preference list in layer~$i$
        }
      }
    }
  }

  \lIf{some agent's preference list consists of only marked agents} 
  {no \aistable{$\ell$} matching exists}
  \lElse{
    \textbf{return} $M=\{\{u,w\}\mid w \leftarrow \text{the first unmarked agent in } u\text{'s}$ preference list in any layer$\}$ as an \aistable{$\ell$} matching\label{line:matching_construction}
  }
\end{algorithm}

The correctness of \cref{alg:ell-individual} follows from \cref{lem:marked-pair-not-ell-stab,lem:unmarked-the-same,lem:M-ell-individual}.

\begin{lemma}\label{lem:marked-pair-not-ell-stab}
  If a pair~$\{u',w\}$ is marked during the execution of \cref{alg:ell-individual}, then no \aistable{$\ell$} matching
  contains this pair.
\end{lemma}

\begin{proof}
  Each pair is marked within two ``foreach'' loops in \Cref{line:outer_foreach_loop} and  \Cref{line:inner_foreach_loop}, respectively (we will refer to them as the ``outer'' loop and the ``inner'' loop). Let us fix an arbitrary $u \in U$ and $i \in [\ell]$ and consider the pairs which was marked when the outer and the inner loops were run for $u$ and $i$, respectively.
  We show the statement via induction on the sequence of pairs which were marked when $u$ and $i$ was considered for the two loops.
  For the induction to begin, let $\{u',w\}$ with $u' \in U$ and $w\in W$ be the first pair that is marked during the execution.
  This implies that agent~$u$ ranks $w$ in the first position in layer~$i$ and that 
  $w$ prefers $u$ to~$u'$ in some (possibly different) layer.
  By \cref{lem:ell-individual}, no \aistable{$\ell$} matching contains~$\{u',w\}$.

  For the induction assumption, let $\{u',w\}$ be the $m^{\text{th}}$ pair (for given $u$ and $i$) that is marked during the execution and no \aistable{$\ell$} matching contains a pair that is marked prior to $\{u',w\}$.
  Suppose for the sake of contradiction that there is an \aistable{$\ell$} matching~$M$ which contains the marked pair~$\{u',w\}$.
  The fact that $\{u',w\}$ has been marked implies that
  \begin{compactenum}
    \item $u$ ranks $w$ in the $p^{\text{th}}$ position in layer~$i$ for some $p$,
    and
    \item $w$ prefers $u$ over $u'$ in some layer~$j$.
  \end{compactenum}
  However, by the description of the algorithm in layer~$i$ (Lines~\ref{line:first_assignment_to_w}--\ref{line:repeat_condition}) for each agent~$w'$ that $u$ prefers to $w$ in layer~$i$,
  \emph{i.e.}\ $w'\pref^{(i)}_{u} w$,
  we have that $\{u,w'\}$ is marked (see the ``until'' condition in \Cref{line:repeat_condition}).
  The induction assumption implies that $M$ does not contain any $\{u,w'\}$ with $w'\pref^{(i)}_{u} w$.
  Thus, it follows that $u$ prefers~$w$ to $M(u)$ in layer~$i$, \emph{i.e.}\ $w\pref^{(i)}_{u}M(u)$.
  This is a contradiction to $M$ being \aistable{$\ell$} on the unmatched pair~$\{u,w\}$ since there is a layer~$j\in[\ell]$ such that $u\pref^{(j)}_{w}u'=M(w)$.
\end{proof}

The following lemma ensures that in \Cref{line:matching_construction} 
if $w$ is matched to an agent~$u$,
then it is the most preferred unmarked agent of $u$ in \myemph{all} layers.

\begin{lemma}\label{lem:unmarked-the-same}
   If no agent's preference list consists of only marked agents
   and there is an agent~$u$ and two different layers~$i,j\in [\ell]$, $i\neq j$ 
   such that
   the first unmarked agent in the preference list of $u$ in layer~$i$
   differs from the one in layer~$j$,
   then \cref{alg:ell-individual} will mark at least one more pair.
\end{lemma}

\begin{proof}
  Suppose towards a contradiction that no new pair is marked, 
  but there is an agent~$u\in U$ 
  such that the first agent unmarked by~$u$ is different in different layers, say $w$ and $w'$ in layers~$i$ and $j$, with $w\neq w'$ and $i\neq j$.
  Since no new pair will be marked, $u$ is the last unmarked agent in the preference lists of $w$ and $w'$ in all layers (see \Cref{alg:marking} of \cref{alg:ell-individual}).
  Since $|U|=|W|$ there is a different agent~$u'\in U\setminus \{u\}$ 
  such that for each agent~$w\in W$ we have that $u'$ is \myemph{not} the last unmarked agent in the preference list of $w$ in any layer.
%that for each agent $w \in W$ is \myemph{not} the last unmarked agent in the preference list of $w$ in all layers. 
  Since no agent's preference list consists of only marked agents, 
  the preference list of $u'$ (in some layer) contains an agent which is unmarked.
  Denote this agent as $w'$.
 % Let $w''$ be the first unmarked agent in the preference list of $u'$ in some layer.
  Again, since no new pair will be marked, 
  in the preference list of $w''$, agent~$u'$ is the last unmarked agent---a contradiction.
\end{proof}

%TODO: prove that for the agent who has two different agents on top in two different layers, at least one pair will be marked (move this argument from the proof of Theorem 4.4).

When \cref{alg:ell-individual} terminates and no agent contains a preference list that consists of only marked agents,
then we can construct an \aistable{$\ell$} matching by assigning to each agent~$u$ its first unmarked agent in any preference list (note that \cref{lem:unmarked-the-same} ensures on termination that for each agent it holds that in 
all layers is first unmarked agent is the same). %has the same first unmarked agent on termination).
%then we can infer that for each agent~$u$, the first unmarked agent of $u$ remains the same in each layer.

\begin{lemma}\label{lem:M-ell-individual}
  If upon termination no agent's preference list consists of only marked agents, 
  then the matching~$M$ computed by \cref{alg:ell-individual} 
  is \aistable{$\ell$}.
\end{lemma}

\begin{proof}
  Towards a contradiction suppose that $M$ returned by \cref{alg:ell-individual} is not \aistable{$\ell$}.
  That is, there is an unmatched pair~$\{u,w\}\in (\pairset)\setminus M$ with $u\in U$ and $w\in W$
  and two layers~$i, j\in [\ell]$ such that \mbox{$u\colon w \pref^{(i)}_{u} M(u)$} and \mbox{$w\colon u \pref^{(j)}_{w} M(w)$}.
  Observe that agent~$u$ is matched with the first agent, denoted as~$x$, in the preference list of $u$ such that the pair $\{u,x\}$ is not marked, and so we infer that $\{u, w\}$ is marked.
  Thus, the innermost loop of the algorithm has been run for the pair $\{u, w\}$ (see \Cref{line:repeat_condition}).
  By \Cref{alg:marking} of \cref{alg:ell-individual}
  for all agents~$u'$ where $w\colon u\pref^{(j)}_{w} u'$ for some layer~$j'\in [\ell]$, 
  the pair~$\{u',w\}$ is marked.
  This includes the pair~$\{M(w), w\}$ since $w\colon u \pref^{(j)}_{w} M(w)$---a contradiction to \cref{lem:marked-pair-not-ell-stab}.
\end{proof}

Finally, we obtain that \cref{alg:ell-individual} computes an \aistable{$\ell$}
 matching if one exists.

\begin{theorem}\label{thm:l_individual_polynomial}
  % For every instance~$I=(U,W,P_1,P_2,\dots,P_{\ell})$,
  % \cref{alg:ell-individual} decides whether $I$ admits an \aistable{$\ell$} matching in $O(\ell \cdot n^2)$~time.
  For $\alpha=\ell$, \cref{alg:ell-individual} solves \ISM in $O(\ell \cdot n^2)$~time. %where the number of agent is $2\cdot n$.
\end{theorem}

\begin{proof}
  Let $I=(U,W,P_1,P_2,\ldots,P_{\ell})$ with $2\cdot n$ agents be the input of  \cref{alg:ell-individual}. 
  By \cref{lem:marked-pair-not-ell-stab}, no \aistable{$\ell$} matching contains a marked pair.
  If there is an agent whose preference list consists of only marked agents, 
  then we can immediately conclude that the given instance is a no-instance.

  Otherwise, \cref{lem:M-ell-individual} proves that the algorithm returns an \aistable{$\ell$} matching.

  It remains to show that the algorithm terminates and has running time~$O(\ell \cdot n^2)$.
  % To show that the algorithm terminates we first show that if 
  % no agent's preference list consists of only marked agents and
  % there exists an agent who has two different agents that are the first unmarked ones in two different layers, 
  % then we will be able to mark some new pairs.
  Since there are in total $O(n^2)$ pairs, the algorithm will eventually terminate, either because some agent's preference list consists of only marked agents or because 
  no new new pair will be marked. %for each agent~$u$ the first agent in the preference lists of $u$ remains the same in all layers. 

  By using a list that points to the first unmarked agent of each agent~$u$ in each layer,
  and by using a table that stores pairs which are already marked and reconsidered by \Cref{alg:marking}, 
  the algorithm needs to ``touch'' each pair at most twice, once when it is not yet marked and a second time when it is already marked.
\end{proof}

% \todo[inline]{Sections 4.2. and following contain many technical hardness 
% proofs and almost all the time I missed in the beginning (before thetechnical proof starts) an intuitive explanation about what the idea is, why the proof is hard, wh the problem to reduce from is particularly suitable, etc.\newline
% Hua: I could take care of this later (after Piotr)}

\subsection{NP-hardness for \gstability{$\ell$} and \abstability{$\ell$}}
In contrast to \aistability{$\ell$}, in this section we show that deciding \GSM is NP-hard as soon as $\ell=2$.
We establish this by reducing the NP-complete 3-SAT problem~\cite{GJ79} to the decision variant of \GSM.
  The idea behind this reduction is to introduce for each variable four agents that admit exactly two possible globally stable matchings, one for each truth value.
  Then, we construct a satisfaction gadget for each clause by introducing six agents.
  These agents will have three possible globally stable matchings.
  We use a layer for each literal contained in the clause to enforce that setting the literal to false will exclude exactly one of the three globally stable matchings.
  Therefore, unless one of the literals in the clause is set to true, no globally stable matching remains.

  Using the above idea, we can already show hardness of deciding \gstability{$\ell$} for $\ell=3$.
  With some tweaks and using a restricted variant of 3-SAT~\cite{GJ79}  (see \cref{lem:np-hard-3SAT-variant}), we can strengthen our hardness result to hold even for $\ell=2$.
  From here on, we call a clause \myemph{monotone} if the contained literals are either all positive or all negative.

\begin{lemma}\label{lem:np-hard-3SAT-variant}
  3-SAT is NP-hard even if each clause has either two or three literals,  and \myemph{no} size-three clause is monotone
  while all size-two clauses are monotone
\end{lemma}

\begin{proof} 
  % Herein, the additional constraint is that each clause has either two or three literals, 
  % and \myemph{no} size-three clause is monotone
  % while all size-two clauses are monotone; a clause is \myemph{monotone}
  % if the contained literals are either all positive or all negative.
  % It is easy to verify that this variant is also $\np$-hard 
  % as we can
  We start with a 3-SAT instance and do the following.
  For each variable~$x_i$ introduce a helper variable~$z_i$,
  and make sure that the helper variable~$z_i$ is set to false if and only if the original variable~$x_i$ is set to true.
  To achieve this, 
  we add to the instance two new clauses $(x_i\vee z_i)$ and $(\overline{x}_i \vee \overline{z_i})$.
  Finally, for each original clause (note that it has size three) that contains only positive literals (resp.\ only negative literals), say~$x_i \vee x_j \vee x_k$ (resp.\ $\overline{x}_i \vee \overline{x}_j \vee \overline{x}_k$), we replace an arbitrary literal, say~$x_i$ (resp.\ $\overline{x}_i$), with $\overline{z}_i$ (resp.\ $z_i$).
  Observe that in the new instance, each original clause has size three and contains at least one negative and at least one positive literal,
  and that the newly introduced clauses have size two and are monotone.
  It is straightforward to see that the original instance is a yes-instance if and only if the new instance is a yes-instance.
\end{proof}
  
 Now, we are ready to present one of our main results.
%  In order to construct a matching that is globally stable in each layer, 
%  we have to make sure that some specific combination of pairs are not included in the matching.

% This enforcement helps in encoding an instance of the NP-complete 3SAT problem
% as an instance of \GSM.

\begin{theorem}\label{thm:global-hard-alpha=ell=constant}
%  For $\alpha=\ell=2$, deciding \gstability{$\alpha$} is $\np$-hard.
  \GSM is NP-hard even if $\alpha=\ell=2$.
\end{theorem}

\begin{proof}
  We provide a polynomial-time reduction from
  an NP-complete restricted variant of 3-SAT as given by \cref{lem:np-hard-3SAT-variant} to the decision 
  version of \GSM.
  Further, without loss of generality we assume that no clause contains two literals of the form $x$ and $\overline{x}$ as it will be satisfied anyway and can be ignored from the input instance.

  Let $(X, \mathcal{C})$ be an instance of the aforementioned 3-SAT variant with $X=\{x_1,x_2,\ldots, x_n\}$ being the set of variables and $\mathcal{C}=\{C_1,C_2,\ldots, C_m\}$ being the set of clauses of size at most three each.
  To unify the expression, for each size-three clause~$C_j=\ell^1_j\vee \ell^2_j \vee \ell^3_j$
  we order the literals so that the first literal is positive and the second one is negative.
  For each size-two clause~$C_j= \ell^{1}_j \vee \ell^{2}_j$ (note that it is monotone),
   we order the literals arbitrarily, and
   we call it a \myemph{positive} clause if it contains only positive literals, otherwise we call it a \myemph{negative} clause.

 % and call the literal with the smallest index the first literal,
 % the literal with the second smallest index the second literal,
 % and the remaining literal the third (or the last) literal.

  For each variable~$x_i\in X$, we create four \myemph{variable agents}~$x_i, \overline{x}_i, y_i, \overline{y}_i$ (we will make it clear when using $x_i$ and $\overline{x}_i$ whether we are referring to the literals or the variable agents).
  We will construct the preference lists of the variable agents so that each globally stable matching contains either $M^{\text{true}}_i\coloneqq \{\{x_i, y_i\}, \{\overline{x}_i, \overline{y}_i\}\}$ or $M^{\text{false}}_i\coloneqq \{\{x_i, \overline{y}_i\}, \{\overline{x}_i, y_i\}\}$.
  Briefly put, using~$M^{\text{true}}_i$ and~$M^{\text{false}}_i$ will correspond to setting the variable~$x_i$ to true or false, respectively.

  For each clause~$C_j\in \mathcal{C}$, we create six \myemph{clause agents}~$a_j,b_j,c_j,d_j,e_j,f_j$.
  We will construct preference lists for these clause agents so that 
for each size-three-clause, there are exactly three different ways in which these agents are matched in a globally stable matching,
and for each size-two-clause, there are exactly two such ways.
We use two layers to enforce that setting a different literal contained in the clause~$C_j$ to false excludes exactly one of these ways.
 % The first layer is used to encode the first two literals in a clause while the second layer is reserved for the last literal.

  In total, we have $4n+6m$ agents and we divide them into two groups~$U$ and $W$ with $U=\{x_i, \overline{x}_i\mid i \in [n]\}\} \cup \{a_i, b_i, c_j \mid j \in [m]\}$ and $W=\{y_i, \overline{y}_i \mid i \in [n]\}\cup \{d_j, e_j, f_j \mid j \in [m]\}$.

\paragraph{Preference lists of the variable agents.} 
  The preference lists of the variable agents restricted to the variable agents have the same pattern. We use the symbol~``$\cdots$'' to denote some arbitrary order of the remaining agents (that is, agents which were not yet explicitly mentioned in the preference list).
    \begin{alignat*}{4}
    &  \text{Layer~$(1)\colon \forall i \in [n]$:} \qquad
       & x_i\colon& y_i \pref {\color{winered}D_i} \pref \overline{y}_i \pref \cdots, &\quad\quad  y_i\colon& \overline{x}_i \pref {\color{winered}A_i} \pref x_i \pref \cdots,\\
    && \overline{x}_i \colon& \overline{y}_i \pref y_i \pref \cdots,  &  \overline{y}_i\colon& x_i\pref \overline{x}_i  \pref \cdots\text{.}\\[1ex]
    & \text{Layer~$(2)\colon \forall i \in [n]$:}\qquad
     &   x_i\colon& \overline{y}_i \pref  y_i \pref \cdots, &\quad\quad  y_i\colon& x_i \pref \overline{x}_i \pref \cdots,\\
    && \overline{x}_i \colon& y_i \pref {\color{winered}D'_i}\pref \overline{y}_i  \pref \cdots,  &  \overline{y}_i\colon& \overline{x}_i\pref {\color{winered}A'_i} \pref x_i  \pref \cdots\text{.}
\end{alignat*}
It remains to specify the meaning of symbols~$A_i, D_i, A'_i$, and $D'_i$.
\noindent
\begin{description}
\item[{\color{winered}$\boldsymbol{A_i}$}] denotes a list (in an arbitrary order) of all clause agents~$a_j$ that satisfy either of the following conditions:
\begin{compactenum}[(a)]
  \item $a_j$ corresponds to a \myemph{size-three-clause}~$C_j$ such that the \myemph{second} literal of clause~$C_j$ (which is a negative literal) is~$\overline{x}_i$, or
  \item it corresponds to a \myemph{negative size-two-clause}~$C_j$ such that the \myemph{first} literal of clause~$C_j$ is $\overline{x}_i$.
\end{compactenum}
\item[{\color{winered}$\boldsymbol{D_i}$}] denotes a list (in an arbitrary order) of all clause agents~$d_j$ 
such that the first literal of clause~$C_j$ is $x_i$ (note that in this case $C_j$ has either three literals or exactly two positive literals). 
% \begin{compactenum}
%   \item $d_j$ corresponds to a \myemph{size-three} clause the \myemph{first} literal of $C_j$ (which is a positive literal) is $x_i$, or
%   \item it corresponds to a \myemph{positive size-two} clause~$C_j$ such that the \myemph{first} literal of $C_j$ is $x_i$.
% \end{compactenum}
\item[{\color{winered}$\boldsymbol{A'_i}$}] denotes a list (in an arbitrary order) of all clause agents~$a_j$ such that the last literal of $C_j$ is $x_i$ (note that in this case $C_j$ has either three literals or exactly two positive literals). 
\item[{\color{winered}$\boldsymbol{D'_i}$}] denotes a list (in an arbitrary order) of all clause agents~$d_j$ such that 
the last literal of $C_j$ is $\overline{x}_i$ (note that in this case  $C_j$ has either three literals or exactly two negative literals). 
\end{description}

\noindent To illustrate the above notation, suppose that variable $x_i$ appears in four size-three-clauses, call them $C_1,C_2,C_3$, and $C_5$, and in two size-two-clauses: $C_4$ and $C_6$.
The positive literal~$x_i$ is the first literal in clauses~$C_1$, $C_3$, and $C_4$.
The negative literal~$\overline{x}_i$ is the second literal in $C_2$,
and the last literal in $C_5$ and $C_6$.
In this case, $A_i=a_2$, $D_i$ could be $D_i=d_1 \pref d_3 \pref d_4$, $A'_i$ is empty, and  $D'_i$ could be $D'_i=d_5\pref d_6$.

\paragraph{Preference lists of the clause agents.} 
The preference lists for the clause agents in the first layer are ``fixed'' when restricted to the clause agents; they only differ in the positions of variable agents. For a clause $C_j$ and an integer $t \in \{1, 2, 3\}$ let $C_j^{(t)}$ denote the $t$-th literal in $C_j$ (there will be no $C_j^{(3)}$ if $C_j$ has two literals). For a literal $\ell_i$ which is $x_i$ or $\overline{x}_i$, by $X(\ell_i)$, $Y(\ell_i)$, $\overline{X}(\ell_i)$, and $\overline{Y}(\ell_i)$, we denote the variable agents $x_i$, $y_i$, $\overline{x_i}$, and~$\overline{y_i}$, respectively, all corresponding to variable~$x_i$. 
For instance, for a clause~$C_j = x_2 \vee \overline{x}_4 \vee x_5$, we have that $\overline{Y}(C_j^{(1)}) = \overline{y}_2$, and $\overline{Y}(C_j^{(2)}) = \overline{y}_4$.
\begin{alignat*}{3}
    \intertext{Layer~$(1), \forall j \in [m]$:}
  &  \text{$|C_j|=3$:~~}
    & a_j\colon& d_j\pref e_j \pref {\color{winered} Y(C_j^{(2)})} \pref f_j\pref \cdots, & \quad d_j\colon & b_j \pref c_j \pref {\color{winered}X(C_j^{(1)})} \pref a_j \pref \cdots,\\
    && b_j\colon& e_j\pref f_j \pref d_j \pref \cdots, & \quad e_j\colon & c_j \pref a_j \pref b_j \pref \cdots,\\
&& c_j\colon& f_j\pref  d_j \pref e_j \pref \cdots, & \quad f_j\colon & a_j \pref b_j \pref c_j\pref \cdots,\\
&\text{$|C_j|=2$ and $C_j$ is \myemph{positive}:}
    & a_j\colon& d_j\pref e_j\pref \cdots, & \quad d_j\colon & b_j \pref {\color{winered}X(C_j^{(1)})} \pref a_j \pref \cdots,\\
    && b_j\colon& e_j\pref d_j \pref \cdots, & \quad e_j\colon & a_j \pref b_j \pref \cdots,\\
    && c_j\colon& f_j\pref  \cdots, & \quad f_j\colon &  c_j\pref \cdots,\\
&\text{$|C_j|=2$ and $C_j$ is \myemph{negative}:}~~
    & a_j\colon& d_j\pref {\color{winered}Y(C_j^{(1)})} \pref e_j\pref \cdots, & \quad d_j\colon & b_j \pref a_j \pref \cdots,\\
    && b_j\colon& e_j\pref d_j \pref \cdots, & \quad e_j\colon & a_j \pref b_j \pref \cdots,\\
    && c_j\colon& f_j\pref  \cdots, & \quad f_j\colon &  c_j\pref \cdots.
\end{alignat*}

\noindent
The preference lists for the second layer depends on the ``positiveness'' of the last literal.
There are two variants:
\begin{alignat*}{3}
\intertext{Layer $(2), \forall j \in [m]$ with $|C_j|=3$:}
&\text{Variant $1$ ($C_j^{(3)}$ is \myemph{positive})}\colon~
  &a_j\colon& f_j\pref d_j \pref {\color{winered}\overline{Y}(C_j^{(3)})} \pref e_j\pref \cdots, & \quad d_j\colon & c_j \pref a_j \pref b_j \pref \cdots,\\
   && b_j\colon& d_j\pref e_j \pref f_j \pref \cdots, & \quad e_j\colon & a_j \pref b_j \pref c_j \pref \cdots,\\ 
 && c_j\colon& e_j\pref  f_j  \pref d_j \pref \cdots, & \quad f_j\colon & b_j \pref c_j \pref a_j\pref \cdots,\\
&\text{Variant $2$ ($C_j^{(3)}$ is \myemph{negative})}\colon~
    & a_j\colon& e_j \pref f_j \pref d_j\pref \cdots, & \quad d_j\colon & a_j \pref b_j \pref {\color{winered}\overline{X}(C_j^{(3)})} \pref c_j \pref \cdots,\\
    && b_j\colon& f_j\pref d_j \pref e_j \pref \cdots, & \quad e_j\colon & b_j \pref c_j \pref a_j \pref \cdots,\\
&& c_j\colon& d_j \pref e_j \pref f_j \pref \cdots, & \quad f_j\colon & c_j \pref a_j \pref b_j \pref \cdots,\\
\intertext{Layer $(2), \forall j \in [m]$ with $|C_j|=2$:}
&\text{Variant $1$ ($C_j$ is \myemph{positive})}\colon~
& a_j\colon& d_j \pref {\color{winered}\overline{Y}(C_j^{(2)})} \pref e_j\pref \cdots, & \quad d_j\colon &  b_j \pref a_j \pref \cdots,\\
   && b_j\colon& e_j\pref d_j \cdots, & \quad e_j\colon & a_j \pref b_j \cdots,\\
   && c_j\colon&   f_j  \pref \cdots, & \quad f_j\colon & c_j\pref \cdots,\\
&\text{Variant $2$ ($C_j$ is \myemph{negative})}
& a_j\colon& d_j \pref e_j\pref \cdots, & \quad d_j\colon &  b_j \pref {\color{winered} \overline{X}(C_j^{(2)})}\pref a_j \pref \cdots,\\
   && b_j\colon& e_j\pref d_j \cdots, & \quad e_j\colon & a_j \pref b_j \cdots,\\
   && c_j\colon& f_j \pref \cdots, & \quad f_j\colon & c_j \pref \cdots.
\end{alignat*}

\noindent
This completes the construction which can be done in polynomial time.

Before we show the correctness of our construction, we first discuss some properties that each \gstable{$2$} matching~$M$ must satisfy.
\begin{claim}\label{claim:variable_gadget}
  Let $M$ be a \gstable{$2$} matching for our two-layer preference profiles.
  For each variable~$x_i \in X$, it holds that either $M^{\text{true}}_{i}\subseteq M$ or $M^{\text{false}}_i \subseteq M$.
\end{claim}
\begin{proof}\renewcommand{\qedsymbol}{(of
      \cref{claim:variable_gadget})~$\diamond$}
  To see this, we distinguish between two cases, depending on whether the partner of $\overline{x}_i$, $M(\overline{x}_i)$, is $\overline{y}_i$ or not.
% let $M(\overline{x}_i)$ and $M(\overline{y}_i)$ be the partners of $\overline{x}_i$ and $\overline{y}_i$ under $M$, respectively.
  If $M(\overline{x}_i)\neq \overline{y}_i$, then by the stability of $M$ for the first layer, 
  it follows that $\overline{y}_i$ prefers its partner~$M(\overline{y}_i)$ to $\overline{x}_i$ in the first layer as otherwise $\overline{x}_i$ and $\overline{y}_i$ are forming a blocking pair for the first layer.
  Since $x_i$  is the only agent that $\overline{y}_i$ prefers to~$\overline{x}_i$ in this layer,
  we have that $M(\overline{y}_i)=x_i$.
  Then, it must hold that $M(\overline{x}_i)=y_i$ as otherwise $\overline{x}_i$ and $y_i$ would block $M$ in the first layer. %both layers).
  Thus, $M^{\text{false}}_{i}\subseteq M$.

  Similarly, if $M(\overline{x}_i)=\overline{y}_i$, then by the stability of~$M$ 
and by construction of the preference lists of $x_i$ and $y_i$ in the second layer we must have that $M(x_i)=y_i$. 
This leads to $M^{\text{true}}_{i}\subseteq M$.
\end{proof}

% that is globally stable in all six layers. % \gstable{$\ell$} matching~$M$. 
We obtain a similar result for the clause agents. 
For each clause~$C_j\in \mathcal{C}$ with $|C_j|=3$,
let $N^{1}_j=\{\{a_j,d_j\}, \{b_j,e_j\}, \{c_j,f_j\}\}$, 
$N^2_j=\{\{a_j,f_j\}, \{b_j,d_j\},\{c_j,e_j\}\}$, 
$N^{3}_j=\{\{a_j,e_j\}, \{b_j,f_j\}, \{c_j,d_j\}\}$.
\begin{claim}\label{claim:clause_gadget}
  Let~$C_j\in \mathcal{C}$ be a size-three-clause, and
  let $x_i$ be a variable that appears (as either a positive or a negative literal) in $C_j$.
  For a \gstable{$2$} matching $M$ the following conditions hold:
  \begin{compactenum}[(i)]
    \item If $x_i$ is the first literal in $C_j$ and if $M^{\text{false}}_{i}\subseteq M$, 
    then either $N^{2}_j\subseteq M$ or $N^{3}_j\subseteq M$.
    \item If $\overline{x}_i$ is the second literal in $C_j$ and if $M^{\text{true}}_{i}\subseteq M$, 
    then either $N^{1}_j\subseteq M$ or $N^{3}_j\subseteq M$.
    \item If $x_i$ is the third literal in $C_j$ and if $M^{\text{false}}_{i}\subseteq M$,
    then either $N^{1}_j\subseteq M$ or $N^{2}_j\subseteq M$.
    \item If $\overline{x}_i$ is the third literal in $C_j$ and if $M^{\text{true}}_{i}\subseteq M$,
    then either $N^{1}_j\subseteq M$ or $N^{2}_j\subseteq M$.
  \end{compactenum}
\end{claim}
\begin{proof}
\renewcommand{\qedsymbol}{(of
      \cref{claim:clause_gadget})~$\diamond$}
We consider the four cases separately:
    \begin{compactenum}[(i)]
    \item Assume that $x_i$ is the first literal in $C_j$ and $M^{\text{false}}_{i}\subseteq M$.
    This implies that $\{x_i, \overline{y}_i\}\in M$. Consider 
    the preference list of $x_i$ in the first layer, and observe
    that $d_j$ appears in $D_i$. 
    Since $x_i$ prefers~$d_j$ to its partner~$\overline{y}_i$ in the first layer,
    it follows that $d_j$ must obtain a partner that it prefers to~$x_j$ in the first layer.
    By the preference list of $d_j$ in the first layer, we have that $M(d_j) \in \{b_j, c_j\}$.
    If $M(d_j)=b_j$, then by the preference list of $b_j$ in the first layer
    it follows that both $e_j$ and $f_j$ must obtain partners that they find better 
    than $b_j$ in the first layer.
    This means that $M(f_j)=a_j$ and $M(e_j) \in \{a_j, c_j\}$, implying that $M(e_j) = c_j$.
    Analogously, if $M(d_j)=c_j$, 
    then $\{a_j, e_j\}\in M$ as otherwise they will block the first layer 
    since the most preferred agents of both $a_j$ and $e_j$ are already assigned to someone else,
    and $a_j$ and $e_j$ are each other's second most preferred agents.
    Then, $b_j$ must obtain a partner that it prefers to~$d_j$.
    Since $f_j$ is the only agent left that $b_j$ prefers to~$d_j$,
    we get that $M(b_j)=f_j$, and so $N^3_j\subseteq M$.
    
    \item Assume that $\overline{x}_i$ is the second literal in $C_j$ and $M^{\text{true}}_{i}\subseteq M$ (thus, in particular, $\{x_i, y_i\}\in M$).
    Since $a_j$ appears in $A_i$ in the preference list of $y_i$ in the first layer, we infer
    that $y_i$ prefers~$a_j$ to its partner~$x_i$ in the first layer.
    Thus, it follows that $a_j$ must obtain a partner that it prefers to~$y_i$ in the first layer, \emph{i.e.}\ that $M(a_j) \in \{d_j, e_j\}$.
    If $M(a_j)=d_j$, then by considering the preference list of $d_j$ in the first layer
    we infer that both $b_j$ and $c_j$ obtain partners that they find better 
    than $d_j$ in the first layer.
    This means that $M(c_j)=f_j$, and $M(b_j) \in \{e_j, f_j\}$. Since $f_j$ is taken by $c_j$, we get that $M(b_j) = e_j$.
    Together, we have that $N^{1}_j \subseteq M$.

    Analogously, if $M(a_j)=e_j$, 
    then $\{b_j, f_j\}\in M$ as otherwise they would block the first layer
    since the most preferred agents of both $b_j$ and $f_j$ are already assigned to someone else,
    and $b_j$ and $f_j$ are each other's second most preferred agent. 
    Moreover, $c_j$ must obtain a partner that it prefers to~$e_j$.
    Since $d_j$ is the only agent left that $c_j$ prefers to~$e_j$,
    we obtain that $M(c_j)=d_j$.
    This leads to $N^3_j\subseteq M$.

    \item  Assume that $x_i$ is the third literal in $C_j$ and $M^{\text{false}}_{i}\subseteq M$.
    Observe that in this case $a_j$ appears in $A'_i$ in the preference list of $\overline{y}_i$ in the second layer,
    and since $\{x_i, \overline{y}_i\}\in M$,  that $\overline{y}_i$ prefers~$a_j$ to its partner~$x_i$ in the second layer.
    It follows that $a_j$ must obtain a partner that it prefers to~$\overline{y}_j$ in the second layer.
    By investigating the preference list of $a_j$ in the second layer (note that we are in Variant~$1$), we have that $M(a_j) \in \{f_j, d_j\}$.
      If $M(a_j)=f_j$, then by looking at the preference list of $f_j$ in the second layer
      we infer that both $b_j$ and $c_j$ must obtain partners that they find better 
      than $f_j$ in the second layer. Thus $M(c_j)=e_j$, and consequently, $M(b_j)=d_j$.
      Summarizing, in this case we have that $N^{2}_j \subseteq M$.
       
       Analogously, if $M(a_j)=d_j$, 
       then $\{b_j, e_j\}\in M$ as otherwise they would block the second layer.
       Moreover, $f_j$ must obtain a partner that it prefers to~$a_j$.
       Since $c_j$ is the only agent left that $f_j$ prefers to~$a_j$,
       we obtain that $M(c_j)=f_j$.
       This leads to $N^1_j\subseteq M$.

    \item  Finally, assume that $\overline{x}_i$ is the third literal in $C_j$ and $M^{\text{true}}_{i}\subseteq M$.
    In this case, we have that $d_j$ appears in $D'_i$ in the preference list of $\overline{x}_i$ in the second layer, and that $\{\overline{x}_i, \overline{y}_i\}\in M$.
    This means that $\overline{x}_i$ prefers~$d_j$ to its partner~$\overline{y}_i$ in the second layer,
    and so it must be the case that $d_j$ obtains a partner that it prefers to~$\overline{x}_j$ in the second layer.
    As a result, we have that $M(d_j) \in \{a_j, b_j\}$.
    If $M(d_j)=a_j$, then the preference list of $a_j$ indicates that both $e_j$ and $f_j$ must obtain partners that they find better 
    than $a_j$ in the second layer. Thus, $M(f_j)=c_j$, and $M(e_j) \in \{b_j, c_j\}$. Consequently, $M(e_j) = b_j$, and we get that $N^{1}_j \subseteq M$.
    Finally, if $M(d_j)=b_j$, 
    then $\{c_j, e_j\}\in M$ as otherwise they would block the second layer.
    Further, $a_j$ must obtain a partner that it prefers to~$d_j$, thus $M(a_j)=f_j$.
    Consequently, $N^2_j\subseteq M$.
    \end{compactenum}
\end{proof}

For each clause~$C_j\in \mathcal{C}$ with $|C_j|=2$,
let $N^{1}_j=\{\{a_j,d_j\}, \{b_j,e_j\}\}$, 
%$N^2_j=\{\{a_j,f_j\}, \{b_j,d_j\},\{c_j,e_j\}\}$, 
$N^{2}_j=\{\{a_j,e_j\}, \{b_j,d_j\}\}$.

\begin{claim}\label{claim:clause_gadget-2}
   Let~$C_j\in \mathcal{C}$ be a size-two-clause, and let $x_i$ be a variable that appears (as either a positive or a negative literal) in $C_j$.
  Assume that $M$ is a \gstable{$2$} matching. The following holds:
  \begin{compactenum}[(i)]
    \item\label{claim-claus-1} If $x_i$ is the first literal in $C_j$ and if $M^{\text{false}}_{i}\subseteq M$, 
    then $N^{2}_j\subseteq M$.
    \item\label{claim-claus-2} If $\overline{x}_i$ is the last literal in $C_j$ and if $M^{\text{true}}_{i}\subseteq M$, 
    then $N^{2}_j\subseteq M$.
    \item\label{claim-claus-3} If $x_i$ is the last literal in $C_j$ and if $M^{\text{false}}_{i}\subseteq M$, 
    then $N^{1}_j\subseteq M$.
    \item\label{claim-claus-4} If $\overline{x}_i$ is the first literal in $C_j$ and if $M^{\text{true}}_{i}\subseteq M$, 
    then $N^{1}_j\subseteq M$.
  \end{compactenum}
\end{claim}
\begin{proof}
  \renewcommand{\qedsymbol}{(of
      \cref{claim:clause_gadget-2})~$\diamond$}
    We show the first two statements together and the last two statements together.
    Assume that one of the conditions in the first two statements holds, that is, 
    \begin{compactenum}
      \item $x_i$ is the first literal in $C_j$ and $M^{\text{false}}_{i}\subseteq M$, 
    or
    \item $\overline{x}_i$ is the last literal in $C_j$ and $M^{\text{true}}_{i}\subseteq M$.
    \end{compactenum}
    This implies that 
    \begin{compactenum}
      \item 
    either $\{x_i,\overline{y}_i\}\in M$,
    and in the first layer the list~$D_i$ contains~$d_j$,
    or 
    \item
    $\{\overline{x}_i,\overline{y}_i\}\in M$, 
    and in the second layer the list $D'_i$ contains~$d_j$ and we are in Variant~$2$.
    \end{compactenum}
    Since $x_i$ prefers all agents from $D_i$ to $\overline{y}_i$ in the first layer
    and $\overline{x}_i$ prefers all agents from $D'_i$ to $\overline{y}_i$ in the second layer, 
    we must have that $d_j$ obtains a partner that it prefers to $x_i$ in the first layer or to $\overline{x}_i$ in Variant~$2$ of the second layer.
    In either case, $b_j$ is the only agent that fulfills the requirement, implying that $\{b_j, d_j\}\in M$.
    By looking at the preference lists of $b_j$ and $e_j$ in the first layer,
    we derive that $\{a_j, e_j\}\in M$. 
    Thus, $N^{2}_j \subseteq M$.
 
    Analogously, assume that one of the conditions in the last two statements holds, that is, 
    \begin{compactenum}
      \item $x_i$ is the last literal in $C_j$ and $M^{\text{false}}_{i}\subseteq M$, or 
      \item $\overline{x}_i$ is the first literal in $C_j$ and $M^{\text{true}}_{i}\subseteq M$.
    \end{compactenum}
    This implies that 
   \begin{compactenum}
      \item $\{x_i,\overline{y}_i\}\in M$,
    and in the second layer we have Variant~$1$ such that the list $A'_i$ contains~$a_j$, or 
    \item $\{x_i,y_i\}\in M$, 
    and in the first layer the list $A_i$ contains~$a_j$.
    \end{compactenum}
    Since $\overline{y}_i$ prefers all agents from $A'_i$ to $x_i$ in the second layer 
    and $y_i$ prefers all agents from $A_i$ to $x_i$ in the first layer,
    we must have that $a_j$ obtains a partner that it prefers to $\overline{y}_i$ in the second layer (Variant~1) or to $y_i$ in the first layer.
    In either case, $d_j$ is the only agent that fulfills the requirement, implying that $\{a_j, d_j\}\in M$.
    By the preference lists of $d_j$ and $b_j$ in the first layer,
    we further derive that $\{b_j, e_j\}\in M$. 
    Thus, $N^{1}_j \subseteq M$.
\end{proof}

Now, we are ready to show that $(X,\mathcal{C})$ admits a satisfying truth assignment if and only if there exists a \gstable{$2$} matching for the so-constructed instance.

\medskip
$\boldsymbol{(\Rightarrow)}$ For the ``only if'' direction, assume that $\sigma\colon X\to \{T,F\}$ is a satisfying truth assignment for $(X,\mathcal{C})$.
We claim that the matching~$M$ constructed as follows is all-layer globally stable.
\begin{compactenum}[(1)]
\item For each variable~$x_i \in X$ with $\sigma(x_i)=T$, let $M^{\text{true}}_{i} \subseteq M$; otherwise let $M^{\text{false}}_{i} \subseteq M$.
\item For each size-three-clause~$C_j$, identify a literal~$\ell_j$ such that $\sigma(\ell_j)$ makes $C_j$ satisfied.
If $\ell_j$ is $C_j$'s first literal, then let $N^1_j\subseteq M$.
If $\ell_j$ is the second literal, then let $N^2_j\subseteq M$.
Otherwise let $N^3_j \subseteq M$.
\item For each size-two-clause~$C_j$,  let $\{c_j, f_j\}\in M$.
Identify a literal~$\ell_j$ such that $\sigma(\ell_j)$ makes $C_j$ satisfied.
If $\ell_j$ is the first literal in $C_j$, then let $N^{1}_j\subseteq M$; otherwise let $N^{2}_j\subseteq M$.
\end{compactenum}

Towards a contradiction suppose that $M$ is not \gstable{$2$}, and let $p=\{u,w\}$ be a possible blocking pair with $u\in U$ and $w\in W$.
First, we observe that $p$ involves neither two variable agents that correspond to different variables nor two clause agents that correspond to different clauses.
Second, $p$ does not involve two variable agents that belong to the same variable since
for each layer and for each two variable agents that correspond to the same variable, 
exactly one of both is already matched to its most preferred agent.
Third, $p$ also does not involve two clause agents that belong to the same size-two-clause 
as either of such clause agents is already matched with its most preferred agent.

Next, consider the case that $u$ and $w$ are two clause agents that belong to the same size-three-clause, say~$C_j$.
If $\{u,w\}$ is blocking $M$ in the first layer,
then we know that $N^{3}_j \subseteq M$ as otherwise either $u$ or $w$ already obtains its most preferred agent.
But then $\{u,w\}$ cannot be blocking $M$ in the first layer as for each agent $w'$ that is preferred to $M(u)$ by $u$ in the first layer, we have that $w'$ prefers $M(w')$ to $u$ in the first layer.
Similarly, if $\{u,w\}$ is blocking $M$ in the second layer with Variant~$1$ (resp.\ Variant~$2$),
then we know that $N^1_j \subseteq M$ (resp.\ $N^2_j\subseteq M$) as otherwise either $u$ or $w$ already obtains its most preferred agent.
Since for each agent~$w'$ such that $u$ prefers $w'$ to $M(u)$ in the second layer 
we have that $w'$ prefers~$M(w')$ to~$u$,
it follows that $\{u,w\}$ cannot be blocking the second layer.

Now, suppose that $p$ involves a variable agent and a clause agent.
By the construction of the preference lists, we can assume that the clause agent involved in the blocking pair~$p$ is either an $a_j$ or a $d_j$ for some $j\in [m]$.
We distinguish between two cases, depending on the size of $C_j$.

\medskip
\noindent
\textbf{Case 1: $\boldsymbol{|C_j|=3}$.}
If $a_j \in p$ and $p$ is blocking the first layer, 
then by the preference list of $a_j$ in the first layer,
we have that $N^{2}_j\subseteq M$.
By the construction of the matching~$M$, we have that the truth assignment of the second literal, say~$\overline{x}_i$, in $C_j$ makes $C_j$ satisfied; note that by our convention,
the second literal is always negative.
This implies that $M^{\text{false}}_i\subseteq M$. % since, by our convention, the second literal in $C_j$ is negative.
Since $p$ involves $a_j$ and is blocking the first layer,
it follows that the agent~$y_i$ that corresponds to variable~$x_i$ prefers $a_i$ to $M(y_i)$ in the first layer.
This means that $M^{\text{true}}_i\subseteq M$, a contradiction.

Analogously, if $a_j \in p$ and $p$ is blocking the second layer, 
then by the preference list of $a_j$ in the second layer,
we have that the preference list of $a_j$ comes from Variant~$1$ 
and $N^{3}_j\subseteq M$.
By the construction of the matching~$M$, we have that the truth assignment of the third literal which is positive in $C_j$ (recall that Variant~$1$ was used) makes $C_j$ satisfied.
Let this literal be $x_i$.
Then, it must hold that $M^{\text{true}}_i\subseteq M$.
Since $p$ involves $a_j$ and is blocking the second layer (due to Variant $1$),
it follows that the agent~$\overline{y}_i$ that corresponds to variable~$x_i$ prefers $a_j$ to $M(\overline{y}_i)$ in the second layer.
This means that $M_i^{\text{false}}\subseteq M$, which results in a contradiction.

If $d_j\in p$ and $p$ is blocking the first layer,
then by the preference list of $d_j$ in the first layer,
we have that $N^{1}_j\subseteq M$.
By the construction of the matching~$M$, we have that the truth assignment of the first literal, say~$x_i$, in $C_j$ makes $C_j$ satisfied; note that by convention,
the first literal is always positive.
This implies that $M^{\text{true}}_i\subseteq M$.
Since $p$ involves $d_j$ and is blocking the first layer,
it follows that the agent~$x_i$ that corresponds to variable~$x_i$ prefers~$d_j$ to~$M(x_i)$ in the first layer.
By the preference list of $x_i$ in the first layer, 
it follows that $M^{\text{false}}_i\subseteq M$, which also yields a contradiction.

Finally, if $d_j\in p$ and $p$ is blocking the second layer,
then by the preference list of $d_j$ in the second layer,
we have that the preference list of $d_j$ comes from Variant~$2$ 
and $N^{3}_j\subseteq M$.
By the construction of the matching~$M$, we have that the truth assignment of the third literal which is negative in $C_j$ makes $C_j$ satisfied.
Denote this literal by $\overline{x}_i$.
Then, it follows that $M^{\text{false}}_i\subseteq M$.
Since $p$ involves $d_j$ and is blocking the second layer (due to Variant $2$),
it follows that the agent~$\overline{x}_i$ that corresponds to 
variable~$x_i$ in $C_j$ prefers~$d_j$ to~$M(\overline{x}_i)$ in the second layer.
By the preference list of $x_i$ in the first layer, 
this means that $M^{\text{true}}_i\subseteq M$, which is a contradiction.

\smallskip
\noindent
\textbf{Case 2: $\boldsymbol{|C_j|=2}$.} If $a_j \in p$, then by the preference lists of $a_j$ in any of the two layers,
we have that $N^{2}_j\subseteq M$.
Hence, the second literal in $C_j$ makes it satisfied.
We distinguish between two cases.
If the second literal in $C_j$ is positive, say~$x_i$,
then $M^{\text{true}}_i\subseteq M$ and the other involved agent in $p$ must be
either $y_i$ or $\overline{y}_i$.
Since $\overline{y}_i$ prefers its partner~$\overline{x}_i$ to $a_j$ in both layers,
it follows that $y_i$ is the other involved agent.
However, by the preference list of $y_i$ in the first layer, $A_i$ does not contain $a_j$, meaning that $y_i$ also prefers its partner~$x_i$ to $a_j$ in both layers, which is a contradiction to $p$ being a blocking pair.

If the second literal in $C_j$ is negative, say~$\overline{x}_i$,
then $M^{\text{false}}_i\subseteq M$ and our reasoning is very similar.
First we infer that the other involved agent in $p$ must be
either $y_i$ or $\overline{y}_i$.
Since $y_i$ prefers its partner~$\overline{x}_i$ to $a_j$ in both layers,
it follows that the other agent in $p$ is $\overline{y}_i$.
However, by the preference list of $\overline{y}_i$ in the second layer, $A'_i$~does not contain $a_j$, which means that $\overline{y}_i$ prefers its partner to $a_j$ in both layers. Thus, in this case we also get a contradiction.

If $d_j \in p$, then our reasoning is very similar. First, by looking at the preference lists of $d_j$ in any of the two layers,
we infer that $N^{1}_j\subseteq M$.
By the construction of the matching~$M$, we get that the first literal in $C_j$ makes it satisfied.
We consider two cases.
If this literal is positive, say~$x_i$,
then $M^{\text{true}}_i\subseteq M$ and the other involved agent in $p$ must be
either $x_i$ or $\overline{x}_i$.
Agent $x_i$ prefers its partner~$y_i$ to~$d_j$ in both layers and so it cannot be involved in the blocking pair.
Thus, $\overline{x}_i$ is the other involved agent.
However, by the preference list of $\overline{x}_i$ in the first layer, $D'_i$ does not contain $d_j$, meaning that $\overline{x}_i$ also prefers its partner~$\overline{y}_i$ to $d_j$ in both layers, a contradiction to $p$ being a blocking pair.

If the first literal in $C_j$ is negative, say~$\overline{x}_i$,
then $M^{\text{false}}_i\subseteq M$ and the other involved agent in $p$ must be
either $x_i$ or $\overline{x}_i$.
Since $\overline{x}_i$ prefers its partner~$y_i$ to $d_j$ in both layers,
it follows that $x_i$ is the other involved agent.
However, by the preference list of $x_i$ in the last layer, $D_i$ does not contain $d_j$, meaning that $x_i$ prefers its partner~$\overline{y}_i$ to $d_j$ in both layers, which is again a contradiction.

% a $y_i$, that is, $w=y_i$,
% such that the corresponding variable~$x_i$ appears (either as a positive or as a negative literal) in $C_j$,
% and that the involved clause agent~$u$ is either $a_j$, or $d_j$, or $c_j$.
% If $y_i$ corresponds to the first literal in $C_j$,
% then the other involved agent must be $u=a_j$. 
% Since $p=\{y_i, a_j\}$ is blocking $M$,
% it follows that $\{a_j, f_j\}\in M$ (due to the preference list of $a_j$)
% and $a_j \pref_{y_i} M(y_i)$. 
% By the construction of $M$, $\{a_j, f_j\}\in M$ implies that the first literal in $C_j$ makes $C_j$ satisfied. 
% By the construction of the preference list of $y_i$,
%  $a_j \pref_{y_i} M(y_i)$ implies that either $a_j \in A_i$ or $a_j \in B_i$.
% The first case implies that $C_j$ contains the negative literal~$\overline{x}_i$. Since $C_j$ is satisfied by the first literal which is $\overline{x}_i$, we have that $\{\overline{x}_i, y_i\}\in M$---a contradiction to $\{y_i, a_j\}$ being a blocking pair.
% Analogously, the second case also leads to the contradiction of $\{y_i, a_j\}$ being a blocking pair since $a_j \in B_i$ implies that $C_j$ contains the positive literal~${x}_i$ and $\{x_i, y_i\}\in M$.
% Similarly, if $y_i$ corresponds to the second or the third literal in $C_j$,
% then $y_i$ will also obtain its most preferred agent as a partner---a contradiction.

\medskip
$\boldsymbol{(\Leftarrow)}$ For the ``if'' direction, let $M$ be a \gstable{$2$} matching. We construct a truth assignment~$\sigma$ as follows.
For each variable agent~$x_i$, if $M^{\text{true}}_i\subseteq M$, then let $\sigma(x_i)=T$; otherwise by \cref{claim:variable_gadget} we have that $M_i^{\text{false}}\subseteq M$, and let $\sigma(x_i)=F$.
Suppose, towards a contradiction, that $\sigma$ is not a satisfying assignment and let $C_j$ be a clause where none of the literals is evaluated to true. 
We distinguish between two cases.

\medskip
\noindent
\textbf{Case 1: $\boldsymbol{|C_j|=3}$.} Let $x_r$, $\overline{x}_s$, and $\ell_t$ be the first, second, and the third literal in $C_j$.
Since none of these literals is evaluated to true, 
it follows that $M^{\false}_r, M^{\true}_s\subseteq M$.
By statements~$(i)$ and $(ii)$ in \cref{claim:clause_gadget}, 
we must have that $N^{3}_j\subseteq M$.
However, by the statements~$(iii)$ and $(iv)$ in \cref{claim:clause_gadget}, 
applied for $\ell_t$, we must have that either $N^{1}_j\subseteq M$ or $N^{2}_j\subseteq M$, a contradiction.

\smallskip
\noindent
\textbf{Case 2: $\boldsymbol{|C_j|=2}$.} Let the first and the second literals in $C_j$ correspond to variables~$x_i$ and $x_k$, respectively.
If $C_j$ is positive, then since $C_j$ is not satisfied, we have that $M^{\text{false}}_i,
M^{\text{false}}_k \subseteq M$.
By \cref{claim:clause_gadget-2}, we have that both $N^1_j$ and $N^2_j$ must belong to $M$, which leads to a contradiction.

Analogously, if $C_j$ is negative, then since $C_j$ is not satisfied, we have that $M^{\text{true}}_i,
M^{\text{true}}_k \subseteq M$.
By \cref{claim:clause_gadget-2}, we have that both $N^1_j$ and $N^2_j$ must belong to $M$, a contradiction.

Altogether, we showed that the constructed matching is globally stable for two layers. This concludes the proof of \cref{thm:global-hard-alpha=ell=constant}.
\end{proof}

Since \gstability{$\ell$} equals \abstability{$\ell$} (\cref{prop:g_stability_equiv_pair_stability}), by \cref{thm:global-hard-alpha=ell=constant} we obtain the following corollary for the pair stability.
\begin{corollary}
  \label{cor:pair-hard-alpha=ell=constant}
  \BSM is NP-hard even if $\alpha=\ell=2$.
\end{corollary}

By adding to the profile constructed in the proof of \cref{thm:global-hard-alpha=ell=constant} an arbitrary number of layers with preferences that are stable anyway, we can deduce hardness for arbitrary $\alpha=\ell\ge 2$.

\begin{proposition}\label{prop:global=pair-hard-alpha=ell}
  For each $\alpha=\ell\ge 2$, both \GSM and \BSM are NP-hard.
\end{proposition}
\begin{proof}
  We add to the profile constructed in the proof of \cref{thm:global-hard-alpha=ell=constant} $\ell-2$ layers with preferences of the following form:
      \begin{alignat*}{4}
 \text{Layers } (3)\text{--}(\ell),  &  \forall i \in [n]\colon \quad
       &  x_i\colon& y_i \pref \overline{y}_i \pref \cdots, &\quad\quad  y_i\colon& \overline{x}_i \pref x_i \pref \cdots,\\
  &  & \overline{x}_i \colon& \overline{y}_i \pref y_i \pref \cdots,  &  \overline{y}_i\colon& x_i\pref \overline{x}_i  \pref \cdots \text{.}\\
 &   \forall j \in [m]\colon\quad
    &  a_j\colon& d_j\pref e_j  \pref f_j\pref \cdots, & \quad d_j\colon & b_j \pref c_j \pref a_j \pref \cdots,\\
&    & b_j\colon& e_j\pref f_j \pref d_j \pref \cdots, & \quad e_j\colon & c_j \pref a_j \pref b_j \pref \cdots,\\
    && c_j\colon& f_j\pref  d_j \pref e_j \pref \cdots, & \quad f_j\colon & a_j \pref b_j \pref c_j\pref \cdots\text{.}
  \end{alignat*}
  It is straightforward that a matching is \gstable{$\ell$} if and only if  it is \gstable{$2$} for the first two layers. 
\end{proof}

\section{Multi-Layer Stable Marriage with $\alpha < \ell$}
\label{sec:notall-stab}

In this section we show that for each of the three concepts that we introduced in \cref{sec:defi} the problem of computing a multi-layer stable matching is computationally hard as soon as $2\leq\alpha < \ell$.

\subsection{NP-hardness for \gstability{$\alpha$}}
To find a matching~$M$ that is \gstable{$\alpha$}, even if $\alpha< \ell$,
the main difficulty is not just to determine $\alpha$~layers where $M$ should be stable.
In fact, we sometimes need to find a matching that is stable in some specific layers.
This requirement allows us to adapt the construction in the proof of \cref{thm:global-hard-alpha=ell=constant}
to show hardness for deciding \gstability{$\alpha$} for the case when $2\le \alpha < \ell$.

\begin{proposition}\label{prop:global-alpha_ge_2}
  For each fixed number~$\alpha$ with $\alpha\ge 2$, \GSM is NP-hard. % deciding whether an $\ell$-layer profile admits an \gstable{$\alpha$} matching is NP-hard.
\end{proposition}
\begin{proof}
 To prove the NP-hardness, we adapt the reduction in the proof of \cref{thm:global-hard-alpha=ell=constant} which shows that deciding \gstability{$\alpha$} for $\alpha=\ell=2$ is NP-hard.
 Let $\Pot$ be the constructed two-layer instance in the proof of \cref{thm:global-hard-alpha=ell=constant}.
 Besides the original agents from $\Pot$, we introduce two sets~$U$ and $W$ of dummy agents with $|U|=|W|=2\cdot (\ell - \alpha + 1)$, where $U=\{u_1, u_2, \ldots, u_{\ell-\alpha+1}\}$ and $W=\{w_1, w_2, \ldots, w_{\ell-\alpha+1}\}$.
  The idea of introducing such dummy agents is to make sure that each \gstable{$\alpha$} matching must include all pairs~$\{u_j, w_j\}$, $j\in [\ell-\alpha+1]$.
  However, this is the case only when such matching is stable in the two layers constructed in the NP-hardness proof of \cref{thm:global-hard-alpha=ell=constant}; we denote these two layers as layers~$(1)$ and $(2)$.
In the following, we use ``$\cdots$'' to denote an arbitrary order of the unmentioned agents.

  \paragraph{Preferences of the original agents.} The preferences of the original agents in the first two layers are the same as in $\Pot$ in the proof of
\cref{thm:global-hard-alpha=ell=constant}.
  For each other layer, the preferences of the original agents are as follows. %, where the notation~``$\cdots$'' denotes an arbitrary order of the unmentioned agents:
   \begin{alignat*}{4}
     \text{Layers } (3)\text{--}(\ell),  &  \forall i \in [n]\colon \quad
       &  x_i\colon& y_i \pref \overline{y}_i \pref \cdots, &\quad\quad  y_i\colon& \overline{x}_i \pref x_i \pref \cdots,\\
  &  & \overline{x}_i \colon& \overline{y}_i \pref y_i \pref \cdots,  &  \overline{y}_i\colon& x_i\pref \overline{x}_i  \pref \cdots \text{.}\\
 &   \forall j \in [m]\colon\quad
    &  a_j\colon& d_j\pref e_j  \pref f_j\pref \cdots, & \quad d_j\colon & b_j \pref c_j \pref a_j \pref \cdots,\\
&    & b_j\colon& e_j\pref f_j \pref d_j \pref \cdots, & \quad e_j\colon & c_j \pref a_j \pref b_j \pref \cdots,\\
    && c_j\colon& f_j\pref  d_j \pref e_j \pref \cdots, & \quad f_j\colon & a_j \pref b_j \pref c_j\pref \cdots\text{.}
  \end{alignat*}  
  \paragraph{Preferences of the dummy agents.}
  The preferences of the dummy agents are as follows% , where the notation~``$\cdots$'' denotes an arbitrary order of the unmentioned agents
  ; let $\hat{\ell}=\ell-\alpha+1$:
  \begin{alignat*}{4}
     \text{Layers } (1)\text{--}(\alpha),~~  &  \forall j \in [\hat{\ell}]\colon \quad
       &  u_j\colon& w_j \pref \cdots, &\quad\quad  w_j\colon& u_j \pref \cdots \text{,}\\
       \text{Layer } (i+\alpha), 1 \le i \le \hat{\ell}-1,~~  &  \forall j \in [\hat{\ell}]\colon \quad
        &  u_j\colon& w_{(j \bmod \hat{\ell})+i} \pref \cdots, &\quad\quad  w_j\colon& u_{(j-1 \bmod \hat{\ell}) +i} \pref \cdots \text{.}
  \end{alignat*}
  
Observe that each dummy agent obtains a different partner in different layers
with indices higher than~$\alpha$.
%  Observe that each layer with index higher than $\alpha$ admits exactly one stable matching regarding the dummy agents which is different from any other layer.
  More precisely, for each layer~$(i+\alpha)$ with $1\le \alpha \le \hat{\ell}-1$,
  the only stable matching in this layer must include $M_i=\{\{u_{j}, w_{(j\bmod \hat{\ell})+1} \} \mid 1\le j \le \hat{\ell}\}\}$ since $u_j$ and $w_{(j\bmod \hat{\ell})+1}$ are each other's most preferred agent.
  Moreover, by the same reasoning, each layer with index at most $\alpha$ admits exactly the same stable matching regarding the dummy agents which is different from any layer with index higher than $\alpha$, namely $M_0=\{\{u_j, w_j\} \mid 1\le j \le \hat{\ell}\}$.
  For each two distinct values~$i,j\in \{0,1,\ldots, \hat{\ell}-1\}$, however, 
  we have that $M_i\cap M_j =\emptyset$.
  This means that each \gstable{$\alpha$} matching must include $M_0$ and must be stable in the first $\alpha$ layers, including the first two layers.
  
  Now, it is straightforward to see that a matching~$M$ is \gstable{$2$} for~$\Pot$ if and only if $M\cup M_0$ is \gstable{$\alpha$} for our new instance.
\end{proof}

We remark that our proof for \cref{prop:global-alpha_ge_2} also implies hardness for $\alpha=\ell$ for arbitrary $\ell\ge 2$. 
% \begin{theorem}\label{thm:global_stability_hard}
% To decide whether there exists an \gstable{$\alpha$} matching is NP-hard, assuming that $\alpha$ is a part of the input.
% \end{theorem}
% \begin{proof}
% We give a polynomial-time many-one reduction from the \textsc{Exact Set Cover (X3C)} problem. Let $I$ be an instance of \textsc{X3C}, where we are given a set of $3n$ elements $E = \{e_1, \ldots, e_{3n}\}$ and a collection $\calS$ of 3-element subsets of $E$. We ask if $\calS$ contains $n$ such sets that altogether cover all elements from $E$. W.l.o.g., we assume that each element from $E$ belongs to exactly three sets from $\calS$~\cite{GJ79}.

\subsection{NP-hardness for \aistability{$\alpha$}}
For \aistability{$\alpha$}, we also obtain a hardness result by reducing from the NP-hard \textsc{Perfect SMTI} problem, the problem of finding a perfect \myemph{SMTI-stable} matching with (possibly) incomplete preference lists and ties~\cite{IwMiMoMa1999,MaIrIwMiMo2002} which is defined as follows.
A preference list is \myemph{incomplete} if not all agents from one side are considered acceptable to an agent from the other side.
A preference list has a \myemph{tie} if there are two agents in the list which are considered to be equally good.
As a result, a preference list of an agent~$u$ from one side can be considered as a weak (\emph{i.e.}\ transitive and complete) order~$\succeq_u$ on a subset of the agents on the other side.   
We use $\succ_u$ and $\sim_u$ to denote the asymmetric and symmetric part of the preference list, respectively.
Equivalently, two agents~$x$ and $y$ are said to be tied by $u$ if $x\succeq_u y$ and $y\succeq_u x$, denoted as $x\sim_u y$.
Formally, we say that a matching $M$ for a   \textsc{Perfect SMTI} with two disjoints sets~$U$ and $W$ of agents is \myemph{SMTI-stable} if there are no SMTI-blocking pairs for $M$. A pair $\{u, w\}$ is \myemph{SMTI-blocking} $M$ if all of the following three conditions are satisfied:
\begin{inparaenum}[(i)]
\item $u$ and $w$ appear in the preference lists of each other, 
\item $w \succ_u M(u)$ or $u$ is not matched to any agent from $W$, and
\item $u \succ_w M(w)$ or $w$ is not matched to any agent from $U$.
\end{inparaenum}

The reduction is based on the following ideas:
In a \textsc{Perfect SMTI} instance~$I$ the preference list of an agent~$z$ may have ties.
To encode ties, we first ``linearize'' the preference list of $z$ in $I$ to obtain two linear preference lists such that the resulting lists restricted to the tied agents are reverse to each other.
Then, we let half of the $\ell$ layers have one of the preference lists and let the remaining half to have the other list.
Since $\alpha < \nicefrac{\ell}{2}$, it is always possible to find half layers which fulfill our \aistability{$\alpha$} constraint.

In~$I$, two agents, say $x$ and $y$, may not be acceptable to each other. 
To encode this, we introduce to the source instance $\ell$~pairs of dummy agents with $\ell$ layers of preferences that preclude $x$ and $y$ from being matched together. 
However, to make sure that an agent~$x$ is not matched to any dummy agent, we have to require that $\ell \ge 4$.

\begin{theorem}\label{thm:hardness_individual_stable}
  For each fixed number~$\ell$ of layers with $\ell\ge 4$ and for each fixed value~$\alpha$ with $2\le \alpha \le \lfloor \nicefrac{\ell}{2} \rfloor$, 
  \ISM is NP-hard, even if on one side, the preference list of each agent is the same in all layers.
  % Given a preference profile with $\ell$ layers, $\ell\ge 4$,
  % for each number~$\alpha$ with $2\le \alpha \le \lfloor \nicefrac{\ell}{2}\rfloor$ deciding whether there exists an \aistable{$\alpha$} matching is NP-hard, even if the preferences of each agent on one side remain the same in all layers.
\end{theorem}
\begin{proof}
Assume that  $\ell \geq 4$ and that $2\le \alpha \leq \lfloor \nicefrac{\ell}{2} \rfloor$.
We give a polynomial-time reduction from the NP-hard \textsc{Perfect SMTI} problem~\cite{IwMiMoMa1999,MaIrIwMiMo2002}.

Let $I$ be an instance of \textsc{Perfect SMTI}. In $I$ we are given two disjoint sets of agents, $U$ and $W$, with $|U| = |W| = n$; each agent $u \in U$ is endowed with a weak order $\succeq_u$ on a subset of~$W$.
By the SMTI-stability, an agent~$u$ prefers not to be assigned to any agent rather than to be assigned to an agent outside of its preference list. 

%In~$I$ we are additionally given an integer $t$ and we ask if there exists a stable matching such that at least $t$ agents from each set, $U$ and $W$, are matched. 
We assume that ties occur in the preferences of the agents from the side $U$ only, that there is at most one tie per list, and each tie is of length two as this variant remains NP-hard~\cite[Theorem 2.2]{MaIrIwMiMo2002}.

From $I$ we construct an instance $I'$ of the problem of finding an \aistable{$\alpha$} matching in the following way. We copy the sets of agents $U$ and $W$; further, we introduce a set of $2\ell\cdot n$ dummy agents $P \cup R$ with $|P| = |R| = \ell \cdot n$. We denote the elements in these sets as: $P = \{p_{i,j}\mid 1\le i \le n \wedge 1\le j \le \ell\}$
and $R = \{r_{i, j} \mid1\le i \le n \wedge 1\le j \le \ell \}$. One side of the bipartite ``acceptability graph'' will consist of the agents from $U \cup P$, and the other side of the agents from $W \cup R$. We construct the preference orders of the agents as follows:
\begin{description}
\item[Agents from $\boldsymbol{U}$.]
Consider an agent $u \in U$, and let $L_u$ denote its preference list in~$I$. In~$I'$ we construct the preference list of $u$ from~$L_u$ as follows. We iterate over~$L_u$ starting from the top position. If agent~$u$ prefers~$x$ over~$y$, then we assume that $u$ also prefers~$x$ over~$y$ in all layers in~$I'$. If $x$ and~$y$ are tied in~$L_u$, then we assume that $u$ prefers~$x$ over~$y$ in the first $\lceil \nicefrac{\ell}{2} \rceil$ layers and $y$ over~$x$ in the remaining $\lfloor \nicefrac{\ell}{2} \rfloor$ layers (or the other way around). The remaining parts of the preference orders of $u$ are constructed as follows: first, let us assume that $u = u_i \in U$. Right after all agents from~$L_u$, agent~$u$ puts in all layers $r_{i, 1}, r_{i, 2}, \ldots, r_{i, \ell}$, %in each of the $\ell$~layers, 
respectively in the following orders:
\begin{inparaenum}[(1)]
\item $r_{i, 1} \succ r_{i, 2} \succ \ldots \succ r_{i, \ell}$,
\item $r_{i, 2} \succ r_{i, 3} \succ \ldots \succ r_{i, \ell} \succ r_{i,1}$, and so on, until   
\item[($\ell$)] $r_{i, \ell} \succ r_{i, 1} \succ r_{i, 2} \succ \ldots \succ r_{i, \ell-1}$.
\end{inparaenum}
Next, $u$ puts all the remaining agents in an arbitrary order. For example, if $\ell = 5$ and $L_u$ is equal to $w_1 \succ w_4 \succ w_3 \sim w_5$, then the preference lists of~$u$ in the five layers will be as follows, where ``$\cdots$'' denotes some arbitrary but fixed order of the remaining agents.
\begin{quote}
$\text{Layer } (1)\colon \text{agent } u_i \colon w_1 \succ w_4 \succ w_3 \succ w_5 \succ  r_{i, 1} \succ r_{i, 2} \succ r_{i, 3} \succ r_{i, 4} \succ r_{i, 5} \succ \cdots,$ \\
$\text{Layer } (2)\colon \text{agent } u_i \colon w_1 \succ w_4 \succ w_3 \succ w_5 \succ 
                     r_{i, 2} \succ r_{i, 3} \succ r_{i, 4} \succ r_{i, 5}  \succ r_{i, 1} \succ \cdots\text{,}$\\
$\text{Layer } (3)\colon \text{agent } u_i \colon w_1 \succ w_4 \succ w_3 \succ w_5 \succ 
                     r_{i, 3} \succ r_{i, 4}  \succ r_{i, 5}\succ r_{i, 1} \succ r_{i, 2} \succ \cdots,$ \\
$\text{Layer } (4)\colon \text{agent } u_i \colon w_1 \succ w_4 \succ w_5 \succ w_3 \succ 
                     r_{i, 4} \succ r_{r, 5} \succ r_{i, 1} \succ r_{i, 2} \succ r_{i, 3} \succ \cdots\text{,}$\\
$\text{Layer } (5)\colon \text{agent } u_i \colon w_1 \succ w_4 \succ w_5 \succ w_3 \succ 
                     r_{r, 5} \succ r_{i, 1} \succ r_{i, 2} \succ r_{i, 3} \succ   r_{i, 4} \succ \cdots\text{.}$
\end{quote}
Observe that in the first three layers $w_3$ is preferred to $w_5$ and in the remaining layers $w_5$ is preferred to $w_3$.
\item[Agents from $\boldsymbol{W}$.]
For each agent $w_i$ from $W$, we recall that the preferences of $w_i$ in $I$ do not have ties (that was one of the assumptions in the problem we reduce from).
Let $L_{w_i}$ denote the preference list of $w_i$.
The preferences of $w_i$ in $I'$ are the same in all layers, where the second~``$\cdots$'' denotes some arbitrary but fixed order of the remaining agents. 
\begin{quote}
$\text{Layers } (1)\text{-}(\ell)\colon \text{agent } w_i\colon L_{w_i} \succ  p_{i, 1} \succ p_{i, 2} \succ \cdots \succ p_{i, \ell} \succ \cdots\text{.}$
\end{quote}

% \item[Agents from $B \cup C$.]
% We describe the construction of preferences of the agents from $C$; the construction for the agents from $B$ will be similar. The preference order of each agent from $C$ is the same, and further, it is the same in all four layers. Such agent first ranks all agents from $U$ in some fixed arbitrary order (say, from $u_1$ to $u_n$), next all the agents from $B$ (also in some fixed arbitrary order, say, from $b_1$ to $b_{n-t}$), and next all the remaining agents:
% \begin{align*}
% \text{layer 1-4}\colon u_1 \succ \ldots \succ u_n \succ b_1 \succ \ldots \succ b_{n-t} \succ \ldots
% \end{align*}
   
\item[Agents from $\boldsymbol{P \cup R}$.]
Consider $p_{i, j} \in P$: this agent ranks $w_i$ first, next $r_{i, j}$, and next all the remaining agents in some arbitrary order:
\begin{quote}
$\text{Layers } (1)\text{-}(\ell)\colon \text{agent } p_{i,j}\colon w_i \succ r_{i, j} \succ \cdots\text{.}$
\end{quote}
The preferences of an agent from~$r_{i, j} \in R$ are constructed analogously:
\begin{quote}
$\text{Layers } (1)\text{-}(\ell)\colon \text{agent } r_{i,j}\colon u_i \succ p_{i, j} \succ \cdots\text{.}$
\end{quote}
\end{description}

This completes the description of the construction of instance~$I'$.
Obviously, the preferences of each agent from $W\cup R$ are the same in all layers.
Now, we will show that there exists a perfect SMTI-stable matching in $I$ if and only if there exists an \aistable{$\alpha$} matching in $I'$.

For the ``only if'' direction, assume that the instance $I$ has a perfect SMTI-stable matching~$M$.
We claim that $M'=M\cup \big\{\{r_{i,j},p_{i,j}\} \mid 1\le i \le n\wedge 1\le j\le \ell\big\}$ is \aistable{$\alpha$} for $I'$.
First, $M'$ is a perfect matching for $I'$ as $M$ is a perfect matching for $I$.
Second, observe that no two agents from $R\cup P$ are blocking $M'$ as 
for each agent~$a$ in $R\cup P$ it prefers its partner~$M'(a)$ to every other agent in $R\cup P$ in all layers.
Third, no agent from $U\cup W$ can form a blocking pair with an agent from $R\cup P$ for the matching~$M'$
as each agent $a$ from $U\cup W$ prefers its partner~$M'(a)=M(a)\in L_a$ to every other agent in $R\cup P$ in all layers.
Now, consider two arbitrary agents~$u$ and $w$ with $u\in U$ and $w\in W$ such that $\{u,w\}$ is unmatched.
We show that we can always find a set $S$ of $\lfloor \ell \rfloor\ge \alpha$ layers such that $M'(u)\succ_u w$ in each layer from $S$ or $M'(w)\succ_w u$ in each layer from $S$.
We can assume that $u$ and $w$ are acceptable to each other in $I$ as otherwise both $u$ and $w$ prefer to be with their respective partner~$M'(u)$ and $M'(w)$ rather than with each other in all layers.
Let us consider the following cases; note that the preference list of $w$ does not have ties and remains the same in all layers. 
\begin{description}
\item[Case 1: $\boldsymbol{u \pref_w M'(w)=M(w)}$.] By the stability of $M$ we have that $M'(u)=M(u) \succeq_u w$ in instance~$I$.
For the case that $M'(u)$ and $w$ tied by $u$ we can identify $\lfloor \nicefrac{\ell}{2} \rfloor$ layers, either the first $\lfloor \nicefrac{\ell}{2} \rfloor$ layers or the last $\lfloor \nicefrac{\ell}{2} \rfloor$ layers where $u$ prefers $M'(u)$ to $w$.
For the case that $M'(u)\succ_u w$ we have that $u$ prefers $M'(u)$ to $w$ in all layers.
\item[Case 2: $\boldsymbol{M'(w) \pref_w u}$.] In this case we know that $w$ prefers $M'(w)$ to $u$ in all layers.
\end{description}
%Observe that we do not have to consider the case when $M'(w)$ and $u$ are not acceptable, since $M$ was a perfect matching, and so $M'(w) = M(w)$ and $u$ are acceptable to each other.

%We will show that in $I'$ there exists a \aistable{2} matching. An example \aistable{2} matching in $I'$ has the following form: If $u \in U$ and $w \in W$ are matched in $I$, then they are also matched in $I'$. Whenever $u \in U$ is not matched in $I$, then in $I'$ it is matched to one of the agents from $C$ (the unmatched agent from $U$ with the lowest index is matched to $c_1$, the one with the second lowest index is matched to $c_2$, etc.). Finally, for each $i \in [n]$ and $j \in [\ell]$, agent $p_{i, j}$ is matched with $r_{i, j}$. One can verify that this is a \aistable{2} matching.

For the ``if'' direction, assume that $M'$ is an \aistable{$\alpha$} matching for $I'$.
First, observe that no agent $u_i \in U$ can be matched with an agent that it ranks lower than any agent from $\{r_{i, j} \mid 1\le j \le \ell\}$ in any layer. 
Indeed, for such matching each pair $\{u_i, r_{i, j}\}$, with $j \in [\ell]$, would be blocking~$M'$ in all layers. 
Similarly, $u_i$ cannot be matched with $r_{i, j}$ for $j \in [\ell]$ as the pair~$\{u_i, r_{i,(j+\ell-2\mod \ell)+1}\}$ would be blocking $M'$ in exactly $\ell-1$ layers, namely those layers other than the $j^{\text{th}}$ layer, which are more than $\ell-\alpha$ layers.
For instance, if $u_i$ was matched with~$r_{i, 1}$, then $\{u_i, r_{i, \ell}\}$ would be blocking the~$2^{\text{nd}}, 3^{\text{rd}},\ldots, \ell^{\text{th}}$ layers; if $u_i$ was matched with~$r_{i, 2}$, then $(u_i, r_{i, 1})$ would be blocking the~$1^{\text{st}}, 3^{\text{rd}}, 4^{\text{th}},\ldots, \ell^{\text{th}}$ layers, etc. Thus, each agent~$u$ from $U$ must be matched with someone from $L_u$, where $L_u$ is the preference list of $u$ in~$I$. 
Consequently, no agent in $W$ is matched with an agent in $P$ as otherwise, by the pigeonhole principle, some agent in $U$ must be matched with an agent in $R$ which is not possible by our reasoning above.
% Hua: I comment out the line below again as it does not work with the new construction, as the line after does not work similarly anymore because the construction for the preference list of w is not the same.
%Similarly to the case with $U$, no agent~$w_i\in W$ can be matched with an agent that it ranks lower than $\{p_{i,j}\mid 1\le j \le \ell\}$.
Also, by a similar reasoning we deduce that each agent~$w$ from $W$ must be matched with someone from $L_w$, where $L_w$ is the preference list of $w$ in~$I$.
%The same reasoning that we presented for the agents from $U$ can be also applied to the agents from $W$.
Now, we show that $M=\{\{u,w\}\in M' \mid u,w\in U\cup W\}$ is a perfect SMTI-stable matching for $I$.
Since $M'$ is a perfect matching in $I'$, by the reasoning above, it follows that $M$ is a perfect matching in $I$.
Suppose, for the sake of contradiction, that there is an unmatched pair~$\{x,y\}\notin M$ that is SMTI-blocking $M$.
This means that $x$ and $y$ are acceptable to each other,
and that $y\succ_x M(x)=M'(x)$ and $x\succ_y M(y)=M'(y)$ in $I$.
By the preference lists of $x$ and $y$ in $I'$ and by the definition of $M$,
it follows that in each layer $x$ prefers $y$ to $M'(x)$ and $y$ prefers $x$ to $M'(y)$---a contradiction to $M'$ being \aistable{$\alpha$} since $\alpha>2$.
\end{proof}

\subsection{NP-hardness of \abstability{$\alpha$}}
We note that in the instance constructed in the proof of \cref{thm:hardness_individual_stable} every agent from the side~$W\cup R$ has the same preference list in all layers.
Later on, in \cref{prop:single_layer_equiv} we will show that in such a case the concepts of \aistability{$\alpha$} and \abstability{$\alpha$} are equivalent, Thus, we obtain the same hardness result for pair stability when $\alpha \le \lfloor \nicefrac{\ell}{2} \rfloor$.

\begin{corollary}\label{cor:hardness_pair_stable}
    For each fixed number~$\ell$ of layers with $\ell\ge 4$ and for each fixed value~$\alpha$ with $2\le \alpha \le \lfloor \nicefrac{\ell}{2} \rfloor$, 
    \BSM is NP-hard even if on one side the preference list of each agent is the same in all layers.
% Given a preference profile with $\ell$ layers, $\ell\ge 4$,
  % for each number~$\alpha$ with $2\le \alpha \le \lfloor \nicefrac{\ell}{2}\rfloor$ deciding whether there exists an \abstable{$\alpha$} matching is NP-hard, even if the preferences of each agent on one side remain the same in all layers.
\end{corollary}

%If we drop the requirement that on one side each agent has the same preference list in all layers, we obtain NP-hardness for pair stability for the special case when $\alpha = \lfloor\ell /2 \rfloor+1$.

For $\alpha > \lceil \nicefrac{\ell}{2} \rceil$, we use an idea similar to the one used for showing \cref{prop:global-alpha_ge_2}. 

\begin{proposition}\label{pro:pair-alpha_ge_ell/2}
   For each fixed number~$\alpha$ with $\lceil \nicefrac{\ell}{2} \rceil +1 \le \alpha \le \ell$, \BSM is NP-hard. %, deciding whether an $\ell$-layer profile admits an \abstable{$\alpha$} matching is NP-hard.
\end{proposition}

\begin{proof}
  Assume that $\alpha\ge \lceil \ell /2 \rceil+1$. 
  To prove the NP-hardness, we slightly adapt the reduction in the proof of \cref{thm:global-hard-alpha=ell=constant} which shows that deciding \gstability{$\alpha$} is NP-hard for $\alpha=\ell=2$.
  Let $\Pot$ be the two-layer instance constructed in the proof of \cref{thm:global-hard-alpha=ell=constant}.
  The idea is to copy $\nicefrac{\ell}{2}$ times profile~$\Pot$ and make sure that an \abstable{$\alpha$} matching must be stable in the two layers of the original profile.

  % let $U$ and $W$ be the two sets of agents in $\Pot$.
  % Assume that the agents in $U$ and $W$ are labeled such that they have the form $U=\{u_1, u_2, \dots, u_n\}$ and $W=\{w_1, w_2, \dots, w_n\}$.
  
  % Now, we describe the adaption. 
  % For each agent~$u_i$ in $U$ (resp.\ $w_i$ in $W$), we introduce a copy~$u_i^{c}$ (resp.\ $w_i^c$).
  % These copies will be used to make sure that we want to find a matching that is stable in all layers of a specific subset of $\alpha$ layers. 

  For the preference lists in the $\ell$ layers, we do the following; let $k=\lfloor \nicefrac{\ell}{2} \rfloor$.
  \begin{compactenum}
    \item We make $k$ copies of the profile~$\Pot$.
    \item If $\ell$ is odd, then we add an $\ell^{\text{th}}$ layer with the following preference lists:
  \begin{alignat*}{4}
    \text{Layer }(\ell),  &  \forall i \in [n]\colon \quad
       &  x_i\colon& y_i \pref \overline{y}_i \pref \cdots, &\quad\quad  y_i\colon& \overline{x}_i \pref x_i \pref \cdots,\\
  &  & \overline{x}_i \colon& \overline{y}_i \pref y_i \pref \cdots,  &  \overline{y}_i\colon& x_i\pref \overline{x}_i  \pref \cdots \text{.}\\
 &   \forall j \in [m]\colon\quad
    &  a_j\colon& d_j\pref e_j  \pref f_j\pref \cdots, & \quad d_j\colon & b_j \pref c_j \pref a_j \pref \cdots,\\
&    & b_j\colon& e_j\pref f_j \pref d_j \pref \cdots, & \quad e_j\colon & c_j \pref a_j \pref b_j \pref \cdots,\\
    && c_j\colon& f_j\pref  d_j \pref e_j \pref \cdots, & \quad f_j\colon & a_j \pref b_j \pref c_j\pref \cdots\text{.}
  \end{alignat*}
  \end{compactenum}
  This completes the construction, which can clearly be done in polynomial time.
  We claim that profile~$\Pot$ with two layers has a \gstable{$2$} matching if and only if the new profile with $\ell$ layers has an \abstable{$\alpha$} matching.
  For the ``only if'' direction, it is straightforward to see that each \gstable{$2$} matching for $\Pot$ is \gstable{$\ell$} for the new profile.
  By \cref{prop:relation1}, $M$ is also \abstable{$\ell$} for the new instance.
  
  For the ``if'' direction, assume that $M$ is an \abstable{$\alpha$} matching for the new instance with $\ell$ layers.
  We claim that $M$ is \gstable{$2$} for instance~$\Pot$.
  First, for each unmatched pair~$\{u,w\}\notin M$, 
  let $S^{\text{unblock}}(\{u,v\})$ be 
  the set that consists of all layers that are \myemph{not} blocked by~$\{u,w\}$: $S^{\text{unblock}}(\{u,v\})\coloneqq \{i\in[\ell]\mid M(u)\succ^{(i)}_{u} w\text{ or } M(w)\succ^{(i)}_{w} u\}$.
  Since $M$ is \abstable{$\alpha$} it follows that $|S^{\text{unblock}}(\{u,v\})|\ge \alpha \ge \lceil \nicefrac{\ell}{2}\rceil+1$; the last inequality holds by our assumption on the value of~$\alpha$.
  We claim that $S^{\text{unblock}}(\{u,v\})$ contains at least $k+1$ layers from the first $2\cdot k$ layers.
  If $\ell$ is odd, then $\lceil\nicefrac{\ell}{2}\rceil+1 = k+2$, and the claim follows; otherwise $\ell=2k$ and $\lceil\nicefrac{\ell}{2}\rceil+1=k+1$, implying our claim.
  Since the first $2k$ layers are $k$ copies of the two layers from $\Pot$, we can assume without loss of generality that $S^{\text{unblock}}(\{u,v\})$ contains at least the first $k+1$ layers.
  Now, it is obvious that $\bigcap_{\{u,v\}\notin M}S^{\text{unblock}}(\{u,v\})$ contains at least the first two layers.
  This implies that no unmatched pair is blocking the first two layers, and hence that $M$ is \gstable{$2$} for the instance~$\Pot$.
\end{proof}

\cref{cor:hardness_pair_stable} and \cref{pro:pair-alpha_ge_ell/2} cover the whole range of the potential values of $\alpha$ except for \mbox{$\alpha = \lfloor \nicefrac{\ell}{2} \rfloor +1$} when $\ell$ is odd. As we will see in the next section \cref{thm:hardness_individual_stable} cannot be strengthened to cover the value $\alpha = \lfloor \nicefrac{\ell}{2} \rfloor +1$ (see \cref{prop:single-sided-individual-poly}), and so, also \cref{cor:hardness_pair_stable} cannot be directly strengthened. However, we can tweak the construction from \cref{thm:hardness_individual_stable}, breaking the restriction that on one side the preference list of each agent is the same in all layers, and obtain hardness for \abstability{$\alpha$} for $\alpha = \lfloor \nicefrac{\ell}{2} \rfloor+1$.   

\begin{proposition}\label{prop:hardness_pair_stable-alpha=ell/2+1}
  For each fixed odd number~$\ell$ of layers with $\ell\ge 5$ and for the case when $\alpha = \lfloor \nicefrac{\ell}{2} \rfloor+1$,
  \BSM is NP-hard.
\end{proposition}

% \begin{proposition}\label{prop:hardness_pair_stable_ell_half}
% \BSM is NP-hard for each odd value of $\ell\ge 5$ and $\alpha = \lceil \nicefrac{\ell}{2} \rceil$.
% \end{proposition}
\begin{proof}
We provide a reduction from an instance $I$ of the NP-hard \textsc{Perfect SMTI} problem which is very similar to the reduction given in the proof of \cref{thm:hardness_individual_stable}. Consequently, instead of describing the new reduction from scratch, we will only explain how it differs from the one from the proof of \cref{thm:hardness_individual_stable}.
For each agent $w_i \in W$ in the proof of  \cref{thm:hardness_individual_stable}, the preference list of $w_i$ was the same in all layers. 
Now, in the last $\ell-1$ layers this list is constructed in exactly the same way as in the proof of \cref{thm:hardness_individual_stable}. 
However, in the first layer we set the first part of the list to be reversed in comparison to the remaining $(\ell-1)$ layers; let $\overleftarrow{L_{w_i}}$ be the reverse of the strict preference list of $w$ in the input instance of \textsc{Perfect SMTI} and let the second~``$\cdots$'' denote some arbitrary but fixed order of the remaining agents:
\[\text{Layer } (1)\colon \text{agent } w_i\colon \overleftarrow{L_{w_i}}\pref p_{i,1} \pref p_{i,2} \pref \cdots \pref p_{i,\ell}\pref \cdots\text{.}\]
The preference lists of all other agents are constructed in exactly the same way as in the proof of \cref{thm:hardness_individual_stable}.
We prove that $I$ admits a perfect SMTI-stable matching if and only if the constructed instance $I'$ with $\ell$ layers admits an \abstable{$\alpha$} matching with $\alpha=\lfloor \nicefrac{\ell}{2}\rfloor+1$.

For the ``only if'' direction, assume that $M$ is a perfect SMTI-stable matching for $I$ and consider matching $M'=M\cup \big\{\{r_{i,j},p_{i,j}\} \mid 1\le i \le n\wedge 1\le j\le \ell\big\}$. We will show that $M'$ is \abstable{$\alpha$} for~$I'$. 
First of all, just as in the proof of \cref{thm:hardness_individual_stable} we deduce that no agent from~$P\cup R$ and no agent from $U\cup W$ will form a blocking pair. 
Neither will any two agents that are not acceptable to each other in $I$ form a blocking pair.
Now, consider two arbitrary agents~$u$ and $w$ with $u\in U$ and $w\in W$ that are acceptable to each other in $I$.
%By \cref{prop:pair-stable-alternative-def}, 
We need to find a subset of $\lfloor\nicefrac{\ell}{2}\rfloor+1$ layers, 
where in each layer in the subset $u$ prefers $M'(u)$ to~$w$ or $w$ prefers $M'(w)$ to~$u$ .
We distinguish between two cases concerning the preference list of~$w$ in~$I$; note that it does not have ties:
\begin{description}
\item[Case 1: $\boldsymbol{u \pref_w M'(w)=M(w)}$ in $\boldsymbol{I}$.] This implies that $w$ prefers $M'(w)$ to $u$ in the first layer in $I'$. 
Thus, it suffices if we can find $\lfloor \nicefrac{\ell}{2} \rfloor$ layers in the last $\ell-1$ layers, 
where $u$ prefers $M'(u)$ to $w$.
By the stability of $M$ we have that $M'(u)=M(u) \succeq_u w$ in instance~$I$. For the case that $M'(u)\succ_u w$ in $I$ we have that $u$ prefers $M'(u)$ to $w$ in all layers.
%Thus, $\{u,w\}$ does not block $M'$ in any layer in $I'$.
For the case that $M'(u)$ and $w$ are tied by $u$ we can identify $\lfloor \nicefrac{\ell}{2} \rfloor$ layers, either the layers from $2$ to $\lfloor \nicefrac{\ell}{2} \rfloor+1$ or the layers from $\lfloor \nicefrac{\ell}{2} \rfloor+2$ to $\ell$, where $u$ prefers $M'(u)$ to $w$. Additionally, in the first layer $w$ prefers $M(w)$ to $u$, which gives in total the subset of $\lfloor \nicefrac{\ell}{2} \rfloor+1$ layers.
% Now, let us consider the case when $M'(u)$ and $w$ are tied by $u$ in~$I$. 
% There are two subcases on the \myemph{first} layer in instance~$I'$.
% If $u$ prefers $M'(u)$ to $w$ in the first layer in $I'$, then by our construction, $\{u,w\}$ is dominated by $\{u,M(u')\}$ each layer in~$\{1, 2\dots, \lfloor\nicefrac{\ell}{2}\rfloor+1\}$. % and is dominated by $\{M'(w), w\}$ in the $\ell^{\text{th}}$ layer. 
% Otherwise, $u$ prefers $w$ to $M'(u)$ in the first layer.
% By our construction, $\{u,w\}$ is dominated by $\{u, M(u')\}$ in each layer in $\{\lfloor\nicefrac{\ell}{2}\rfloor+1,  \lfloor\nicefrac{\ell}{2}\rfloor+2,\dots, \ell\}$.
% In any case, there are $\lfloor\nicefrac{\ell}{2}\rfloor+1$~layers, where $\{u,w\}$ is not blocking $M'$.
% $\lceil \nicefrac{\ell}{2} \rceil$ layers, then again $\{u, w\}$ cannot be blocking. Otherwise, $u$ prefers $M'(u)$ to $w$ in the last $\lfloor \nicefrac{\ell}{2} \rfloor$ layers and, by our twaeked construction, $w$ prefers $M(w)$ to $u$ in the first layer. Thus, there are at least $\lceil \nicefrac{\ell}{2} \rceil$ layers where either $u$ or $w$ prefers their respective partner in $M$ to each other, a contradiction.

\item[Case 2: $\boldsymbol{M'(w) \pref_w u}$.] In this case we know that $w$ prefers $M'(w)$ to $u$ in the last $\ell-1$ layers.
\end{description}

For the ``if'' direction, assume that $M'$ is an \abstable{$\alpha$} matching for $I'$. By the same reasoning as in the proof of \cref{thm:hardness_individual_stable} we deduce that each agent $u \in U$ must be matched with someone that it prefers to each agent from $R$ in each layer,
and that $w \in W$ must be matched with an agent that it prefers to each agent from $P$ in all layers. 
Thus, $u$ and $M'(u)$ (resp.\ $w$ and $M'(w)$) must be acceptable to each other in $I$, meaning that $M=\{\{u,w\}\in M' \mid u\in U, w\in W\}$ is a perfect matching for $I$. We show that~$M$ is SMTI-stable. For the sake of contradiction, suppose that an unmatched pair $\{u, w\}$ is SMTI-blocking~$M$. This means that $w$ prefers $u$ to $M(w) = M'(w)$ in the last $\ell -1$ layers in $I'$, and that~$u$ prefers~$w$ to~$M(u) = M'(u)$ in all $\ell$ layers, a contradiction to~$M'$ being \abstable{$\alpha$}, since $\alpha\ge 2$.
\end{proof}

\section{Two Special Cases of Preferences}
\label{sec:constr-pref}

In this section we consider two well-motivated special cases of our general multi-layer framework for stable matchings. We will discuss how the corresponding simplifying assumptions affect the computational complexity of finding multi-layer stable matchings. Interestingly, even under seemingly strong assumptions some variants of our problem remain computationally hard.    

\subsection{Single-layered preferences on one side}

Consider the special case where the preferences of the agents from~$U$ can be expressed through a single layer. Formally, we model this by assuming that for each agent from $U$, 
its preference list is the same in all layers. 
In this case we say that the agents from~$U$ have \myemph{single-layered preferences}.

Single-layered preferences on one side can arise in many real-life scenarios. For instance, consider the standard example of matching residents with hospitals and, similarly as in \cref{ex:individual}, assume that each layer corresponds to a certain criterion. It is reasonable to assume that the hospitals evaluate their potential employees with respect to the level of their qualifications only (thus, having a single layer of preferences), while the residents take into account a number of factors such as how far is a given hospital from the place they live, the level of compensation, the reputation of the hospital, etc.
 
First, we observe that in such a case two out of our three solution concepts from \cref{sec:defi} are equivalent.

\begin{proposition}\label{prop:single_layer_equiv}
If each agent from~$U$ has single-layered preferences, then \abstability{$\alpha$} and \aistability{$\alpha$} are equivalent for each~$\alpha$.
\end{proposition}
\begin{proof}
The fact that \aistability{$\alpha$} implies \abstability{$\alpha$} follows from \cref{prop:relation2}.

Let us prove the other direction. Let $M$ be an \abstable{$\alpha$} matching.
Towards a contradiction, suppose that $M$ is not \aistable{$\alpha$}.
By \cref{prop:individual-stable-alternative-def}, let $p=\{u, w\}$ 
be an unmatched pair of agents 
such that there is a subset~$S_1$ of $\ell-\alpha+1$ layers where $p$ is dominating $\{u, M(u)\}$, and that there is a subset~$S_2$ of $(\ell-\alpha+1)$ layers where $p$ is dominating $\{w, M(w)\}$.
By our assumption of single-layered preferences for each agent in $U$, we have that $p$ is dominating $\{u,M(u)\}$ in every layer.
Thus, the pair~$p$ is blocking $M$ in each layer in $S_2$---a contradiction to \cref{prop:pair-stable-alternative-def}.
% . Thus, in each set $S$ of $\alpha$ layers it must hold that:
% \begin{inparaenum}[(1)]
% \item there exists a layer $i \in S$ such that $\{u,w\}$ dominates $\{u,M(u)\}$, and
% \item there exists a layer $j \in S$ such that $\{u,w\}$ dominates $\{w, M(w)\}$.
% \end{inparaenum}
% From the fact that the preferences of agents from $U$ are single-layer we deduce that $\{u, w\}$ dominates $\{u, M(u)\}$ in each layer. Consequently, we get that in each set $S$ of $\alpha$ layers there is at least one layer $i$ where $\{u, w\}$ dominates both $\{u,M(u)\}$ and $\{w,M(w)\}$. Thus, $\{u, w\}$ witnesses that $M$ is not \abstable{$\alpha$}, a contradiction.
\end{proof}

For profiles with single-layered preferences of the agents on one side, \gstability{$\alpha$} is strictly stronger than the other two concepts:

\begin{example}
Consider four agents with the following three layers of preference profiles:

\noindent 
{\centering
\begin{tikzpicture} 
  \node at (0, -0.5) (Layer1) {\Large $P_1$:};
  \foreach \i in {1, 2}
  {
    \foreach \j/\p/\o in {u/0/left,w/1/right} {
      \node[draw, circle, minimum size=3ex, inner sep=1pt] at (\i*\dist, -\p*\xsc) (n\j\i) {$\j_\i$};
    %  \node[\o = 0pt of \j\i]  (n\j\i) {$\j_\i$}; 
    }
  }

  \foreach \n / \i / \o / \a/\b   in {u1/w/left/1/2, u2/w/left/2/1}{
  % \node[\o = -7pt of n\n] {:};
    \gettikzxy{(n\n)}{\xx}{\yy};
    \node at (\xx,\yy+\ysc*25) {$\i_\a$};
    \node at (\xx,\yy+\ysc*15) {$\i_\b$};
  }  
  \foreach \n / \i / \o / \a/\b  in {%
    w1/u/right/2/1, w2/u/right/1/2}{
  % \node[\o = -7pt of n\n] {:};
    \gettikzxy{(n\n)}{\xx}{\yy};
    \node at (\xx,\yy-\ysc*18) {$\i_\a$};
    \node at (\xx,\yy-\ysc*28) {$\i_\b$};
  }
  \foreach \s/\t in {u1/w1,u2/w2}{
    \draw[blackline] (n\s) -- (n\t);
  } 
  \foreach \s/\t in {u1/w2,u2/w1} {
    \draw[greenline] (n\s) -- (n\t);
  }
\end{tikzpicture}
~~\qquad
\begin{tikzpicture}
  \node at (0, -0.5) (Layer1) {\Large $P_2$:};
  \foreach \i in {1, 2}
  {
    \foreach \j/\p/\o in {u/0/left,w/1/right} {
      \node[draw, circle, minimum size=3ex, inner sep=1pt] at (\i*\dist, -\p*\xsc) (n\j\i) {$\j_\i$};
    %  \node[\o = 0pt of \j\i]  (n\j\i) {$\j_\i$}; 
    }
  }

  \foreach \n / \i / \o / \a/\b   in {u1/w/left/1/2, u2/w/left/2/1}{
  % \node[\o = -7pt of n\n] {:};
    \gettikzxy{(n\n)}{\xx}{\yy};
    \node at (\xx,\yy+\ysc*25) {$\i_\a$};
    \node at (\xx,\yy+\ysc*15) {$\i_\b$};
  }  
  \foreach \n / \i / \o / \a/\b  in {%
    w1/u/right/2/1, w2/u/right/2/1}{
  % \node[\o = -7pt of n\n] {:};
    \gettikzxy{(n\n)}{\xx}{\yy};
    \node at (\xx,\yy-\ysc*18) {$\i_\a$};
    \node at (\xx,\yy-\ysc*28) {$\i_\b$};
  }
  \foreach \s/\t in {u1/w1,u2/w2}{
    \draw[blackline] (n\s) -- (n\t);
  } 
\end{tikzpicture}
~~\qquad
\begin{tikzpicture}
  \node at (0, -0.5) (Layer1) {\Large $P_3$:};
  \foreach \i in {1, 2}
  {
    \foreach \j/\p/\o in {u/0/left,w/1/right} {
      \node[draw, circle, minimum size=3ex, inner sep=1pt] at (\i*\dist, -\p*\xsc) (n\j\i) {$\j_\i$};
    %  \node[\o = 0pt of \j\i]  (n\j\i) {$\j_\i$}; 
    }
  }

  \foreach \n / \i / \o / \a/\b   in {u1/w/left/1/2, u2/w/left/2/1}{
  % \node[\o = -7pt of n\n] {:};
    \gettikzxy{(n\n)}{\xx}{\yy};
    \node at (\xx,\yy+\ysc*25) {$\i_\a$};
    \node at (\xx,\yy+\ysc*15) {$\i_\b$};
  }  
  \foreach \n / \i / \o / \a/\b  in {%
    w1/u/right/1/2, w2/u/right/1/2}{
  % \node[\o = -7pt of n\n] {:};
    \gettikzxy{(n\n)}{\xx}{\yy};
    \node at (\xx,\yy-\ysc*18) {$\i_\a$};
    \node at (\xx,\yy-\ysc*28) {$\i_\b$};
  }
  \foreach \s/\t in {u1/w1,u2/w2}{
    \draw[blackline] (n\s) -- (n\t);
  } 
\end{tikzpicture}
\par}
\noindent
Let $M = \{\{u_1, w_2\}, \{u_2, w_1\}\}$.
We show that $M$ is \abstable{2} but not \gstable{2}.
Indeed, $M$ is stable only in the first layer, so it cannot be \gstable{2}.
To see why $M$ is \abstable{2}, consider the unmatched pairs $\{u_1, w_1\}$ and $\{u_2, w_2\}$.
The first one only blocks the third layer, and the latter one only blocks the second layer.
\hfill $\diamond$
\end{example}

\subsubsection{Global stability}
By \cref{prop:g_stability_equiv_pair_stability,prop:single_layer_equiv}, we can find out whether an instance with single-layered preferences on one side admits an \gstable{$\alpha$} matching in time $O(\ell^{\alpha}\cdot \alpha \cdot n^2)$ by guessing a subset of $\alpha$ layers and using \cref{alg:ell-individual}.
However, the following result shows that fixed-parameter tractability (FPT) for the parameter~$\alpha$, 
\emph{i.e.}, the existence of an algorithm running in $f(\alpha) \cdot (\ell\cdot n)^{O(1)}$ time for some computable function~$f$, is unlikely (for details on parameterized complexity we refer to the books of \citet{CyFoKoLoMaPiPiSa2015}, \citet{DF13},  \citet{FG06}, and \citet{Nie06}).

\begin{theorem}\label{thm:single-sided-global-w-hard}
  Even if the preferences of the agents from $U$ are single-layered \GSM is NP-hard and is W[1]-hard for the threshold parameter~$\alpha$.
  It can be solved in $O(\ell^{\alpha}\cdot \alpha \cdot n^2)$~time.
\end{theorem}
\begin{proof}
  For the running time statement, let $I$ be an instance of \GSM with $\ell$ layers.
  We guess a subset~$S\subseteq [\ell]$ of $\alpha$ layers and check whether a matching~$M$ exists that is stable in all layers of $S$.
  Now consider the instance~$I'$ restricted to the $\alpha$~layers in $S$.
  By \cref{prop:g_stability_equiv_pair_stability}, $M$ is \gstable{$\alpha$} in $I'$ if and only if $M$ is \abstable{$\alpha$} in $I'$.
  Since each agent on one side has the same preference list in all layers, 
  by \cref{prop:single_layer_equiv}, $M$ is \abstable{$\alpha$} in $I'$ if and only if  $M$ is \aistable{$\alpha$} in $I'$.
  By \cref{thm:l_individual_polynomial}, we can use \cref{alg:ell-individual} to decide whether $M$ is \aistable{$\alpha$} in $I'$ in $O(\alpha\cdot n^2)$~time.
  The total running time is thus $O(\ell^{\alpha}\cdot \alpha\cdot n^2)$.

  For the last statement, we provide a parameterized reduction\footnote{A \emph{parameterized reduction} from a parameterized problem~$\Pi_1$
to a parameterized problem~$\Pi_2$ is an algorithm mapping an instance~$(x,k)$
to an instance~$(x',k')$
in $f(k)\cdot|x|^{O(1)}$ time
such that $k'\leq g(k)$,
and that
$(x,k)\in\Pi_1$ if and only if $(x',k')\in\Pi_2$,
where \(f\)~and~\(g\) are some computable functions.}
 from the W[1]-complete \IS problem parameterized by the size~$k$ of the solution.
 We will see that the reduction is also a polynomial-time reduction, showing the first statement since \IS is also NP-hard.

  In the \IS problem we are given an undirected graph~$G=(V,E)$ with vertex set~$V$ and edge set~$E$, and a non-negative integer~$k$, and we ask whether $G$ admits a $k$-vertex \myemph{independent set}, \emph{i.e.}\ a vertex subset~$V'\subseteq V$ with $k$~pairwisely non-adjacent vertices.
  
  Given an \IS instance~$(G=(V,E), k)$, we construct an instance with $|V|$ layers with the set of agents~$U\uplus W$ as follows.
  %We will make each agent from $U$ have the same preferences in all $|V|$ layers.
  Assume that~$V=\{v_1, v_2, \ldots, v_n\}$.
  For each vertex~$v_i\in V$ we construct six agents, three for each side, denoted by $u_i$, $\overline{u}_i$, $w_i$, $\overline{w}_i$, $a_i$ and $b_i$.
  Let $U=\{u_i, \overline{u}_i, a_i \mid 1\le i \le n\}$ and $W=\{w_i, \overline{w}_i, b_i \mid 1\le i \le n\}$.
  We create a layer for each vertex so that if a matching is stable for a layer, then 
  an independent set solution has to include the corresponding vertex and exclude any of its adjacent vertices. 
  Again, the notation~``$\cdots$'' denotes an arbitrary order of the unmentioned agents.
  \paragraph{Agents from \boldmath$U$.}
    The preference list of each agent in $U$ is the same for all layers.
  \begin{align*}
    \forall i \in [n]\colon
    \text{agent } u_i \colon& \overline{w}_i \pref w_i \pref \cdots ,\\
    \text{agent } \overline{u}_i \colon& w_i \pref \overline{w}_i \pref \cdots,\\
    \text{agent } a_i \colon& w_i \pref b_i\pref \cdots.
  \end{align*}
  \paragraph{Agents from \boldmath$W$.}
  For each layer~$i$, $1\le i \le n$, the preference list of agents~$w_i$ and $\overline{w}_i$ and of agents~$w_j$ and $\overline{w}_j$ for $v_j \neq v_i$ are constructed in such a way that they encode the adjacency structure of graph~$G$.
   \begin{alignat*}{4}
     &\text{agent } w_i &\colon& u_i \pref a_i \pref \cdots, &\\
     & \text{agent } \overline{w}_i &\colon& \overline{u}_i \pref \cdots,&\\
     \forall j \text{ with } \{v_i, v_j\} \in E\colon& \text{agent }w_j &\colon& \overline{u}_j \pref \cdots, &\\
    & \text{agent }\overline{w}_j &\colon& u_j \pref \cdots, &\\
     \forall k \text{ with } \{v_i, v_k\} \notin E\colon& \text{agent }w_k &\colon& {u}_k \pref \overline{u}_k\pref \cdots, &\\
    & \text{agent }\overline{w}_k &\colon& \overline{u}_k \pref u_k \pref \cdots\text{.}
  \end{alignat*}
  The preference list of each $b_j$ with $j\in [n]$ remains the same in all layers and is as follows.
  \begin{align*}
    \text{agent } b_j \colon a_j\pref \cdots\text{.} 
  \end{align*}
  Observe that for each layer~$i \in [n]$, $w_i$ is the most preferred agent of $a_i$. 
  Thus, a matching~$M$ is stable in layer~$i$ only if $M(w_i)\in \{u_i, a_i\}$. From this we infer that for $M$ to be stable in layer~$i$ it must also hold that $M(\overline{u}_i)=\overline{w}_i$.
  This implies that $M(w_i)=u_i$ as otherwise $\{u_i, w_i\}$ would be blocking layer~$i$.
  Further, if $M$ is stable in layer~$i$, then for each vertex~$v_j$ that is adjacent to $v_i$,
  it must hold that $M(w_j)=\overline{u}_j$ and $M(\overline{w}_j)=u_j$.
  In the following, for each $i\in [n]$, let $M_i^{\text{vc}}=\{\{u_i, w_i\}, \{\overline{u}_i, \overline{w}_i\}\}$ and $M_i^{\text{ind}}=\{\{u_i, \overline{w}_i\}, \{\overline{u}_i, {w}_i\}\}$.
  Then, we have the following observation for the case that a matching~$M$ is stable in layer~$i$:
  \begin{align}\label{eq:single-layered-global-matching}
    %   \text{If a matching~} M \text{ is stable in layer } i \text{, then }\\
   M_i^{\text{vc}} \subseteq M \text{ and for each vertex~} v_j \text{ incident with } \{v_i,v_j\} \in E \text{ it holds that }M^{\text{ind}}_j \subseteq M\text{.}
  \end{align}
  \noindent Intuitively, the matching~$M_i^{\text{ind}}$ means that vertex~$v_j$ does not belong to an independent set while the matching~$M_i^{\text{vc}}$ means that vertex~$v_i$ belongs to an independent set.

  To complete the construction we let $\alpha=k$.
  Clearly, the construction can be done in polynomial time.
  Now, we claim that the given graph~$G$ has an independent set of size~$k$ if and only if the constructed instance has an \gstable{$\alpha$} matching.

  For the ``only if'' part, assume that $V'\subseteq V$ is a size-$k$ independent set for $G$.
  We show that the matching~$M= \bigcup_{v_i \in V'} M^{\text{vc}}_i \cup \bigcup_{v_j \in V\setminus V'}M^{\text{ind}}_j \cup \{\{a_i, b_i\} \mid i\in [n]\}$ is stable in each layer~$i$ with $v_i \in V'$.
  Suppose, for the sake of contradiction, that there is a layer~$i$ with $v_i \in V'$ and 
  there is an unmatched pair~$p=\{x,y\}$ with $x\in U$ and $y\in W$ that is blocking layer~$i$, \emph{i.e.}\ the following holds.
  \begin{align*}
    \text{In layer }i \text{, agent } x \text{ prefers } y \text{ to } M(x) \text{ and } \text{agent } y \text{ prefers } x \text{ to } M(y).
  \end{align*}
  First, $y\notin \{b_j \mid j\in [n]\}$ since $M(b_j)=a_j$ is the most preferred agent of $b_j$.
  Second, $y\notin \{w_j, \overline{w}_j\mid \{v_i, v_j\}\in E\}$ as both $w_j$ and $\overline{w}_j$ already obtain their most preferred agents, namely $\overline{u}_j$ and $u_j$ in layer~$i$.
  Third, $y\notin \{w_j, \overline{w}_j\mid v_j \in V'\}$, as in this way $w_j$ and $\overline{w}_j$ obtain $u_j$ and $\overline{u}_j$ as their partners, and since $V'$ is an independent set, we have that $\{v_i, v_j\}\notin E$, and so $u_j$ and $\overline{u}_j$ are the most preferred partners of $w_j$ and $\overline{w}_j$.

  Now, the only possible choice for $y$ would be agent~$w_j$ or $\overline{w}_j$ such that $\{v_i, v_j\}\notin E$ and $v_j \notin V'$.
  By our construction of $M$ we have that $M(w_j)=\overline{u_j}$ and $M(\overline{w}_j)=u_j$. First, consider the case when $y = w_j$.
  By the preference list of $w_j$ in layer~$i$, if $\{x,y\}$ is a blocking pair, then $x=u_j$.
  However, $u_j$ already obtains its most preferred agent, namely, $\overline{w}_j$; a contradiction.
  Analogously, the case when $y=\overline{w}_j$ also results in a contradiction.
  Summarizing, matching~$M$ is stable in each layer~$i$ with $v_i \in V'$.
  
  For the ``if'' part, assume that our constructed instance admits an \gstable{$\alpha$} matching~$M$.
  We show that the vertex subset~$V'\coloneqq\{v_i \in V\mid M \text{ is stable in layer~}i\}$ is an independent set of size~$k$.
  Clearly, $|V'|=\alpha=k$.
  Suppose, for the sake of contradiction, that $V'$ contains two adjacent vertices~$v_i$ and $v_j$ with $\{v_i, v_j\}\in E$.
  Since $M$ is stable in layer~$i$ and $\{v_i,v_j\}\in E$, 
  by our observation~(\ref{eq:single-layered-global-matching}) we deduce that $M_i^{\text{vc}}\subseteq M$ and $M_j^{\text{ind}} \subseteq M$.
  However, this is a contradiction to $M$ being stable in layer~$j$.
\end{proof}

\subsubsection{Individual stability and pair stability}
Observe that, by \Cref{thm:hardness_individual_stable}, we know that finding an \aistable{$\alpha$} matching with preferences single-layered on one side is NP-hard if $\ell \geq 4$ and $\alpha \leq \lfloor \nicefrac{\ell}{2}\rfloor$.

Next, we establish a relation between the individual (and thus pair) stability for preferences single-layered on one side, and the stability concept in the traditional single-layer setting, but with general (incomplete and possibly intransitive) preferences as studied by Farczadi~et~al.~\cite{FarGeoKoe2016}. This relation allows us to construct a polynomial-time algorithm for finding an \aistable{$\alpha$} (and hence \abstable{$\alpha$}) matching with single-layer preferences on one side when $\alpha \geq \lfloor \nicefrac{\ell}{2}\rfloor + 1$.

\begin{proposition}\label{prop:single-sided-individual-poly}
  % Given an instance~$I$ with $\ell$ layers, where the preferences of the agents from $U$ are single-layer,
  % deciding whether $I$ has a \aistable{$\alpha$} (and thus \abstable{$\alpha$}) matching for $\alpha \geq \lfloor\nicefrac{\ell}{2}\rfloor + 1$ can be solved in time~$O(\ell\cdot n^2)$.
  If the preferences of the agents on one side are single-layered and $\alpha\ge \lfloor \nicefrac{\ell}{2} \rfloor+1$, then \BSM and \ISM can be solved in $O(\ell\cdot n^2)$~time.
\end{proposition}

\begin{proof}
We will reduce our problem to the \textsc{Stable Marriage with General Preferences (SMG)} problem~\cite{FarGeoKoe2016}. Let $I$ be an instance of our problem, let $U \cup W$ denote the set of agents in $I$, and let $\ell$ be the number of layers. We construct an instance $I'$ of \textsc{SMG} as follows. In $I'$ we have the same set of agents as in~$I$. For each agent $u \in U$ we copy the preferences from $I$ to $I'$ (such an agent has the same preference list in all layers). 
The preferences of the agents from $W$ can be arbitrary binary relations on the set~$V$.
Thus, we use an ordered pair~$(u,u')$ of two agents to denote that $u$ is regarded as good as $u'$.
Formally, for each agent~$w \in W$ we define the general preferences of $w$, denoted as $R_w$, as follows:
For each two agents~$u, u' \in U$ we let $(u,u')\in R_w$ if and only if $w$ prefers $u$ to $u'$ in at least $\alpha$ layers. 
First, by our assumption on the value of $\alpha$, we observe that the preferences of the agents from $W$ are asymmetric as, \emph{i.e.}, for each pair of agents $u, u' \in U$ we have that $|\{(u,u'), (u',u)\}\cap R_w| \le 1$.

Now, we will prove that a matching $M$ is \abstable{$\alpha$} in $I$ if and only if it is stable in $I'$ according to the definition of stability by \citet{FarGeoKoe2016}. 
Indeed, according to \citet{FarGeoKoe2016} a pair $\{u, w\}$ is \myemph{SMG-blocking} 
if and only if the following two conditions hold:
\begin{inparaenum}[(1)] \item $w \succ_u M(u)$ and \item $(M(w),u)\notin R_w$. 
\end{inparaenum} 
The first condition is equivalent to saying that $u$ prefers $w$ to $M(u)$ in each layer in $I$. The second condition holds if and only if $w$ prefers $M(w)$ to $u$ in less than $\alpha$ layers. This is equivalent to saying that $w$ prefers $u$ to $M(w)$ in at least $\ell - \alpha + 1$ layers, and so the two conditions are equivalent to saying that $\{u, w\}$ is $(\ell - \alpha + 1)$-blocking $M$ in $I$.  

For asymmetric preferences \textsc{SMG} can be solved in $O(n^2)$ time~\cite[Theorem 2]{FarGeoKoe2016}, where the number of agents is $2 n$.
The reduction to an SMG instance takes $O(\ell \cdot n^2)$~time.
The number of agents in $I$ equals the number of agents in $I'$.
Thus, the problem (deciding \abstability{$\alpha$}) can be solved in $O(\ell \cdot n^2)$~time.
By \cref{prop:single_layer_equiv}, we obtain the same result for \aistability{$\alpha$}.
\end{proof}

\subsection{Uniform preferences in each layer}

In this section we consider the case when for each layer the preferences of all agents from $U$ (resp.\ all agents from $W$) are the same---we call such preferences \myemph{uniform in each layer}. This special case is motivated with the following observation pertaining to \cref{ex:individual}: preferences uniform in a layer can arise if the criterion corresponding to the layer is not subjective. 
For instance, if a layer corresponds to the preferences regarding the wealth of potential partners, it is natural to assume that everyone prefers to be matched with a wealthier partner (and so the preferences of all agents for this criterion are the same); similarly it is natural to assume that the preferences of all hospitals are the same: candidates with higher grades will be preferred by each hospital.   

First, we observe that for uniform preferences no two among the three concepts are equivalent.

\begin{example}\label{ex:uniform}
Consider four agents with the following four layers of uniform preferences:

\noindent 
{\centering
\begin{tikzpicture} 
  \node at (0, -0.5) (Layer1) {\Large $P_1$:};
  \foreach \i in {1, 2}
  {
    \foreach \j/\p/\o in {u/0/left,w/1/right} {
      \node[draw, circle, minimum size=3ex, inner sep=1pt] at (\i*\dist, -\p*\xsc) (n\j\i) {$\j_\i$};
    %  \node[\o = 0pt of \j\i]  (n\j\i) {$\j_\i$}; 
    }
  }

  \foreach \n / \i / \o / \a/\b   in {u1/w/left/1/2, u2/w/left/1/2}{
  % \node[\o = -7pt of n\n] {:};
    \gettikzxy{(n\n)}{\xx}{\yy};
    \node at (\xx,\yy+\ysc*25) {$\i_\a$};
    \node at (\xx,\yy+\ysc*15) {$\i_\b$};
  }  
  \foreach \n / \i / \o / \a/\b  in {%
    w1/u/right/1/2, w2/u/right/1/2}{
  % \node[\o = -7pt of n\n] {:};
    \gettikzxy{(n\n)}{\xx}{\yy};
    \node at (\xx,\yy-\ysc*18) {$\i_\a$};
    \node at (\xx,\yy-\ysc*28) {$\i_\b$};
  }
  \foreach \s/\t in {u1/w1,u2/w2}{
    \draw[blackline] (n\s) -- (n\t);
  } 
  % \foreach \s/\t in {u1/w2,u2/w1} {
  %   \draw[greenline] (n\s) -- (n\t);
  % }
\end{tikzpicture}
~~\qquad
\begin{tikzpicture}
  \node at (0, -0.5) (Layer1) {\Large $P_2$:};
  \foreach \i in {1, 2}
  {
    \foreach \j/\p/\o in {u/0/left,w/1/right} {
      \node[draw, circle, minimum size=3ex, inner sep=1pt] at (\i*\dist, -\p*\xsc) (n\j\i) {$\j_\i$};
    %  \node[\o = 0pt of \j\i]  (n\j\i) {$\j_\i$}; 
    }
  }

  \foreach \n / \i / \o / \a/\b   in {u1/w/left/1/2, u2/w/left/1/2}{
  % \node[\o = -7pt of n\n] {:};
    \gettikzxy{(n\n)}{\xx}{\yy};
    \node at (\xx,\yy+\ysc*25) {$\i_\a$};
    \node at (\xx,\yy+\ysc*15) {$\i_\b$};
  }  
  \foreach \n / \i / \o / \a/\b  in {%
    w1/u/right/2/1, w2/u/right/2/1}{
  % \node[\o = -7pt of n\n] {:};
    \gettikzxy{(n\n)}{\xx}{\yy};
    \node at (\xx,\yy-\ysc*18) {$\i_\a$};
    \node at (\xx,\yy-\ysc*28) {$\i_\b$};
  }
  \foreach \s/\t in {u1/w2,u2/w1} {
     \draw[blackline] (n\s) -- (n\t);
   }
\end{tikzpicture}
~~\qquad
\begin{tikzpicture}
  \node at (0, -0.5) (Layer1) {\Large $P_3$:};
  \foreach \i in {1, 2}
  {
    \foreach \j/\p/\o in {u/0/left,w/1/right} {
      \node[draw, circle, minimum size=3ex, inner sep=1pt] at (\i*\dist, -\p*\xsc) (n\j\i) {$\j_\i$};
    %  \node[\o = 0pt of \j\i]  (n\j\i) {$\j_\i$}; 
    }
  }

  \foreach \n / \i / \o / \a/\b   in {u1/w/left/2/1, u2/w/left/2/1}{
  % \node[\o = -7pt of n\n] {:};
    \gettikzxy{(n\n)}{\xx}{\yy};
    \node at (\xx,\yy+\ysc*25) {$\i_\a$};
    \node at (\xx,\yy+\ysc*15) {$\i_\b$};
  }  
  \foreach \n / \i / \o / \a/\b  in {%
    w1/u/right/1/2, w2/u/right/1/2}{
  % \node[\o = -7pt of n\n] {:};
    \gettikzxy{(n\n)}{\xx}{\yy};
    \node at (\xx,\yy-\ysc*18) {$\i_\a$};
    \node at (\xx,\yy-\ysc*28) {$\i_\b$};
  }
   \foreach \s/\t in {u1/w2,u2/w1} {
     \draw[blackline] (n\s) -- (n\t);
   }
 \end{tikzpicture}
 ~~\qquad
 \begin{tikzpicture} 
  \node at (0, -0.5) (Layer1) {\Large $P_4$:};
  \foreach \i in {1, 2}
  {
    \foreach \j/\p/\o in {u/0/left,w/1/right} {
      \node[draw, circle, minimum size=3ex, inner sep=1pt] at (\i*\dist, -\p*\xsc) (n\j\i) {$\j_\i$};
    %  \node[\o = 0pt of \j\i]  (n\j\i) {$\j_\i$}; 
    }
  }

  \foreach \n / \i / \o / \a/\b   in {u1/w/left/2/1, u2/w/left/2/1}{
  % \node[\o = -7pt of n\n] {:};
    \gettikzxy{(n\n)}{\xx}{\yy};
    \node at (\xx,\yy+\ysc*25) {$\i_\a$};
    \node at (\xx,\yy+\ysc*15) {$\i_\b$};
  }  
  \foreach \n / \i / \o / \a/\b  in {%
    w1/u/right/2/1, w2/u/right/2/1}{
  % \node[\o = -7pt of n\n] {:};
    \gettikzxy{(n\n)}{\xx}{\yy};
    \node at (\xx,\yy-\ysc*18) {$\i_\a$};
    \node at (\xx,\yy-\ysc*28) {$\i_\b$};
  }
  \foreach \s/\t in {u1/w1,u2/w2}{
    \draw[blackline] (n\s) -- (n\t);
  } 
  % \foreach \s/\t in {u1/w2,u2/w1} {
  %   \draw[greenline] (n\s) -- (n\t);
  % }
\end{tikzpicture}
\par}
\noindent
Let $M_1=\{\{u_1,w_1\},\{u_2,w_2\}\}$ and $M_2=\{\{u_1,w_2\},\{u_2,w_1\}\}$.
We observe that $M_1$ is stable only in layers~$(1)$ and $(4)$ whereas $M_2$ is stable only in layers~$2$ and $3$.
Thus, neither is \gstable{$3$}. 
One can check that neither is \aistable{$3$}.
However, both $M_1$ and $M_2$ are \abstable{$3$}.
To see why this is the case we observe that for $M_1$, both unmatched pairs~$\{u_1,w_2\}$ and $\{u_2,w_1\}$ are each blocking only one layer, 
namely layers~$2$ and $3$, respectively.
Moreover, if we restrict the instance to only the first three layers, then we have the following results:
\begin{compactenum} 
\item $M_1$ is not even \gstable{$2$} but it is \aistable{$2$} as
for each unmatched pair~$p$ with respect to $M_1$, at least one agent in $p$ obtains a partner which is its most preferred agent in at least two layers. 
\item $M_2$ is \gstable{$2$} (it is stable in layers~$2$ and $3$) but it is not \aistable{$2$}.
To see why it is not \aistable{$2$}, we consider the unmatched pair~$\{u_1, w_1\}$.
There are $3-2+1=2$ layers where $u_1$ prefers $w_1$ to its partner~$M_2(u_1)=w_2$ and there are two layers where $w_1$ prefers $u_1$ to its partner~$M_2(w_1)$.
\end{compactenum}
\hfill $\diamond$
\end{example}

\subsubsection{Individual stability}

For preferences that are uniform in each layer,
we find that there is a close relation between \ISM and \textsc{Graph Isomorphism}, the problem of finding an isomorphism between graphs or deciding that there exists none.
Herein, an instance of \textsc{Graph Isomorphism} consists of two undirected graphs~$G$ and $H$ with the same number of vertices and the same number of edges.
We want to decide whether there is an edge-preserving bijection~$f\colon V(G)\to V(H)$ between the vertices, \emph{i.e.}\
for each two vertices~$u,v\in V(G)$ it holds that $\{u,v\}\in E(G)$ if and only if $\{f(u), f(v)\}\in E(H)$.
We call such bijection an isomorphism between $G$ and~$H$.

We explore this relation through the following construction. For an instance $I$ of \ISM we construct two directed graphs $G_I$ and $H_I$ as follows. The agents from $U$ and $W$ will correspond to the vertices in $G_I$ and $H_I$, respectively. For each two vertices $u, u' \in U$ we add to $G_I$ an arc~$(u,u')$ if the agents from $W$ prefer $u$ to $u'$ in at least $(\ell - \alpha + 1)$ layers.
Analogously, for each two agents~$w$ and $w'$, we add to $H_I$  an arc~$(w,w')$ if the agents from $U$ prefer $w$ to $w'$ in at least $(\ell - \alpha + 1)$ layers.
Let $E(G_I)$ and $E(H_I)$ denote the arc sets of $G_I$ and $H_I$, respectively.

We first explain how the so constructed graphs can be used to find an \aistable{$\alpha$} matching in the initial instance $I$, or to claim there is no such a matching.

\begin{proposition}\label{thm:uniform_prefs_isomorphism}
A matching $M$ for instance~$I$ is \aistable{$\alpha$} if and only if the following two properties hold.
\begin{compactenum}
  \item For each two vertices~$u, u' \in U$, it holds that: $(u,u')\in E(G_I)$ implies $(M(u'), M(u)) \notin E(H_I)$.
  \item For each two vertices~$w,w' \in W$, it holds that: $(w,w')\in E(H_I)$ implies $(M(w'), M(w)) \notin E(G_I)$.
\end{compactenum}
\end{proposition}
\begin{proof}
For the ``only if'' direction, assume that $M$ is \aistable{$\alpha$}.
Towards a contradiction, suppose that one of both properties stated above does not hold.
That is, there exist two vertices $u, u' \in U$ such that $(u,u')\in E(G_I)$ and $(M(u'), M(u))\in E(H_I)$ or there exist two vertices~$w,w'\in W$ such that $(w,w')\in E(H_I)$ and $(M(w'), M(w))\in E(G_I)$.
For the first case, 
it follows that there is a subset~$S$ of at least $\ell-\alpha+1$ layers such that all agents (including $M(u')$) in $W$ prefer $u$ to $u'$ in each layer of $S$,
and there is a (possibly different) subset~$R$ of at least $\ell-\alpha+1$ layers such that all agents (including $u$) in $U$ prefer $M(u')$ to $M(u)$ in each layer of $R$.
This means that the pair~$\{u, M(u')\}$ is $(\ell-\alpha+1)$-dominating $\{u', M(u')\}$ and is $(\ell-\alpha+1)$-dominating~$\{u, M(u)\}$---a contradiction to \cref{prop:individual-stable-alternative-def}.
%Thus, there are \myemph{no} $\alpha$ layers~$T$, such that $\{u', M(u')\}$ dominates $\{u,M(u')\}$ in each layer of $T$ or $\{u,M(u)\}$ dominates $\{u,M(u')\}$ in each layer of $T$, a contradiction to $M$ being \aistable{$\alpha$} regarding the unmatched pair~$\{u,M(u')\}$.
Analogously, the second case also leads to a contradiction to \cref{prop:individual-stable-alternative-def} regarding the unmatched pair~$\{M(w'), w\}$.
% there is an edge $x \to y$ in $G_I$ and an edge $M(y) \to M(x)$ in $H_I$ (the other case, when $x, y \in W$, can be handled analogously). We argue that $\{x, M(y)\}$ witnesses that $M$ is not \aistable{$\alpha$}. Indeed, since $M(y) \to M(x) \in H_I$, we infer that $x$ prefers $M(y)$ to $M(x)$ ($\{x, M(y)\}$ dominates $\{x, M(x)\}$) in at least $(\ell - \alpha + 1)$ layers and so there exists no $\alpha$ layers where $\{x, M(x)\}$ would dominate $\{x, M(y)\}$. Analogously, there exists no $\alpha$ layers where $\{y, M(y)\}$ would dominate $\{x, M(y)\}$, a contradiction.

For the ``if'' direction, assume that both properties hold. 
Towards a contradiction and by \cref{prop:individual-stable-alternative-def}, suppose that there is an unmatched pair, 
% $\Leftarrow$ Assume that there is a pair,
call it $\{u, w\}$ with $u\in U$ and $w\in W$, which is $(\ell-\alpha+1)$-dominating $\{u, M(u)\}$
and is $(\ell-\alpha+1)$-dominating $\{w, M(w)\}$.
By our construction of $G_I$ and $H_I$, it follows that $(w, M(u))\in E(H_I)$ and $(u, M(w))\in E(G_I)$---a contradiction to the second property.
% We distinguish between two cases, depending on whether $(u, M(w))$ exists in $G_I$.
% First, consider the case that $(u,M(w)) \notin E(G_I)$. 
% This means that there are less than $(\ell - \alpha + 1)$ layers where $\{u, w\}$ dominates $\{M(w), w\}$. Consequently, there are at least $\alpha$ layers where $\{M(w), w\}$ dominates $\{u, w\}$, and so $\{u, w\}$ cannot witness that $M$ is not \aistable{$\alpha$}, a contradiction.
% Next, consider the case that $(u, M(w)) \in E(G_I)$; by our assumption we get that $(w, M(u)) \notin E(H_I)$. As before, we deduce that there are less than $(\ell - \alpha + 1)$ layers where $\{u, w\}$ dominates $\{u, M(u)\}$, and so $\{u, M(u)\}$ dominates $\{u, w\}$ in at least $\alpha$ layers, a contradiction to $\{u,w\}$ being a witness that $M$ is not \aistable{$\alpha$}.
\end{proof}

Using a construction by \citet{McG1953}, given an arbitrary directed graph, we can indeed construct multi-layer preferences that induce this graph.
\begin{lemma}\label{lemma:uniform_prefs_isomorphism_complete}
For each two directed graphs $G_I$ and $H_I$ with $m$ arcs each there exists an instance with $\ell=2m$ layers of preferences that induces $G_I$ and $H_I$ via McGarvey's construction, where $m$ denotes the number of arcs in $G_I$ (and thus in $H_I$). 
\end{lemma}
\begin{proof}
%  We aim to construct an instance with $\ell=2m$ layers such that the preference lists of all agents on each side are the same in each layer, where $m$ denotes the number of arcs in $G_I$ (and thus in $H_I$).
%  Recall that we require the preference lists of all agents on each side are the same in each layer.
  Our profile uses the vertex sets of $G_I$ and $H_I$ as the two disjoint subsets of agents, denoted as $U$ and~$W$, respectively. 
  For each arc in the graphs~$G_I$ (resp.\ $H_I$) we construct two layers of preference lists for the agents from $W$ (resp.\ $U$) to ``encode'' it.
  We describe how to construct the preference lists for an arc in $G_I$, say $(u,u')$.
  The construction for the arcs in $H_I$ works analogously.
  To encode the arc~$(u,u')$, we construct two layers where the preference lists for the agents in $W$ are as follows; denote by $\overrightarrow{X}$ some arbitrary but fixed order of the agents other than $u$ and $u'$ and let $\overleftarrow{X}$ be the reverse of $\overrightarrow{X}$.
  \begin{align*}
   & \text{One layer}\colon u \pref u' \pref \overrightarrow{X}\text{.}\\
   & \text{The other layer}\colon \overleftarrow{X} \pref u \pref u'\text{.}
  \end{align*}
  Finally, we set $\alpha = \nicefrac{\ell}{2}=m$.
  Note that for each two agents~$u$ and $u'$ in $G_I$ all agents from $W$ prefer $u$ to~$u'$ in either exactly $m-2$ layers, or exactly $m$ layers, or exactly $m+2$ layers.
  It is straightforward to see that our constructed instance induces the two input graphs $G_I$ and~$H_I$.
\end{proof}
%\todo[inline]{Hua: Can we say something for constant ell and arbitrary alpha?}

\begin{figure}[t!]
  \begin{center}
    \hbox{\hspace{0.0cm}\includegraphics[scale=0.65]{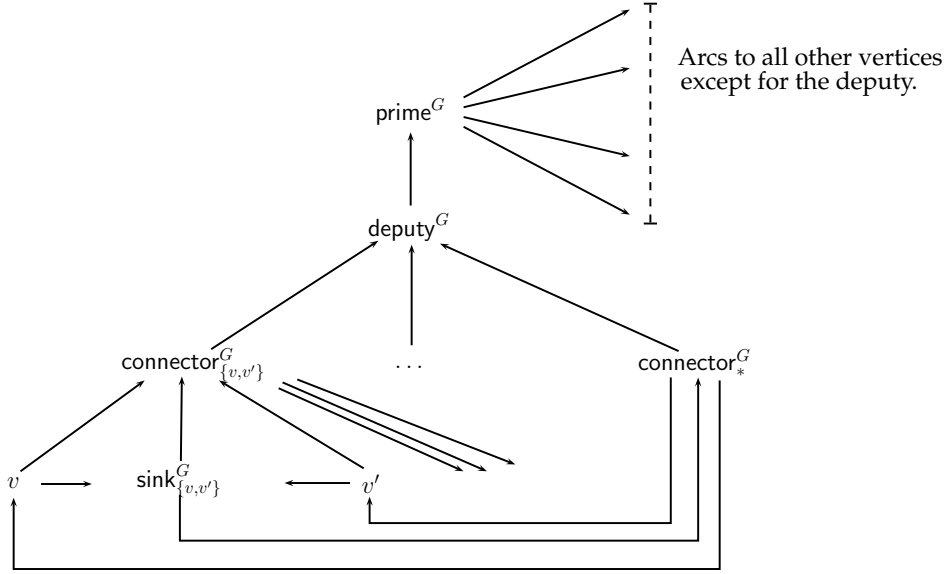}}
  \end{center}
  \caption{Construction of the graphs $G_I$ and $H_I$ in the proof of \cref{thm:uniform_istable_matching_as_hard_as_isomorphism}. The diagram shows the case $\{v, v'\} \in G$.}
  \label{fig:isomorphism_reduction}
\end{figure}

Using \cref{thm:uniform_prefs_isomorphism} and \cref{lemma:uniform_prefs_isomorphism_complete} as tools, we can show that for the preferences of the agents being uniform in all layers, the problem of finding an \aistable{$\alpha$} matching is at least as hard as the graph isomorphism problem.

\begin{theorem}\label{thm:uniform_istable_matching_as_hard_as_isomorphism}
  \textsc{Graph Isomorphism} is polynomial-time reducible to \ISM, where the preferences of the agents are uniform in all $\ell$ layers, $\ell$ is even, and $\alpha=\nicefrac{\ell}{2}$.
\end{theorem}
\begin{proof}
Let $G, H$ be two undirected graphs that form an instance of \textsc{Graph Isomorphism}. Without loss of generality, assume that
$n\geq 4$ (number of vertices) and $m\geq 3$ (number of edges).
% the number~$n$ of vertices and the number~$m$ of edges in $G$ and $H$ 
%are at least four and at least three, respectively.
%Since the graph isomorphism for regular graphs is GI-complete, and since the problem for the graphs with bounded degree is solvable in polynomial time, we can assume without loss of generality, that each vertex in both graphs is connected to at least three other vertices.
From $G$ and $H$ we construct an instance $I$ of the problem of deciding whether there exists an \aistable{$\alpha$} matching for preferences being uniform in all $\ell$ layers, where $\alpha=\nicefrac{\ell}{2}$.
By \cref{lemma:uniform_prefs_isomorphism_complete} we can describe this instance by providing the corresponding directed graphs $G_I$ and $H_I$. 
We show how to construct $G_I$ from $G$.
The graph~$H_I$ is constructed from $H$ analogously.

\newcommand{\prim}{\mathsf{prime}}
\newcommand{\deputy}{\mathsf{deputy}}
\newcommand{\connector}{\mathsf{connector}}
\newcommand{\sink}{\mathsf{sink}}
\newcommand{\source}{\mathsf{source}}
\newcommand{\Source}{\mathsf{SOURCE}}
\newcommand{\Sink}{\mathsf{SINK}}
\newcommand{\Connector}{\mathsf{CONNECTOR}}

Let $V$ and $E$ denote the sets of vertices and edges in $G$.
We copy all vertices from $G$ to $G_I$ (we will refer to these vertices as \myemph{non-special}) and we additionally introduce five types of ``special'' vertices, which we call $\prim^{G}$, $\deputy^G$, $\binom{n}{2}+1$ $\connector{s}$, $m$ $\sink{s}$, and $\binom{n}{2} - m$ $\source{s}$. 
We specify the connectors, the sinks, and the sources formally:
With each pair of vertices $p = \{v, v'\}\in \binom{V}{2}$ we associate one connector, 
denoted as $\connector^{G}_{p}$. 
One remaining connector is distinguished and denoted as $\connector^G_*$. 
Further, for each edge~$e \in E$ we have one sink, denoted as $\sink^{G}_e$, 
and for each non-edge $e = \{v, v'\} \notin E$ we have a source, denoted as~$\source^{G}_e$. 
The arcs in $G_I$ are constructed as follows.
Let $\Source^G=\{\source_{e} \mid e\in \binom{V}{2} \setminus E\}$, 
$\Sink^G=\{\sink_{e}\mid e\in E\}$, and $\Connector^G=\{\connector^G_e\mid e\in \binom{V}{2}\}$ denote 
the set of all sources,
the set of all sinks, 
and the set of non-special connectors, respectively:
%I rewrite the construction a bit as it is not clear or conventional what it means to connect a vertex with a vertex.
\begin{compactitem}[$\bullet$]
\item For each pair $e = \{v, v'\} \subseteq V$ of vertices, we do the following:
\begin{compactenum}
  \item If $\{v,v'\}\in E$, then we add to $G_I$ the following five arcs:
  $(v,\sink^{G}_e)$, $(v',\sink^{G}_e)$, $(v,\connector^{G}_e)$, $(v',\connector^{G}_e)$, and $(\sink^{G}_e, \connector^{G}_e)$. Otherwise, we add to $G_I$ the following five arcs:
  $(\source^{G}_e,v)$, $(\source^G_e,v')$,  $(v,\connector^{G}_e)$, $(v',\connector^{G}_e)$, and $(\source^{G}_e, \connector^{G}_e)$.
  \item For each other non-special vertex~$u\in V\setminus \{v,v'\}$, we add to $G_I$ the arc~$(\connector^{G}_e, u)$.
  \item For each other source or sink~$x\in \Source^G\cup \Sink^G\setminus \{\sink^G_e, \source^G_e\}$, we add to $G_I$ the arc~$(\connector^{G}_e, x)$.
\end{compactenum}
%\item For each pair of vertices $p = \{v, v'\}$ we add an arc from $v$, $v'$, and from $\mathit{sink}_p$ or $\mathit{source}_p$ (depending if $p$ is an edge or not) to $\mathit{connector}_{p}$. For each other non-special vertex and for each other source and sink $x$ we add an arc from $\mathit{connector}_{p}$ to $x$.
\item For each source and each sink $x\in \Source^G\cup \Sink^G$ we add to $G_I$ an arc~$(x, \connector^G_*)$. 
For each non-special vertex~$v\in V$, we additionally add to $G$ the arc~$(\connector^G_*, v)$.
\item For each vertex~$x\in V \cup \Source^G\cup \Sink^G\cup \Connector^G \cup \{\connector^G_*\}$ except for the deputy, 
we add to $G_I$ an arc $(\prim^G, x)$.
%We also add to $G_I$ the arc~$(\deputy^G, \prim^G)$. %from $\mathit{prime}$ to all other vertices but $\mathit{deputy}$, and there is an arc from $\mathit{deputy}$ to $\mathit{prime}$.
\item For each vertex~$x\in \Source^G\cup \Sink^G \cup \{\prim^G\}$ we add to $G_I$ the arc~$(\deputy^G, x)$.
For each connector~$y \in \Connector^G\cup \{\connector^G_*\}$ we add to $G_I$ the arc~$(y, \deputy^G)$.
%There is an arc from $\mathit{deputy}$ to all other vertices except for all the $\mathit{connectors}$. For each connector $x$ there is an arc from $x$ to $\mathit{deputy}$. 
\end{compactitem}
This completes the construction of $G_I$ (crucial elements of this construction are illustrated in \cref{fig:isomorphism_reduction}).
Observe a useful property: 
Each vertex in $G_I$ except $\prim^G$ has at least two incoming arcs; $\prim^G$ has exactly one incoming arc, namely from $\deputy^G$.
Indeed, since $G_I$ has at least three edges (and so at least three connectors), $\deputy^G$
has at least three incoming arcs from the connectors. 
Each connector $\connector^G_{\{v,v'\}}$ has incoming arcs from $v$ and $v'$; 
The special connector~$\connector^G_*$ has arcs from all sources and sinks (there are at least six of them). 
For each sink, source, and each non-special vertex there are at least two incoming arcs, one from the prime and the other from the deputy.    

The arcs in $H_I$ are constructed analogously.
We claim that there is an \aistable{$\alpha$} matching in the instance $(G_I, H_I)$ with $\alpha=\nicefrac{\ell}{2}$
if and only if there is an isomorphism between $G$ and~$H$. 

\medskip

($\Leftarrow$) For the ``if'' direction, assume that there is an isomorphism~$f\colon V(G)\to V(H)$ between $G$ to $H$.
We show that the following matching~$M$ is \aistable{$m$}. 
For each vertex~$v\in V(G)$ let $M(v)=f(v)\in V(H)$. 
Further, for each source or sink~$s^G_{\{v,v'\}} \in \Source^G\cup \Sink^G$ we let $M\left(s^G_{\{v, v'\}}\right)=s^H_{\{f(v), f(v')\}}$. 
Why should such a matching exist? If $s^G_{\{v, v'\}}$ is a sink, then $\{v, v'\} \in E(G)$; since $f$ is an isomorphism, we have that $\{f(v), f(v')\}\in E(H)$ and so the sink $s^H_{\{f(v), f(v')\}}$ exists.
Analogously,  if $s^G_{\{v, v'\}}$ is a source, then $s^H_{\{v,v'\}}$ is also a source and exists in $H$.
Similarly, we match the corresponding connectors: for each pair of vertices $\{v,v'\}$ we let
$M\left(\connector^G_{\{v, v'\}}\right)=\connector^H_{\{f(v), f(v'))\}}$,
and we let $M(\connector^G_*)=\connector^H_*$. 
Finally, we let $M(\deputy^G)=\deputy^H$ and $M(\prim^G)=\prim^H$.
Since $f$ is an isomorphism between $G$ and $H$, 
by the construction of $G_I$ and $H_I$, an arc~$(x,y)$ in $G_I$ (resp.~$H_I$) implies no arc $(M(x), M(y))$ in $H_I$ (resp.\ $G_I$). 
By \Cref{thm:uniform_prefs_isomorphism}, $M$ is \aistable{$m$}. 

\medskip

($\Rightarrow$) For the ``only if'' direction, assume that $M$ is an \aistable{$m$} matching for $(G_I, H_I)$. 
We start by showing $M(\prim^G)=\prim^H$. %that $\mathit{prime}$ in $G_I$ must be matched with $\mathit{prime}$ in $H_I$. 
For the sake of contradiction, suppose that this is not the case, and let $y = M(\prim^G) \in V(H_I)$, with $y \neq \prim^H$. 
Since $y\neq \prim^H$, there are at least two vertices~$w_1,w_2\in V(H_I)$ with $(w_1,y), (w_2,y)\in E(H_I)$. 
Thus, one of them is not matched to $\deputy^G$, say $M(w_1) \neq \deputy^G$.
Since $(w_1, y) \in E(H_I)$, by \cref{thm:uniform_prefs_isomorphism}, it must hold that $(\prim^G, M(w_1))=(M(y),M(w_1))\notin E(G_I)$.
This is a contradiction with $M(w_1) \neq \deputy^G$ since $\prim^{G}$ has outgoing arcs to all vertices but $\deputy^G$.

%Thus, by the pigeonhole principle there exists $v$ such that there is an arc from $\mathit{prime}$ to $v$ in $G_I$ and that there is an arc from $M(v)$ to $y$ in $H_I$, a contradiction with \Cref{thm:uniform_prefs_isomorphism}.
Second, we show that $M(\deputy^G)=\deputy^H$.
The reason for this is that using \cref{thm:uniform_prefs_isomorphism} for the arc~$(\deputy^G, \prim^G)$ implies that $(M(\prim^G), M(\deputy^G))=(\prim^H, M(\deputy^G))$ cannot exists in~$H_I$.
By construction, $\prim^H$ has an arc to all vertices except $\deputy^H$.
Thus, $M(\deputy^G)=\deputy^H$.
%$\mathit{deputy}$ in $G_I$ is matched with $y$ which is not a deputy in $H_I$. Then there is an arc from $\mathit{deputy}$ to $\mathit{prime}$ in $G_I$ and an arc from $M(\mathit{prime})$, which is the prime in $H_I$, to $y = M(\mathit{deputy})$ in $H_I$, a contradiction. 
By an analogous reasoning focusing on the deputy vertex, we deduce that each connector in $G_I$ must be matched with some connector in $H_I$. 
Further, for each source or sink~$s^{G}\in \Source^G\cup \Sink^G$ 
since $(s^G,\connector^G_*)\in E(G_I)$, by \cref{thm:uniform_prefs_isomorphism},
it follows that $(M(\connector^G_*),M(s^G))$ does not exist in~$H_I$; thus there are at least 6 forbidden arcs from $M(\connector^G_*)$ to some vertices which are sinks, sources, or non-special vertices.
This means that $\connector^G_*$ can only be matched to $\connector^H_*$ because each other connector in~$H_I$ has all but three outgoing arcs to such vertices.
%has more incoming arcs from the remaining vertices than other connectors,
%using a similar reasoning, we infer that the distinguished connectors $\mathit{connector}_*$ must be matched. 
Thus, we deduce that each source/sink must be matched with a source/sink only, and consequently, that non-special vertices in $G_I$ are matched with non-special vertices in $H_I$.
We show that the bijection~$f\colon V(G_I)\to V(H_I)$ derived from $M$ by setting $f(v)=M(v)$ for each $v \in V$ is an isomorphism between $G$ and $H$.

Consider an arbitrary pair $e = \{v,v'\}\subseteq V(G)$ of non-special vertices and its corresponding connector $\connector^G_e$. 
We know that $M(\connector^G_e)\in \Connector^H$, say $M(\connector^G_e)=\connector^H_{\{w,w'\}}$.
We claim that $\{w,w'\}=\{M(v), M(v')\}$.
Towards a contradiction, suppose that $M(v)\notin \{w, w'\}$.
Then, $(v, \connector^G_e)\in E(G_I)$ and $(\connector^G_{\{w,w'\}}, M(v))\in E(H_I)$, a contradiction to \cref{thm:uniform_prefs_isomorphism}.
Finally, we show that $e\in E(G)$ if and only if $\{f(v), f(v')\}=\{M(v), M(v')\}=\{w,w'\}\in E(H)$.
If $e\in E(G)$, then $\sink^{G}_e$ exists and $(\sink^{G}_e, \connector^G_e)\in E(G_I)$.
By \cref{thm:uniform_prefs_isomorphism},
it follows that $(\connector^{H}_{\{w,w'\}}, M(\sink^{G}_e))=(M(\connector^G_e), M(\sink^{G}_e))\notin E(H_I)$.
Since $M(\sink^G_e)\in \Source^H\cup \Sink^H$, we have that $M(\sink^G_e)=s^H_{\{w,w'\}}$, where $s$ is a source or a sink.
Further, since $(v, \sink^G_e) \in E(G_I)$ we infer that $(s^H_{\{w,w'\}}, w) \notin E(G_I)$, and so $s$ is a sink, and $w$ and $w'$ are connected in $H$.
Similarly, if $e\notin E(G)$, then we deduce that $\source^G_e$ and $\source^H_{\{w,w'\}}$ exist,
and thus $\{w,w'\}\notin E(H)$.
% By analyzing these two connectors, we get that a source/sink $s_p$ must be matched with the source/sink $s_{\{w, w'\}}$ (but, perhaps $s_p$ is a source and $s_{\{w, w'\}}$ is a sink, or vice versa). Similarly, we get that $v$ is matched to $w$ and $v'$ to $w'$ or that $v$ is matched with $w'$ and $v'$ with $w$. In any case, by looking at $s_p$ and  $s_{\{w, w'\}}$, we infer that both vertices need to be sources or that both need to be sinks. Thus, if $v, v'$ were connected in $G$ then $M(v), M(v')$ (which are $w$ and $w'$ or $w'$ and $w$, respectively) must have been connected in $H$. Consequently, $M$ truncated to the non-special vertices defines an isomorphism. This completes the proof. 
\end{proof}

\cref{thm:uniform_istable_matching_as_hard_as_isomorphism} implies that 
developing a polynomial-time algorithm for our problem is currently out of scope, since the question of whether \textsc{Graph Isomorphism} is solvable in polynomial time is still open.
Besides~\cref{thm:uniform_istable_matching_as_hard_as_isomorphism} there are other interesting implications of \cref{thm:uniform_prefs_isomorphism} and \cref{lemma:uniform_prefs_isomorphism_complete}.
For $\alpha \geq \nicefrac{\ell}{2}+1$ our problem can be reduced to the \textsc{Tournament Isomorphism} problem, which, given two tournament graphs, asks whether there is an arc-preserving bijection between the vertices of the two tournaments~\cite{BabEug83, Schweitzer17, Wagner07}.
 \textsc{Tournament Isomorphism} has been studied extensively in the literature (for a more detailed discussion see, e.g.,~\cite{BabEug83, Wagner07, Schweitzer17}), but to the best of our knowledge it is still open whether it is solvable in polynomial time~\cite{BabEug83}. 
The best known algorithm solving \textsc{Tournament Isomorphism} runs in $n^{O(\log n)}$ time, where $n$ denotes the number of vertices.

\begin{corollary}\label{cor:uniform-alpha>l/2-individual-subexp}
If the preferences of the agents are uniform in all $\ell$ layers and $\alpha \geq \nicefrac{\ell}{2}+1$,
then \ISM can be solved in $n^{O(\log n)} + O(\ell \cdot n^2)$ time, where $n$ denotes the number of agents.
 %then there exists an algorithm for finding \aistable{$\alpha$} matchings running in time $n^{O(\log n)} + O(\ell \cdot n^2)$, where $n$ denotes the number of agents.
\end{corollary}

\begin{proof}
  Assume that $\alpha \geq \nicefrac{\ell}{2}+1$ and construct the two directed graphs as we did for \cref{thm:uniform_prefs_isomorphism}.
  Since $\ell-\alpha+1\leq\nicefrac{\ell}{2}$, for each pair of vertices~$x$ and $y$ in $G_I$ (resp.\ $H_I$) at least one arc from $(x,y)$ and $(y,x)$ exists. 
  Note that such graphs are not necessarily tournaments, where for each two vertices~$x$ and~$y$ exactly one of $(x,y)$ and $(y,x)$ exists.
  But we can deal with the case when one of the graphs is not a tournament.
  If both $(x,y)$ and $(y,x)$ exist in $G_I$ or in $H_I$, then by \cref{thm:uniform_prefs_isomorphism} no \aistable{$\alpha$} matchings exists.
%  Further, if for each two vertices, $u, v \in G_I$ there exists an arc from $u$ to $v$ and an arc from $v$ to $u$, then by using \cref{thm:uniform_istable_matching_as_hard_as_isomorphism} we can infer that an \aistable{$\alpha$} does not exist. 
  Thus, the only non-trivial case is when the graphs $G_I$ and $H_I$ are tournaments.
  Moreover, in such a case, the condition from \cref{thm:uniform_prefs_isomorphism} can be reformulated as 
  $(x,y) \in G_I$ if and only if $(M(x),M(y)) \in H_I$ and this is the condition that $M$ is a tournament isomorphism between $G_I$ and~$H_I$. Consequently, for $\alpha \geq \nicefrac{\ell}{2}+1$ the problem of finding an \aistable{$\alpha$} matching can be reduced to \textsc{Tournament Isomorphism}.
  Note that the number of vertices in the constructed graphs equals the number of agents in our problem.
 % \textsc{Tournament Isomorphism} has been studied extensively in the literature (for a more detailed discussion see, e.g.,~\cite{BabEug83, Wagner07, Schweitzer17}), but to the best of our knowledge it is still not known if this problem is solvable in polynomial time. 
  By the result of Babai and Luks~\cite{BabEug83}, we obtain an algorithm for our problem with the desired running time.
\end{proof}

\subsubsection{Global stability}
%\subsubsection{\gstability{$\alpha$}}
There exists a fairly straightforward polynomial-time algorithm for finding \gstable{$\alpha$} matchings.

\begin{proposition}\label{prop:uniform-global-p}
If the preferences of the agents are uniform in all~$\ell$ layers, then 
\GSM can be solved in $O(\ell\cdot n)$~time, 
%there exists an algorithm for finding \gstable{$\alpha$} matchings running in $O(\ell \cdot n)$ time, 
where $n$ denotes the number of agents.
\end{proposition}
\begin{proof}
It is apparent that for uniform preferences of the agents, each layer admits a unique stable matching: for each $i \in [n]$ the $i$-th most preferred agent from~$U$ is matched with the $i$-th most preferred agent from~$W$. Further, such a matching can be computed in $O(n)$ time. Thus, our algorithm proceeds as follows. For each layer we compute a unique stable matching, and we pick the matching which is stable in the largest number of layers. Our algorithm returns this matching if it is stable in at least $\alpha$ layers; otherwise, the algorithm outputs that there exists no \gstable{$\alpha$} matchings for the given instance.     
\end{proof}

\section{Open Problems and Conclusions}

We have considered a new multi-layer model of preferences in the context of the \SM problem. We identified three natural concepts of stability and discussed their relations with each other. Our results show that the algorithmic problem of finding stable matchings according to each of the three concepts is, in general, computationally hard. On the positive side, we also managed to identify a number of natural special cases which are tractable (\Cref{table:results} summarizes our results).
Interestingly, while in the world of multi-layer stable matchings the 
case of two layers already leads to most computational hardness results,
in the world of maximum-cardinality matching in two-layer graphs 
one obtains polynomial-time solvability, while in the case of three layers one encounters NP-hardness~\cite{BKKMNS17}.

Our work provides a rich structure for analyzing computational properties of the problems we considered, and we view our work as only initiating this line of research. Indeed, it directly leads to the following open questions:
\begin{enumerate}[(1)]
\item How hard is it to find an \aistable{$\alpha$} matching for $\lceil \nicefrac{\ell}{2} \rceil < \alpha < \ell$?
\item When the preferences of the agents are uniform in each layer and $\alpha \geq \nicefrac{\ell}{2}+1$, we have shown that the decision variant of \ISM is solvable in quasi-polynomial time $n^{O(\log{n})} +O(\ell \cdot n^2)$) which implies that the problem is in the complexity class LOGSNP~\cite{PY96}. It would be interesting to know whether it is also complete for LOGSNP.\footnote{LOGSNP-hardness 
has been encountered and discussed for natural problems in 
Computational Social Choice~\cite{BCFGNW14,BCHKNSW14}.}
However, LOGSNP-hardness for our problem would also imply LOGSNP-hardness for \textsc{Graph Isomorphism}~(see~\cref{thm:uniform_istable_matching_as_hard_as_isomorphism}).
\item When the preferences of the agents are uniform in each layer,
 how hard is it to search for an \abstable{$\alpha$} matching for arbitrary $\alpha>1$ or an \aistable{$\alpha$} matching when $\alpha < \nicefrac{\ell}{2}$, or in general when the number of layers is constant.
\end{enumerate}   

We also believe that a number of other parameters and special cases can be motivated naturally in the context of our model, in particular parameters quantifying 
the degree to which the preferences of the agents differ.
Analogous parameterizations have been studied in computational social choice, for instance for the 
NP-hard \textsc{Kemeny Score} problem~\cite{BBN14,BGKN11}.
%We mention in passing that there is a growing body of research regarding the parameterized complexity of preference-based stable matching problems~\cite{GuRoSaZe2017,MarxSchlotter2010,MarxSchlotter2011,MnichSchlotter2017,MeeksRastegari2017}.

Continuing our research on special cases of input preferences (\cref{sec:constr-pref}), it might be interesting to study stable matching with multi-layer \emph{structured} preferences, such as \emph{single-peaked}~\cite{Black1958}, \emph{single-crossing}~\cite{Roberts1977,Mirrlees1971}, and \emph{1-Euclidean}~\cite{Coombs1964,Knoblauch2010,ChePruWoe2017} preferences.
We note that it can be detected in polynomial time whether a preference profile has any of these structure~\cite{BartholdiTrick1986,DoiFal1994,EscLanOez2008,DoiFal1994,ElkFalSli2012,BreCheWoe2013a,Knoblauch2010} and we refer the reader to \citet{BreCheWoe2016} and \citet{ElLaPe2017} for an overview of the literature on single-peakedness and single-crossingness.
We also note that \citet{BartholdiTrick1986} worked on stable roommates for narcissistic and single-peaked preferences, 
while \citet{BreCheFinNie2017} extended this line by also studying other structured preferences and including preferences with ties and incompleteness. 

Finally, our multi-modal view on the bipartite variant (\SM) can be generalized to the non-bipartite variant (\SR) and the case with incomplete preferences with ties.
It would be interesting to see whether our computational tractability results can transfer to these cases.

%\todo[inline]{RN: Shouldn't we use hyperref/pagebackref for the references. 
%At least the arXiv version would profit...}

\bibliographystyle{abbrvnat}
%\bibliography{bib} 

\end{document}